\documentclass{article}
\usepackage[utf8]{inputenc}
\usepackage{float}
\usepackage{caption}
\usepackage{subcaption}
\usepackage{bbm}
\usepackage{amssymb}
\usepackage{amsmath}
\usepackage{amsthm}
\usepackage{mathabx}
\usepackage{mathtools}
\usepackage{mathrsfs, enumerate}
\usepackage{appendix}
\usepackage[round]{natbib}
\usepackage[margin=1.2in]{geometry}
\usepackage[normalem]{ulem}
\usepackage[colorinlistoftodos,prependcaption]{todonotes}
\usepackage{algorithm}
\usepackage{algpseudocode}

\usepackage{xargs}
\newcommandx{\yxnote}[2][1=]{\todo[linecolor=blue,backgroundcolor=blue!25,bordercolor=blue,#1]{#2}}

\graphicspath{{./figs/}}

\theoremstyle{plain}
\newtheorem{thm}{Theorem}[section]
\newtheorem{prop}{Proposition}[section]
\newtheorem{cor}{Corollary}[section]
\newtheorem{lemma}{Lemma}[section]

\theoremstyle{definition}
\newtheorem{example}{Example}[section]
\newtheorem{assum}{Assumption}[section]
\newtheorem{cond}{Condition}[section]

\theoremstyle{remark}
\newtheorem{remark}{Remark}[section]

\DeclareMathAlphabet\EuRoman{U}{eur}{m}{n}
\SetMathAlphabet\EuRoman{bold}{U}{eur}{b}{n}

\newcommand{\cB}{\mathcal{B}}

\newcommand{\cE}{\mathcal{E}}

\newcommand{\cG}{\mathcal{G}}
\newcommand{\cH}{\mathcal{H}}
\newcommand{\cN}{\mathcal{N}}
\newcommand{\cP}{\mathcal{P}}

\newcommand{\cR}{\mathcal{R}}

\newcommand{\RR}{\mathbb{R}}
\newcommand{\EE}{\mathbb{E}}

\newcommand{\PP}{\mathbb{P}}

\newcommand{\one}{\mathbf{1}}
\newcommand{\ep}{\varepsilon}

\newcommand{\RSS}{\textnormal{RSS}}

\newcommand{\FDR}{\textnormal{FDR}}
\newcommand{\hFDR}{\widehat{\FDR}}
\newcommand{\FDP}{\textnormal{FDP}}
\newcommand{\hFDP}{\widehat{\FDP}}
\newcommand{\PFER}{\textnormal{PFER}}
\newcommand{\hPFER}{\widehat{\PFER}}

\newcommand{\FPR}{\textnormal{FPR}}
\newcommand{\se}{\textnormal{s.e.}}
\newcommand{\hse}{\widehat{\textnormal{s.e.}}}

\newcommand{\ltheta}{\widehat{\theta}^{\textnormal{Lasso}}}
\newcommand{\fstheta}{\widehat{\theta}^{\textnormal{FS}}}

\newcommand{\simiid}{\overset{\textnormal{i.i.d.}}{\sim}}

\newcommand{\sgn}{{\textnormal{sign}}}

\newcommand{\pth}[1]{\left( #1 \right)}
\newcommand{\br}[1]{\left[ #1 \right]}

\newcommand{\vct}[1]{\boldsymbol{#1}}
\newcommand{\mat}[1]{\boldsymbol{#1}}
\newcommand{\abs}[1]{\left| #1 \right|}
\newcommand{\norm}[1]{\left\| #1 \right\|}
\newcommand{\set}[1]{\left \{  #1 \right \}}

\newcommand{\Var}{\textnormal{Var}}

\newcommand{\mim}{\wedge}
\newcommand{\mam}{\vee}
\newcommand{\setcomp}{\mathsf{c}}

\newcommand*{\tran}{{\mkern-1.5mu\mathsf{T}}}
\newcommand{\independent}{\perp \!\!\! \perp}

\newcommand{\footremember}[2]{%
    \footnote{#2}
    \newcounter{#1}
    \setcounter{#1}{\value{footnote}}%
}

\DeclareMathOperator*{\argmin}{argmin}
\DeclareMathOperator*{\argmax}{argmax}

\title{Estimating the False Discovery Rate of Variable Selection}
\author{
Yixiang Luo\footremember{A}{Department of Mathematics, University of California, Berkeley, yixiangluo@berkeley.edu} \and
William Fithian\footremember{B}{Department of Statistics, University of California, Berkeley, wfithian@berkeley.edu} \and 
Lihua Lei\footremember{C}{Graduate School of Business, Stanford University, lihualei@stanford.edu}}
\date{}

\begin{document}

\maketitle

\begin{abstract}
We introduce a generic estimator for the false discovery rate of any model selection procedure, in common statistical modeling settings including the Gaussian linear model, Gaussian graphical model, and model-X setting. We prove that our method has a conservative (non-negative) bias in finite samples under standard statistical assumptions, and provide a bootstrap method for assessing its standard error. For methods like the Lasso, forward-stepwise regression, and the graphical Lasso, our estimator serves as a valuable companion to cross-validation, illuminating the tradeoff between prediction error and variable selection accuracy as a function of the model complexity parameter.
\end{abstract}

\section{Introduction}
\label{sec:introduction}

When selecting variables for inclusion in a statistical model, data analysts typically seek to balance multiple goals, including predictive accuracy, reliable identification of signal variables, model parsimony, and interpretability. Perhaps the best-known paradigm is to use a variable selection method such as the Lasso \citep{tibshirani1996regression}, forward stepwise regression \citep{hocking1976biometrics}, or the graphical Lasso \citep{friedman2008sparse},in conjunction with a tuning parameter that governs model complexity. The optimal value of this tuning parameter is commonly chosen by minimizing the cross-validation (CV) estimate of out-of-sample prediction error. Although this framework is flexible and effective, it inherently prioritizes prediction accuracy over other desirable properties. In particular, CV-based selection can be quite liberal in including irrelevant explanatory variables \citep[see, e.g.,][]{baumann2003cross}, which is undesirable in applications where accurate variable selection is of primary importance.

% One common paradigm for variable selection is to use variable-selecting procedures like Lasso \citep{tibshirani1996regression} or forward stepwise regression \citep{hocking1976biometrics}, tuned using cross-validation (CV) to minimize out-of-sample prediction error. While highly versatile, CV prioritizes predictive accuracy above all else. Because it can be quite liberal in including irrelevant explanatory variables, it may be inappropriate in applications where accurate variable selection is a priority \citep{baumann2003cross}.

In such applications, one may forgo estimation-based approaches like the Lasso in favor of multiple testing procedures, such as knockoffs \citep{barber2016knockoff,candes2018panning}, or the Benjamini--Hochberg procedure and its dependence-corrected variants \citep{benjamini1995controlling,benjamini2001control,fithian2022conditional}, which guarantee control of the false discovery rate (FDR) at a user-specified level. However, multiple testing methods do not yield fitted models, and the model formed by including only the selected variables may exhibit very poor fit to the data.\footnote{For example, in the knockoff procedure, certain correlation structures can cause a strong signal variable to have a negative knockoff statistic with approximately 50\% probability \citep{li2021whiteout}. If this occurs, the strong sigal variable will not be selected regardless of how large the FDR threshold is.}

In this work, we introduce a novel approach for assessing variable selection accuracy of estimation methods such as the Lasso: an estimator of the FDR applicable to any variable selection procedure under common statistical settings. This estimator can be evaluated in conjunction with the cross-validation curve, enabling analysts to examine the trade-off between variable selection accuracy and predictive performance across the entire regularization path. The proposed method is applicable in common statistical settings, including the Gaussian linear model, the Gaussian graphical model, and the nonparametric model-X framework of \citet{candes2018panning}, without imposing any additional assumptions beyond those required for standard hypothesis testing. The estimator exhibits non-negative finite-sample bias that is typically small, and its standard error can be assessed via the bootstrap.

The idea of estimating the FDR of a selection procedure dates back to \citet{storey2002direct}, who proposed a conservative estimator for variable selection based on thresholding independent $p$-values. \citet{goeman2011multiple} further developed an upper confidence bound for the number of null variables in any given selection set, yielding a conservative median estimator. Their method relies on a closed testing framework, and a feasible implementation typically requires independent or specific positively correlated $p$-values. In addition, certain multiple testing procedures, such as Benjamini–Hochberg  \citep{benjamini1995controlling} and knockoffs \citep{barber2016knockoff}, incorporate intrinsic False Discovery Proportion (FDP) estimators to achieve FDR control based on their particular variable ordering and testing strategy. In contrast, our proposed FDR estimator is broadly applicable across common statistical settings and can be used with any variable selection procedure deemed appropriate by the analyst.

% In the special case of thresholding independent $p$-values for multiple hypothesis tests, our method is almost identical to the FDR estimation method advocated by \citet{storey2002direct}, but slightly less conservative.

\begin{figure}[tbp]
    \centering
    \begin{subfigure}[b]{0.48\linewidth}
      \centering      \includegraphics[width=\linewidth]{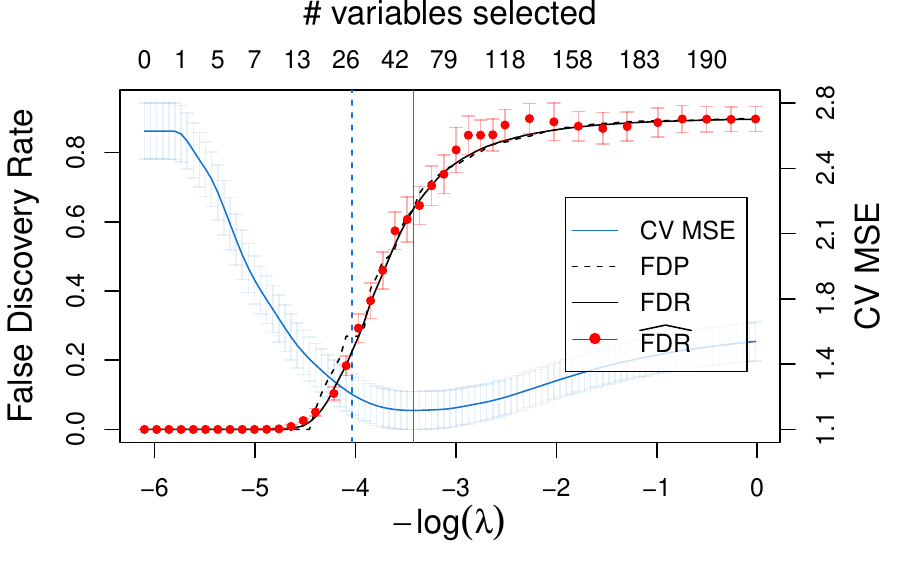}
      \caption{Easy variable selection scenario}
      \label{fig:illustration-indpt}
    \end{subfigure}
    \begin{subfigure}[b]{0.48\linewidth}
      \centering      \includegraphics[width=\linewidth]{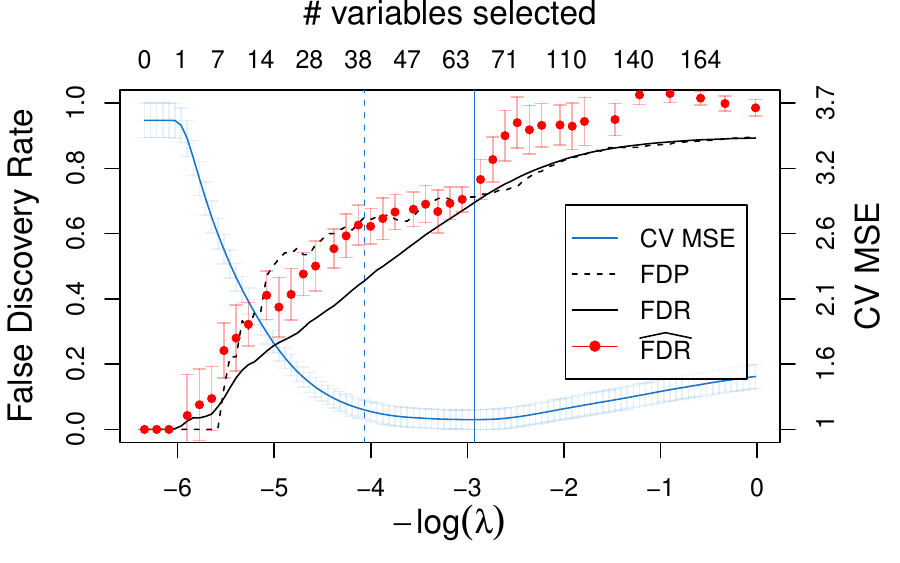}
      \caption{Difficult variable selection scenario}
      \label{fig:illustration-corr}
    \end{subfigure}
    \caption{Cross-validation MSE (blue), true FDR (black), and our FDR estimator (red), for Lasso regression in two scenarios with explanatory variables that are (a) independent, and (b) highly correlated. In scenario (a), the minimum-CV model (solid blue vertical line) has high FDR, while the one-standard-error rule (dashed vertical line) achieves a reasonably low FDR. In scenario (b), there is no model that simultaneously achieves good predictive performance and low FDR. In both scenarios, our FDR estimator successfully captures important information about variable selection performance that is not evident from the CV curve.}
    \label{fig:illustration-cv_FDR}
\end{figure}

Figure~\ref{fig:illustration-cv_FDR} illustrates the utility of our proposed method by comparing two Lasso regression scenarios that exhibit markedly different trade-offs between prediction error and FDR. In both scenarios, we generate $n = 600$ observations from a Gaussian linear model with $d = 200$ standard Gaussian explanatory variables, among which only $d_1 = 20$ are signal variables of equal moderate strength. The sole difference between the two settings lies in the dependence structure of the explanatory variables: in the first scenario, the variables are independent, whereas in the second, they exhibit an autoregressive $\text{AR}(1)$ correlation structure with a correlation coefficient $0.8$. This distinction substantially alters the variable selection performance. When the explanatory variables are independent, a model exhibiting both high predictive accuracy and reliable variable selection can be obtained using the one-standard-error rule, though not by minimizing the CV error. In contrast, when the variables are correlated, all models achieving competitive predictive accuracy exhibit an FDR exceeding $40\%$. While the CV curve alone offers no insight into variable selection performance in either case, our proposed estimator remains accurate across both scenarios. The following section introduces its construction.

\subsection{FDR estimation in the Gaussian linear model}\label{sec:linear_model}

Although our method applies in a more general setting, the special case of the Gaussian linear model with fixed design provides a concrete and familiar context in which to introduce our key ideas. In this model, we assume
\begin{equation}\label{eq:linear-model}
\vct Y = \sum_{j=1}^{d} \vct X_j \theta_j + \vct \ep = \mat X \vct \theta + \vct \ep, \quad \vct \ep \sim \cN(\vct 0, \sigma^2 \vct I_n),
\end{equation}
where $\vct X_j = (x_{1,j}, \ldots, x_{n,j})^\tran$ and $\vct Y = (y_{1}, \ldots, y_{n})^\tran$ denote the vectors of $n$ observed values of the explanatory variables $X_j$ and outcome variable $Y$, respectively, and $\mat X$ denotes the $n$-by-$d$ matrix of all observed explanatory variables. We call $X_j$ a {\em noise variable} if the hypothesis $H_j:\; \theta_j = 0$ is true and a {\em signal variable} otherwise. $\mat X$ is assumed to be fixed with $d<n$, while $\vct \theta \in \mathbb{R}^d$ and $\sigma^2 > 0$ are assumed to be unknown.

Any variable selection procedure yields an index set $\cR$, which we can regard as an estimator for the set of signal variables. For example, the Lasso estimator for a fixed regularization parameter $\lambda>0$ selects $\cR(\lambda) = \set{j:\; \widehat{\vct\theta}^{\textnormal{Lasso}}_j(\lambda) \neq 0},$ where
\begin{equation}
\label{eq:lasso}
    \vct \ltheta(\lambda) = \argmin_{\vct \theta} \; \frac{1}{2} \, \norm{\vct Y - \mat X \vct \theta}_2^2 + \lambda \norm{\vct \theta}_1.
\end{equation}

By contrast, forward stepwise regression (FS) selects variables by starting from the null model and adding one variable at a time, greedily minimizing at each step the next model's residual sum of squares. After $k = 1, 2, \ldots$ steps, the estimator $\vct \fstheta(k)$ is given by estimating an ordinary least squares (OLS) regression for the $k$ selected variables. 

\citet{benjamini1995controlling} define the \textit{false discovery rate} (FDR) as the expectation of the \textit{false discovery proportion} (FDP), defined here as the fraction of selected variables that are actually irrelevant. Let $\cH_0$ be the index set of all noise variables and $R = |\cR|$ be the total number of selections, we have
\begin{equation}\label{eq:FDR}
\FDR = \EE[\FDP] \quad\textnormal{and}\quad
\FDP = \frac{|\cR \cap \cH_0|}{R},
\end{equation}
using the convention that $0/0 = 0$; i.e., $\FDP = 0$ when no variables are selected.

Inspired by \citet{fithian2022conditional}, we linearly decompose the FDR as

% Our method employs a linear decomposition of FDR w

% is rather direct, and employs a linear decomposition of FDR as in \citet{fithian2022conditional}:
\begin{equation}\label{eq:linear-decomp}
\FDR = \sum_{j = 1}^d \FDR_j, \quad
\FDR_j = \EE \br{\frac{\one \set{j \in \cR}}{|\cR|}} \cdot \one \set{j \in \cH_0}.
\end{equation}
This technique allows us to treat each term separately and decouple the complicated interactions by conditioning. A promising estimator of the FDR is
\begin{equation} \label{eq:hFDR_linear}
\hFDR = \sum_{j = 1}^d \hFDR_j, \quad
\hFDR_j = \EE_{H_j} \br{\frac{\one \set{j \in \cR}}{|\cR|} \mid \vct S_j } \cdot \frac{\one \set{p_j > 0.1}}{0.9},
\end{equation}
where $p_j$ is the $p$-value for the usual two-sided $t$-test, and $\vct S_j = (\mat X_{-j}^\tran \vct Y, \, \norm{\vct Y}^2)$ is a sufficient statistic under the null model $H_j$. Under $H_j$, the first factor in $\hFDR_j$ is the uniform minimum-variance unbiased (UMVU) estimator of $\FDR_j$, obtained by Rao-Blackwellization, while the second factor attempts to zero out non-null terms.

% That is to say, the conditional distribution of $\vct Y$ given $\vct S_j$ and $H_j$ is known. Hence the conditional expectations in $\hFDR_j$ are computable.

Because $p_j$ and $\vct S_j$ are independent under $H_j$, the product $\hFDR_j$ is an unbiased estimator for its counterpart $\FDR_j$ whenever $j \in \cH_0$; otherwise $\hFDR_j \geq \FDR_j = 0$. As a direct consequence, we have
\begin{equation}\label{eq:conservative_bias}
\EE [\hFDR] \geq \FDR.
\end{equation}
The bias of $\EE [\hFDR]$ is small if most signal variables have $p_j \leq 0.1$ with high probability.
% Figure \ref{fig:hFDR_illustration} shows our $\hFDR$ is close to the true FDR in this illustrative example.

The definition of our estimator $\hFDR$ in \eqref{eq:hFDR_linear} is highly general and imposes no assumptions on the form of the variable selection procedure $\cR$. Consequently, it is applicable to any selection method, including the Lasso, forward stepwise regression, LARS \citep{efron2004least}, the elastic net \citep{zou2005regularization}, among many others.

Our approach requires no assumptions beyond those of the standard linear model in \eqref{eq:linear-model}, which are required to justify the use of marginal $t$-tests for each variable. In the common setting where the procedure $\cR$ depends only on $\mat X^\tran \mat X$ and $\mat X^\tran \vct Y$--a property shared by all methods listed above--the sole statistical assumption needed to justify the inequality \eqref{eq:conservative_bias} is that the $t$-statistic follows a $t_{n-d}$ distribution under $H_j$, conditional on $\vct S_j$. Notably, this assumption may still hold approximately even when the parametric assumptions of the linear model are violated.

% , and in many cases it relies only on the assumption that $t$-statistics are $t$-distributed. All of the methods named above share the property that they are functions of $\mat X^\tran \mat X$ and $\mat X^\tran \vct Y$ only. The first quantity is non-random by assumption, while the second can be recovered from $\vct S_j$ and the $t$-statistic for variable $j$. As a result, the only statistical assumption required to justify the inequality \eqref{eq:conservative_bias} is that the $t$-statistic follows a $t_{n-d}$ distribution under $H_j$, conditional on $\vct S_j$. This 

As we show next, our method extends far beyond the Gaussian linear model to other parametric and nonparametric settings, provided that the modeling assumptions yield an appropriate sufficient statistic $\vct S_j$ for each null hypothesis $H_j$. Section~\ref{sec:framework} describes our method in full generality and discusses additional applications, including the graphical Lasso \citep{friedman2008sparse} and the nonparametric model-X framework of \citet{candes2018panning}. Section~\ref{sec:uncertainty} presents a theoretical analysis of the standard error in stylized settings and describes our bootstrap procedure for estimating the standard error, which is used to construct the error bars in all figures. Sections~\ref{sec:real_data} and \ref{sec:simulations} illustrate the behavior of our estimator in real-data and simulation studies, respectively. Section~\ref{sec:compuation} addresses computational considerations, and Section~\ref{sec:discussion} concludes.

\section{General formulation of our estimator}
\label{sec:framework}

\subsection{Notation and statistical assumptions}\label{sec:notation}

Throughout the paper, we use uppercase letters to represent random variables and lowercase letters to denote individual observations. For example, $y_i$ denotes the $i$th observation of the variable $Y$, and $x_{i,j}$ denotes the $i$th observation of the variable $X_j$. We use bold uppercase letters to denote the collection of all $n$ observations. For instance, $\vct X_j = (x_{1,j}, \ldots, x_{n,j})^\tran$ represents the vector of observations for the random variable $X_j$, and $\mat X \in \RR^{n \times d}$, whose $(i,j)$th element is $x_{i,j}$, denotes the matrix of observations for all explanatory variables.
Moreover, let $[d] = {1, \ldots, d}$. For any subset $A \subseteq [d]$, we denote by $\vct X_A$ the submatrix of $\mat X$ containing only the columns corresponding to the variables in $A$. Similarly, we use $\vct X_{-j}$ as shorthand for $\vct X_{[d] \setminus {j}}$ and $\vct X_{-A}$ for $\vct X_{[d] \setminus A}$.

We denote by $\mat D$ a generic random object representing the entire dataset, whose distribution $P$ is assumed to belong to a known family $\cP$. If $\cP$ is a parametric family, we use $\vct \theta$ to denote its parameters and write $P_{\vct \theta}$ for the corresponding distribution. In the Gaussian linear model discussed in Section~\ref{sec:linear_model}, the irrelevance of explanatory variable $X_j$ is expressed through the null hypothesis $H_j: \vct \theta_j = 0$, while other problem settings may involve different forms of null hypotheses. Moreover, we denote by $\cR$ the index set of variables selected by a generic variable selection procedure based on the random dataset $\mat D$, and define $R = |\cR|$ as the total number of selected variables.

To define our estimator in a generic modeling context and prove that its bias is non-negative, the sole statistical assumption we will rely on is the availability of a sufficient statistic $\vct S_j$ for the submodel corresponding to each null hypothesis $H_j$. In one sense this assumption is vacuous, because $\vct D$ itself is a sufficient statistic for any model, but we will only obtain a useful estimator if $\vct S_j$ encodes less information than $\vct D$. While not strictly required, it is also convenient to assume that there is a standard choice of $p$-value $p_j$ for each $H_j$, which is valid conditional on $\vct S_j$.
% the Gaussian linear model and Gaussian graphical model settings, we will use $p$-value from $t$-tests, which are independent of $\vct S_j$, while in the model-X setting we assume $p_j$ arises from a conditional randomization test (CRT), which is independent of $\vct S_j$ by construction.

Section~\ref{sec:ex_models} discusses appropriate choices of $\vct S_j$ and $p_j$ in two additional modeling settings where our method can be applied: the nonparametric model-X setting, and the Gaussian graphical model.  In both examples, as in the Gaussian linear modeling example, $\vct S_j$ is a complete sufficient statistic for $H_j$ and there exists a natural choice of $p_j$ which is independent of $\vct S_j$ under $H_j$.

\subsection{FDR estimation}\label{sec:hFDR}

We now present a generic construction of our $\hFDR$, assuming only that a sufficient statistic $\vct S_j$ is available for each $H_j$. Recall the decomposition \eqref{eq:linear-decomp}:
\[
\FDR = \sum_{j = 1}^d \FDR_j, \quad
\text{where } \FDR_j = \EE \br{\frac{\one \set{j \in \cR}}{R}} \cdot \one \set{j \in \cH_0}.
\]
To estimate each variable's contribution $\FDR_j$, we proceed by estimating each of the two factors in turn.

% Since the numerator of the FDP can be decomposed as
% \[ |\cR \cap \cH_0| = \sum_{j = 1}^d \one\set{j \in \cR \cap \cH_0} = \sum_{j = 1}^d \one\set{j \in \cR} \cdot \one\set{j \in \cH_0}, \]
% we can accordingly decompose the FDR into a sum of contributions from each variable:
% \[
% \FDR = \sum_{j = 1}^d \FDR_j, \quad
% \text{where } \FDR_j = \EE \br{\frac{\one \set{j \in \cR}}{R}} \cdot \one \set{j \in \cH_0}.
% \]
% Since $\FDR_j \neq 0$ only if $j \in \cH_0$, we can equivalently write
% \[
% \FDR_j = \EE_{H_j} \br{\frac{\one \set{j \in \cR}}{R}} \cdot \one \set{j \in \cH_0},
% \]
% where the subscript $H_j$ emphasizes that the expectation is taken over a certain distribution in the subfamily described by $H_j$. If indeed $j \in \cH_0$, such distribution is the true distribution $P$.
% To estimate each $\FDR_j$ term, we proceed by estimating each of the two factors in turn.

Step 1: Estimating the first factor. Crucially, note that the factor $\EE[\one\{j\in\cR\}/R]$ is only relevant when $j\in \cH_0$. Therefore, when estimating it, we may assume that $H_j$ holds. Since $\vct S_j$ is sufficient for the submodel under $H_j$, we can estimate the first factor by Rao-Blackwellizing the integrand \citep{rao1945information,blackwell1947conditional}, obtaining
\begin{equation}\label{eq:hFDRstar}
    \hFDR^*_j(\vct S_j) := \EE_{H_j} \br{\frac{\one \set{j \in \cR}}{R} \mid \vct S_j }.
\end{equation}
For $j \in \cH_0$, $\hFDR_j^*$ is unbiased by construction, and it is the UMVU estimator for the first factor of $\FDR_j$ whenever $\vct S_j$ is complete \citep{lehmann1950completeness}. 
Since the conditional null distribution of the data $\vct D$ given $\vct S_j$ is known by the sufficiency, $\hFDR^*_j$ is computable and can generally be approximated using a Monte Carlo algorithm. For certain variable selection procedures, we have developed fast and exact computational methods. Further details are provided in Section~\ref{sec:compuation}.

Step 2: Estimating the second factor. We could obtain a conservative estimator by simply upper-bounding the indicator $\one\{j\in \cH_0\} \leq 1$, resulting in the estimator $\hFDR^* = \sum_j \hFDR_j^*$. However this estimator's bias, while still non-negative, could be quite large. In particular, if some of the signal variables are selected with high conditional probability given $\vct S_j$, their excess contributions may be large.

A better estimator of the indicator $\one\{j \in \cH_0\}$ should reliably eliminate terms for strong signal variables without introducing bias in the terms for noise variables. For any estimator $\psi_j(\vct D) \in [0,1]$, define its normalized version:
\begin{equation}\label{eq:phij}
\phi_j(\vct D) = \frac{\psi_j(\vct D)}{\EE_{H_j} [\psi_j(\vct D) \mid \vct S_j]},
\end{equation}
or $\phi_j(\vct D) = 1$ if $\EE_{H_j} [\psi_j | \vct S_j] = 0$. Note we have $\EE_{H_j}[\phi_j \mid \vct S_j] = 1$ almost surely.
For the same reason as in Step 1, the normalizer $\EE_{H_j} [\psi_j(\vct D) \mid \vct S_j]$ is always computable.

As an example, commonly, there is a standard test for $H_j$, whose $p$-value $p_j$ is uniform and independent of $\vct S_j$ under $H_j$. Then, we can take $\psi_j = \one\{p_j > \zeta\}$ for $\zeta \in (0,1)$, and
\begin{equation}\label{eq:phi_canonical}
\phi_j(p_j) = \frac{\one\{p_j > \zeta\}}{1-\zeta}.
\end{equation}
The parameter $\zeta$ governs the bias–variance trade-off: a larger $\zeta$ leads to a smaller bias but a larger variance. Our numerical experiments indicate that $\zeta = 0.1$ provides a good balance between bias and variance. When the signal is strong enough for the test to reject $H_j$ with high probability at level $\zeta$, the contribution of variable $j$ will almost always be zero. In Appendix~\ref{app:signal_bias}, we establish an upper bound on the bias contributed by a strong-signal variable in the Gaussian linear model, showing that it decays exponentially with the signal strength. We adopt \eqref{eq:phi_canonical} as the canonical form of $\phi_j$.

Combining \eqref{eq:hFDRstar} and \eqref{eq:phij}, we obtain the estimator
\begin{equation}\label{eq:hFDR_def}
\hFDR = \sum_{j=1}^d \hFDR_j, \quad \text{where } \hFDR_j(\vct D) = \hFDR_j^*(\vct S_j) \cdot \phi_j(\vct D).
\end{equation}

When we adopt the canonical $\phi_j$, we can write the canonical version of $\hFDR_j$ more concretely as
\begin{equation}\label{eq:hFDR_p_def}
\hFDR_j(\vct D) = \EE_{H_j} \br{\frac{\one \set{j \in \cR}}{R} \mid \vct S_j } \cdot \frac{\one\{p_j > \zeta\}}{1-\zeta}.
\end{equation}
These two factors are independent when $p_j \independent \vct S_j$.\footnote{When a suitable $p_j$ that is uniform and independent of $\vct S_j$ under $H_j$ is not available to adopt the canonical form, researchers may instead choose an appropriate $\psi_j$ function and employ the general formulation \eqref{eq:hFDR_def} of $\hFDR$.}
% \[
% \hFDR_j(\vct D) = \hFDR_j^*(\vct S_j) \cdot \phi_j(\vct D) = \EE_{H_j} \br{\frac{\one \set{j \in \cR}}{R} \mid \vct S_j } \cdot \frac{\psi_j}{\EE_{H_j} [\psi_j | \vct S_j]}
% \]

The estimator derived above, in its generic form \eqref{eq:hFDR_def}, is always conservatively biased:
\begin{thm}\label{thm:bias}
Suppose for each $j = 1, \ldots, d$, $\vct S_j$ is a sufficient statistic under the null submodel $H_j$. Then for any selection procedure $\cR$, and for any estimator $\psi_j:\; \vct D \to [0,1]$ of $\one\{j \in \cH_0\}$, the estimator $\hFDR$ defined in \textnormal{(\ref{eq:hFDRstar}, \ref{eq:phij}, \ref{eq:hFDR_def})} has non-negative bias:
\[ \EE[\hFDR] \geq \FDR. \]
\end{thm}
\begin{proof}
It suffices to show $\EE[\hFDR_j] \geq \FDR_j$ for all $j$. For $j \in \cH_0$, $\hFDR_j$ is unbiased:
\begin{align*}
    \EE[\hFDR_j] &= \EE \br{\EE_{H_j}\br{\hFDR_j^*(\vct S_j) \cdot \phi_j(\vct D) \mid \vct S_j}} \\[5pt]
    &= \EE \br{\hFDR_j^*(\vct S_j) \cdot \EE_{H_j}[ \phi_j(\vct D) \mid \vct S_j]}\\[5pt]
    &= \EE \br{\hFDR_j^*(\vct S_j)}\\
    &= \FDR_j.
\end{align*}
For $j \not\in \cH_0$, $\hFDR_j \geq \FDR_j = 0$ almost surely.
\end{proof}

\begin{remark}
    While we might have initially held out hope for an unbiased estimator, attaining one is typically impossible in practice. We can see this in the Gaussian linear model \eqref{eq:linear-model}, where the $\FDR$ of any nontrivial procedure $\cR$ is bounded and discontinuous in $\vct \theta$, but the distribution of the response $\vct Y$ depends on $\vct \theta$ in a continuous way. As a result, no statistic $T(\vct Y)$ can have $\EE_{\vct \theta}[T(\vct Y)] = \FDR_{\vct \theta}(\cR)$ for all $\vct \theta$.
\end{remark}
    % To understand why, consider the Gaussian linear model with a single variable that is always selected. Let $\theta_1 \in \mathbb{R}$ denote the single regression coefficient in the model, with $H_1:\; \theta_1=0$. Then the method has $\FDR = \PP\{\cR = \{1\}\}\cdot \one\{\theta_1 = 0\}$, which is discontinuous in $\theta_1$. Because the density of the response $\vct Y$ is continuous in $\theta_1$, the expectation of $\hFDR(\vct Y)$ must be continuous in $\theta_1$ as well. 

\begin{remark} 
    Extending the previous remark, we might view the smoothness of $\EE[\hFDR]$ with respect to $\vct \theta$, and the corresponding conservatism of $\hFDR$, as a robustness feature of our estimator, if we have conceptual concerns about the use of point nulls $H_j:\;\theta_j = 0$ in defining the $\FDR$. If many of the $\theta_j$ parameters are minuscule, but none are exactly zero, then the formal $\FDR$ of any procedure is exactly zero, but a more practical perspective would consider these ``near-nulls'' to be false discoveries whenever they are selected. Because the distribution of $\hFDR$ is continuous in $\vct \theta$, it effectively treats the near-nulls as nulls. That is, while it may be a severely biased estimator for the formal $\FDR$ (which is zero), it may be nearly unbiased for the ``practical $\FDR$'' that counts near-null selections as Type I errors.
\end{remark}

\begin{remark}
  As discussed in Section~\ref{sec:notation}, our sole statistical assumption --- availability of a sufficient statistic $\vct S_j$ for the submodel defined by $H_j$ --- is formally vacuous, because $\vct D$ is itself a sufficient statistic. However, if we take $\vct S_j = \vct D$, we have
  \[
  \hFDR_j(\vct D) = \EE_{H_j} \br{\frac{\one \set{j \in \cR}}{R} \mid \vct D } \cdot \frac{\psi_j(\vct D)}{\EE_{H_j} [\psi_j(\vct D) \mid \vct D]} = \frac{\one \set{j \in \cR}}{R},
  \]
  and consequently $\hFDR = 1$ almost surely. Theorem~\ref{thm:bias} still applies to this estimator, which trivially has non-negative bias for the true FDR, but the estimator is useless.
\end{remark}

\begin{remark}
  For variable selection by thresholding independent p-values $p_j$ at threshold $c \in (0, 1)$, \citet{storey2002direct} proposed a conservatively biased FDR estimator
  \[
  \hFDR^{\textnormal{Storey}} = \frac{c}{R(c)} \cdot \frac{\#\{p_j > \zeta\}}{1-\zeta}, \quad
  \textnormal{ where } \cR(c) = \set{j: \; p_j \leq c}.
  \]
  To apply our approach, let $\vct S_j = p_{-j}$ and assume $c \leq \zeta$. Then we have
  \[
  \hFDR = \sum_{j=1}^d \EE_{H_j} \br{\frac{\one \set{p_j \leq c}}{\#\set{j: \; p_j \leq c}} \mid p_{-j} } \cdot \frac{\one\{p_j > \zeta\}}{1-\zeta} = \frac{c}{R(c) + 1} \cdot \frac{\#\{p_j > \zeta\}}{1-\zeta}.
  \]
%   \begin{align*}
%   \hFDR(c) &= \sum_{j=1}^d \EE_{H_j} \br{\frac{\one \set{p_j \leq c}}{\#\set{j: \; p_j \leq c}} \mid p_{-j} } \cdot \frac{\one\{p_j > \zeta\}}{1-\zeta}\\[5pt]
%    &= \sum_{j=1}^d \frac{c}{R + \one \set{p_j > c}} \cdot \frac{\one\{p_j > \zeta\}}{1-\zeta}\\[5pt]
%    &\geq \frac{c}{R} \cdot \frac{\#\{p_j > \zeta\}}{1-\zeta} = \hFDR^{\textnormal{Storey}}(c)
%   \end{align*}
%   Without the $\one \set{p_j > \lambda}$ terms, 
Thus, our $\hFDR$ is nearly identical to Storey's estimator but slightly less conservative.
\end{remark}

%     Readers may expect an unbiased estimator for FDR at first sight. But it is infeasible. Consider an almost unidentifiable Gaussian linear model \eqref{eq:linear-model} where we have two variables $X_1$ and $X_2$ arbitrarily close to each other. A simple selection procedure selects $X_1$ always. Then the FDR is $0$ in cases 1 where $\beta_1 = 1$ and $\beta_2 = 0$ but $\FDR=1$ in case 2 where $\beta_1 = 0$ and $\beta_2 = 1$. By construction, the distribution of $\mat D$ in case 1 and case 2 are arbitrarily close to each other. A reasonably smooth FDR estimator based on $\mat D$ is not able to tell whether FDR is $0$ or $1$ and hence cannot be unbiased. Our conservatively biased estimator $\hFDR$ targets the largest possible FDR and has an expectation $1$. Moreover, we prefer overestimation to underestimation since we care more about the upper bound of FDR: when an overestimation is small, we have better confidence in the selection set.
% \end{remark}

The next section elaborates on our method's definition to help readers build intuition for how it works.

\subsection{Understanding our method}\label{sec:understanding}

% To understand our estimator, it is convenient to rewrite its canonical form as
% \[
% \hFDR 
% = \frac{1}{1-\zeta}\sum_{j\in \cU} \EE_{H_j} \br{\frac{\one \set{j \in \cR}}{R} \mid \vct S_j }, \quad \text{ where } \;\;\cU(\vct D) = \{j:\; p_j > \zeta\}
% \]
% is the index set of variables with ``unimpressive'' $p$-values. In words, we sum the estimated FDR contributions for only the variables that appear likely to be null based on their marginal $p$-values, then we multiply by a small inflation factor $(1-\zeta)^{-1}$, to adjust for the null variables whose $p$-values were small by random chance.

% When $R$ is relatively large and stable, we can make the helpful approximation
% \[
% \hFDR \approx \frac{1}{(1-\zeta)R} \sum_{j \in \cU} \PP_{H_j} \br{j \in \cR \mid \vct S_j}.
% \]

To better understand our method, it is helpful to consider a simpler but closely related method that we could use to estimate a different Type I error rate, called the \textit{per-family error rate} (PFER)
\begin{align*}
\PFER 
&= \EE\br{|\cR \cap \cH_0|} \\
&= \sum_{j=1}^d \PP\br{j \in \cR} \cdot \one \set{j \in \cH_0}\\
&= \sum_{j \in \cH_0} \PP\br{j \in \cR}.
\end{align*}
We can follow the same approach as in Section~\ref{sec:hFDR} to obtain the PFER estimator
\begin{align}\nonumber
\hPFER 
&= \sum_{j=1}^d \PP_{H_j}\br{j \in \cR \mid \vct S_j} \cdot \frac{\one\{p_j > \zeta\}}{1-\zeta}\\[5pt]
\label{eq:hPFER}
&= \frac{1}{1-\zeta}\cdot \sum_{j:\, p_j > \zeta} \PP_{H_j}\br{j \in \cR \mid \vct S_j} .
\end{align}
The $j$th term in the sum \eqref{eq:hPFER} is an estimate of variable $j$'s chance of being selected, calculated under the assumption that $H_j$ holds. The estimator $\hPFER$ adds up these estimated selection probabilities over all variables with unimpressive $p$-values, then multiplies the sum by an inflation factor $(1-\zeta)^{-1}$ to account for the null terms that were eliminated by chance. If $\zeta = 0.1$, the inflation factor is $10/9 \approx 1.1$.

In problems where the number $R$ of selections is relatively large and stable, such that $1/R$ is roughly constant conditional on $\vct S_j$, we have
\[
\EE_{H_j} \br{\frac{\one \set{j \in \cR}}{R} \mid \vct S_j } \approx \frac{\PP_{H_j}\br{j \in \cR \mid \vct S_j}}{R},
\]
and as a result
\begin{equation}\label{eq:hFDR_approx}
\hFDR \approx \frac{\hPFER}{R} = \frac{1}{(1-\zeta)R}\cdot \sum_{j:\, p_j > \zeta} \PP_{H_j}\br{j \in \cR \mid \vct S_j}.
\end{equation}
For variable selection methods where $R$ is deterministic, including forward stepwise regression with a fixed number of steps, this approximation is exact.

This approximation aids our intuition by telling us that $\hFDR$ is well-approximated by a sum of estimated inclusion probabilities, divided by the number $R$ of selected variables, and inflated by $(1-\zeta)^{-1}$. Roughly speaking, $\hFDR$ is large when some unimpressive variables stood a good chance of being selected (or actually were selected), and otherwise $\hFDR$ is small.

\begin{remark} The PFER is interesting in its own right as an upper bound for the {\em family-wise error rate} (FWER), defined as 
\[
\textnormal{FWER} = \PP\br{|\cR \cap \cH_0| \geq 1} \leq \PFER.
\]
As a result, the estimator $\hPFER$ in \eqref{eq:hPFER} can also be understood as a conservatively biased estimator for FWER. We leave exploration of this extension to future work.
\end{remark}

\subsection{Example model assumptions}\label{sec:ex_models}

In this section we give additional examples beyond the Gaussian linear model and demonstrate the choices of $\vct S_j$ and the ways to compute a valid $p$-value.

\begin{example}[Nonparametric ``model-X'' setting]\label{ex:modelX}
\citet{candes2018panning} introduced novel nonparametric modeling assumptions for multiple testing in supervised learning problems, which allow the analyst to test for conditional independence between the response and each explanatory variable, given the other variables. Assume that we observe $n$ independent realizations of the $(d+1)$-variate random vector $(X,Y) \sim P$, where the joint distribution $P_X$ of explanatory variables is fully known, but nothing is assumed about the conditional distribution $P_{Y \mid X}$.\footnote{These assumptions can be relaxed to allow for $P_X$ to be known up to some unknown parameters; see \citet{huang2020relaxing}.}

Under these modeling assumptions, $\vct S_j = (\vct X_{-j}, \vct Y)$ is a complete sufficient statistic for the $j$th null hypothesis $H_j:\; Y \independent X_j \mid X_{-j}$. Under $H_j$, after fixing the observed $\vct S_j$ we can sample from the conditional distribution of $\vct D$ by sampling a new $\vct X_j$ vector according to its known distribution given $\vct X_{-j}$. For any test statistic $T_j(\vct D)$, we can compare the observed value to its simulated distribution to obtain a $p$-value that is independent of $\vct S_j$; the corresponding test is called the conditional randomization test (CRT). In our simulations, we use the marginal covariance $\textnormal{cov}(\vct X_j, \vct Y)$ as our CRT test statistic and perform a two-sided test, but other choices may be preferred depending on the problem setting and the manner in which the analyst expects $Y$ to depend on the signal variables.
\end{example}

% Assume the outcome $Y$ and explanatory variables $X_1, \ldots, X_d$ are drawn independently from the distribution $P$,
% \[ (y_i, x_{i,1}, \ldots, x_{i,d}) \simiid P, \quad i = 1, \ldots, n. \]
% We make no model assumption on the conditional distribution $Y \mid x_{[d]}$ but assume the joint distribution $P_x$ of the explanatory variables $x_{[d]}$ is known. This is the model-X assumption first proposed in \citet{candes2018panning}. Usually, the covariates are assumed to be multivariate Gaussian with the covariance estimated from the data. More sophisticated ways to learn a $P_x$ from the data can be found in [cite]. We call variable $X_j$ a null if $Y$ is conditionally independent of $X_j$ given the other variables. Formally, $H_j:\, Y \independent X_j \mid X_{-j}$.

% \begin{prop}
% With the model-X assumption, let $\vct S_j = (\vct X_{-j}, \vct Y)$. Under $H_{j}$, the conditional distribution $\vct X_j \mid \vct S_j \eqd \vct X_j \mid \vct X_{-j}$, where $\vct X_j \mid \vct X_{-j}$ is derived from the known $P_x$.
% \end{prop}
% The conditional randomization test (CRT) is shown powerful in the model-X settings \citep{candes2018panning}. We recommend setting $\psi_j = \one\set{p_j \geq 0.1}$, where $p_j$ is CRT p-value with a simple test statistics $|\vct X_j^\tran \vct Y|$ to make the computation fast. Again $\vct S_j \independent p_j$ under $H_j$ and therefore $\EE_{H_j} [\psi_j \mid \vct S_j] = 0.9$.

\begin{example}[Gaussian graphical model] \label{ex:gauss_grapic}
Assume we observe $n$ independent realizations of a $d$-variate Gaussian random vector $X$ ($n > d$), with mean zero and positive-definite covariance matrix $\mat \Sigma$. In the corresponding precision matrix $\mat \Theta = \mat \Sigma^{-1}$, the entry $\mat\Theta_{jk}$ is zero if and only if $X_j$ and $X_k$ are conditionally independent given the other entries. In other words, the sparsity pattern of the undirected dependence graph $\cG$ for $X$ is identical to the sparsity pattern of $\mat \Theta$. In lieu of estimating all $d(d-1)/2$ off-diagonal entries of $\mat \Theta$, a common modeling assumption is that the dependence graph is sparse, and by extension that $\mat \Theta$ is sparse as well \citep{meinshausen2006high,yuan2007model,friedman2008sparse}.

We associate to each pair $j\neq k$ the null hypothesis that $(j,k)$ is absent from the true dependence graph $\cG$, or equivalently $H_{jk}:\;\mat \Theta_{jk} = 0$. The Gaussian graphical model is a parametric exponential family model whose natural parameters are the entries of $\mat\Theta$, with corresponding sufficient statistics given by the entries of $\widehat{\mat \Sigma} = {\mat X}^\tran {\mat X}$, which follows a Wishart distribution. As a result, a complete sufficient statistic for $H_{jk}$ is given by observing every entry of $\widehat{\mat \Sigma}$ {\em except} $\widehat{\mat \Sigma}_{jk}$ and $\widehat{\mat \Sigma}_{kj}$, or equivalently by 
\[
\vct S_{jk} = \left(\mat X_{-k}^\tran \mat X_{-k}, \;\mat X_{-\{j,k\}}^\tran \mat X_k, \;\|\mat X_k\|^2\right).
\]
This sufficient statistic is strongly reminiscent of the sufficient statistic for $H_j$ in the Gaussian linear model, with $\mat X_{-k}$ playing the role of the design matrix and $\mat X_{k}$ playing the role of the response, and the standard test for $H_{jk}$ is indeed identical to the $t$-test for $H_j$ in that regression problem. In our simulations, we use the corresponding two-sided $t$-test $p$-value $p_{jk}$. See \citet{drton2007multiple} for additional details.

In this modeling context, we will study the FDR of the graphical Lasso algorithm of \citet{friedman2008sparse}, which is obtained by maximizing the log-likelihood with a Lasso penalty on the off-diagonal entries of $\mat \Theta$:
\begin{equation}\label{eq:graph_lasso}
    \mat {\widehat{\Theta}}(\lambda) = \argmax_{\mat \Theta \succ \mat 0} \; \log\det\mat\Theta - \textnormal{trace}(\mat {\widehat\Sigma} \mat \Theta) - \lambda \sum_{j<k} |\Theta_{jk}|.
\end{equation}
\end{example}

\section{Standard error of $\hFDR$}
\label{sec:uncertainty}

In this section, we discuss the standard error of our estimator $\hFDR$. Since our general method can be adapted for many different selection algorithms in many different statistical modeling contexts, it is very difficult to provide a general analysis that covers all contexts in which our method might be applied, but we can offer theoretical results in a narrower context, as well as a practical estimation strategy that we have found to be reasonably reliable in our simulation studies. Section~\ref{sec:se_bound} develops a theoretical finite-sample bound for our estimator under stylized assumptions, leading to conditions under which the variance of $\hFDR$ is vanishing. Section~\ref{sec:se_est} discusses a specialized bootstrap method for standard error estimation, which we use to produce all figures in this work. Section~\ref{sec:bootstrap_consistency} introduces the theoretical property of our bootstrap scheme.

\subsection{Theoretical bound for standard error} \label{sec:se_bound}

Intuitively, if $R$ is reasonably large and many different variables independently have comparable chances of being selected, our estimator's standard error should be small. To illustrate this heuristic in a stylized setting, we develop a finite-sample bound for the standard error, $\se(\hFDR)$, for the Lasso in the Gaussian linear model, under a block-orthogonality assumption. We make no assumptions about the variables' dependence within each block of variables, but we assume that variables in different blocks are orthogonal to each other:
\begin{assum}[block-orthogonal variables]\label{ass:block-orth}
Let $\mat X$ be an $n\times d$ matrix, and let index set $(b) \subseteq [d]$ for $b = 1, \ldots, B$ represent a partition of $[d]$ into $B$ disjoint subsets (that is, $\bigcup_b (b) = [d]$). For each $b$, let $\vct X_{(b)} \in \RR^{n \times m_b}$ be the submatrix of $\mat X$ containing only the columns corresponding to the variables in $(b)$, where $m_b = |(b)|$ is size of block $(b)$. We say that $\mat X$ has \emph{block orthogonal variables} with respect to the partition $\{(b)\}_{b=1}^B$ if
\[
\vct X_{(b_1)}^\tran \vct X_{(b_2)} = \mat 0_{m_{b_1} \times m_{b_2}},
\quad \text{for all } b_1 \neq b_2.
\] 
\end{assum}

Assumption~\ref{ass:block-orth} allows us to treat the Lasso on the entire matrix $\mat X$ as a separable optimization problem on the $B$ different blocks of explanatory variables. When this assumption is in effect, we will denote the size of the largest block as $m = \max_b m_b$.
We will assume throughout this section that every explanatory variable has mean zero and unit norm: $\sum_{i=1}^n x_{i,j} = 0$ and $\|\vct X_j\|^2 = 1$, for all $j = 1,\ldots, d$. This last assumption comes with little practical loss of generality since conventional implementations of the Lasso begin by normalizing all variables in this way.

For our results below, it will be important to be able to provide a probabilistic lower bound $q$ on the average number of selections per block:
\begin{cond}[enough selections] \label{cond:adequate_sels}
We have $\PP(R \leq B q) \leq \rho$ for some $q,\rho \geq 0$.
\end{cond}
Because Condition~\ref{cond:adequate_sels} always holds for {\em some} $q$ and $\rho$, it has force only after we show that it holds for specific values of $q$ and $\rho$, defined in terms of other problem quantities. Proposition~\ref{prop:signals-lasso} gives examples of such bounds for Lasso problems.

\begin{prop} \label{prop:signals-lasso}
In the Gaussian linear model with block-orthogonal variables (Assumption~\ref{ass:block-orth}), the following statements hold for the Lasso with tuning parameter $\lambda > 0$:
\begin{enumerate}
    \item Let $\delta(z) := 2 (1 - \Phi(z))$
    % $\delta(\lambda; \sigma) := 2 (1 - \Phi(\lambda / \sigma))$
    denote the probability that a Gaussian random variable with mean $0$ and variance 1 exceeds $z$ in absolute value. Then $R$ is stochastically larger than a $\textnormal{Binom}(B,\delta(\lambda/\sigma))$ random variable, and for any $q < \delta(\lambda/\sigma)$, the Lasso procedure satisfies Condition~\ref{cond:adequate_sels} with
    \[
    \rho = \exp \pth{ -\;\frac{(\delta(\lambda/\sigma) - q)^2}{2 \delta} \cdot B}.
    \]
    \item Assume no fewer than $B \xi$ blocks contain at least one variable with $|\vct X_j \cdot \vct \mu| \geq \lambda$ for some $\xi \in (0, 1]$, where $\vct \mu = \mat X \vct\theta$ is the mean of $\vct Y$. Then $R$ is stochastically larger than a $\textnormal{Binom}(B \xi, 1/2)$ random variable, and for any $q < \xi/2$, the Lasso procedure satisfies Condition~\ref{cond:adequate_sels} with
    \[
    \rho = \exp \pth{-\; \frac{(\xi - 2q)^2}{4 \xi} \cdot B}.
    \]
    \end{enumerate}
\end{prop}

A finite-sample upper bound for $\text{Var}(\hFDR)$ for Lasso with block-orthogonal variables then follows.

\begin{thm}[finite sample variance bound] \label{thm:variance}
In the Gaussian linear model with block-orthogonal variables (Assumption \ref{ass:block-orth}), suppose we 
apply Lasso for variable selection at $\lambda>0$, and in the form of $\hFDR$, $\phi_j = \phi(p_j)$ only depends on $p_j$, the two-sided t-test p-value for $H_j$, where $\phi$ is a $L_\phi$-Lipschitz continuous function upper bounded by $c_\phi \geq 1$. 

If Condition~\ref{cond:adequate_sels} holds for given values of $q$ and $\rho$, we have
\[ \Var(\hFDR)
= 
O \pth{\pth{\frac{ m^4}{q^4} + \frac{m^6}{B q^4} } \cdot \frac 1B + \frac{m^2}{q^2} \cdot \frac{1}{n-d} + m^6 B^2 (\rho \mam e^{-(n-d)/32})},
\]
where the big-O notation only hides constants that depend on $c_\phi$, $L_\phi$, and $\sigma^2$.
\end{thm}

\begin{remark}
    An explicit formula of the variance bound that contains $c_\phi$, $L_\phi$, and $\sigma^2$ is provided in Appendix~\ref{app:proofs}.
\end{remark}

We briefly introduce the core idea of the proof which relies on the concentration laws. Interested readers may refer to Appendix~\ref{app:proofs} for the full proof.
Assumption~\ref{ass:block-orth} has three key consequences:
\begin{enumerate}
    \item Conditioning on $\vct S_j$ fixes $\vct X_k^\tran \vct Y$ for every $k$ that is not in the same block as $j$. As a result, resampling the data conditional on $\vct S_j$ would not change the selection in all other blocks. Hence, the number of selections is conditionally stable: with resampling conditional on $\vct S_j$ it is $R(\vct D) + O(m)$.
    \item For any $j \in (b)$, let $\mat X_{(b); -j}$ denote the variables in block $(b)$ except $\vct X_j$. The selection of variable $\vct X_j$ depends on $\vct X_{-j}^\tran \vct Y$ only through $\vct X_{(b); -j}^\tran \vct Y$. Formally, we have
    \[
    E_j
    := \EE_{H_j}[\one\set{j \in \cR} \mid \vct X_{(b); -j}^\tran \vct Y]
    = \EE_{H_j}[\one\set{j \in \cR} \mid \vct X_{-j}^\tran \vct Y]
    \approx \EE_{H_j}[\one\set{j \in \cR} \mid \vct S_j].
    \]
    Note these $E_j$ are block-independent.\footnote{Alert readers might have noticed that conditioning on $\vct S_j$ is not the same as conditioning on $\vct X_{-j}^\tran \vct Y$: there is an additional $\norm{\vct Y}$ term. The analysis of the effects of $\norm{\vct Y}$ is rather technical and it is the source of the $O(1/(n-d))$ term in the variance bound.}
    \item Likewise, the p-value $p_j$ depends on $\vct X_{-j}^\tran \vct Y$ only through $\mat X_{(b); -j}^\tran \vct Y$. So $\phi_j(p_j)$ are block-independent.
\end{enumerate}
As a result, we have
\[
\hFDR_j \;=\; \EE_{H_j} \br{\frac{\one \set{j \in \cR}}{R} \mid \vct S_j } \cdot \phi_j
\;\approx\; \frac{E_j \cdot \phi_j}{R + O(m)}
\;\approx\; \frac{E_j \cdot \phi_j}{R}.
\]
The last approximation holds when $R > Bq$ with high probability and $B$ is substantially larger than $m$. Then
\[
\hFDR \;=\; \sum_j \hFDR_j
\;\approx\; \frac{\sum_k \pth{\sum_{j \in (b)} E_j \cdot \phi_j}}{R}.
\]
This is an average of 
$B$ independent random variables. Therefore, the variance is of order $O(1/B)$ by the concentration laws. A formal asymptotic result is shown in Corollary~\ref{cor:variance}. 

\begin{cor}[Asymptotically vanishing variance] \label{cor:variance}
Under the same settings as in Theorem~\ref{thm:variance}, consider a series of problems where we hold $c_\phi$, $m$, $L_\phi$ and $\sigma$ fixed but allow $\vct \theta$, $\lambda$, and $n$ to vary as $B$ increases. Suppose that at least $s(B)$ blocks contain a variable with $|\vct X_j \cdot \vct \mu| \geq \lambda$, where $\vct \mu = \mat X \vct\theta$ is the mean of $\vct Y$.

\begin{enumerate}
\item For $0 < \ep < 1/4$, suppose $s(B) \geq B^{1-4\ep}$, and let $B \to \infty$ with $n - d \geq B^{1 - 2\ep}$. Then
\[
\Var(\hFDR) = O(B^{-1 + 4 \ep }) \to 0.
\]
\item For $0 < \xi < 1$, suppose $s(B) \geq B\xi$, and let $B \to \infty$ with $n - d\geq 64\log(B)$. Then
\[
\Var(\hFDR) = O(B^{-1}) \to 0.
\]
\end{enumerate}
\end{cor}

The idea of bounding $\se(\hFDR)$ by the concentration laws applies to more general settings. The key properties it requires are that 1) the selection of one variable is merely affected by a relatively small group of variables and 2) the perturbation on one variable merely affects the selection of a relatively small group of variables. For example, with the block-orthogonal variables, such a group is the block where the variable belongs to. It allows us to decouple $R$ with an individual selection and treat it as roughly constant. Then $\hFDR = \sum_j \hFDR_j$ is a sum of variables with only local dependency. The independent parts of randomness in $\hFDR_j$ cancel each other and yield an $\hFDR$ with a small variance. However, if perturbing a single variable greatly impacts the entire selection set, $\hFDR$ can be noisy. We describe an example where $\hFDR$ has a non-vanishing standard error as $d$ increases in Appendix~\ref{app:asymp_noisy}.

% fixed-X gaussian linear model, lasso, theorem

% general theorem

% discuss if condition is necessary

% \begin{assum}[local permutation]
% $j$ only affect the selection decision on a subset of other variables with size $o(p)$
% \end{assum}
% Give examples in model and R that satisfy this assumption. Explain in practice this is reasonable[cite].

% \begin{thm}[Asymptotic variance]
% Assume local permutation, local $\psi_j$,
% \[ \Var(\hFDR) = O(1/p) \]
% \end{thm}

% \begin{thm}[consistency]
% Assume local permutation, good $\psi_j$ s.t. $\psi_j \to 0$ for all non-null $j$, $\hFDR \to \FDR$.
% \end{thm}

\subsection{Standard error estimation by bootstrap} \label{sec:se_est}
 
In practice, we use the bootstrap to estimate the standard error $\se(\hFDR)$ to understand the noise level of $\hFDR$. That is, we calculate the standard error of our estimator on data drawn from a distribution $\widehat P$, where $\widehat P$ is an estimator for the true distribution $P$. For example, we could sample i.i.d. from the empirical distribution of data $\mat D$ in nonparametric settings (we will refer to this as the i.i.d. bootstrap), or the distribution specified by the maximum likelihood estimator (MLE) in parametric settings.

Making a good choice of $\widehat P$ for estimating $\se(\hFDR)$ requires careful consideration. One major issue with both the standard i.i.d. bootstrap and the MLE-based parametric bootstrap is that the noise variables in $P$ would all become weak signal variables in $\widehat P$; as a result, $\cR$ might select many more variables in a typical draw from $\widehat P$ (in which the true model is dense) than it would in a typical draw from $P$, resulting in a poor estimate of the standard error of $\hFDR$. A similar issue could arise in any parametric bootstrap method using an estimator for $\vct \theta$ that is too liberal about letting variables into the model. However, an estimator that is too aggressive in shrinking signal variables could have the opposite issue, if the signal strength in $\widehat P$ is systematically smaller than the signal strength in $P$.

% in which case a liberal selection procedure might select a lot more variables under $\widehat P$ than under $P$. This would introduce a bias in estimating $\se(\hFDR)$ using $\widehat P$.

% To see this, consider a simple example in the Gaussian linear model \eqref{eq:linear-model} with orthogonal design. That is, the explanatory variable vectors $\vct X_j$ for $j = 1, \ldots, d$ are orthogonal to each other. Assume $d = 100$ and $10$ of them are significant -- 

To ameliorate this issue, we try to construct $\widehat P$ in such a way that most of the null variables in $P$ remain null in $\widehat P$, but the signal variables are estimated with little bias. Formally, let $\widehat{\cH_0}$ denote an estimate of the index set of noise variables. In the parametric setting $\mat D \sim P_{\vct \theta}$, we find a sparse MLE of $\vct\theta$, denoted as $\hat{\vct\theta}$, under the constraints that $\theta_j = 0$ for all $j \in \widehat{\cH_0}$. Then we set $\widehat P = P_{\hat{\vct\theta}}$ and draw bootstrap samples to calculate $\hse(\hFDR)$.
We formally describe our bootstrap scheme in the parametric setting in Algorithm \ref{alg:para_BS}.

% \footnote{With a slight abuse of the notation, here $K$ is a generic positive integer that represents the number of bootstrap samples. We set $K = 10$ in our numerical studies.}
\begin{algorithm}
\caption{$\hFDR$ standard error estimation (parametric)} \label{alg:para_BS}
\begin{algorithmic}
\State \textbf{Input}: The observed data $\mat D$ that follows parametric distribution $P_{\vct \theta}$; an estimator $\widehat{\cH_0}$ for the indices with $H_j:\, \theta_j = 0$ is true for $j \in \widehat{\cH_0}$.

\State \textbf{Compute}: the MLE under the intersection null model with $\theta_j = 0$ for all $j \in \widehat{\cH_0}$
\[
\hat{\vct\theta}^{(*)}(\widehat{\cH_0}; \mat D) =  \underset{\theta_j = 0,\, j \in \widehat{\cH_0}}{\argmax}\; P_{\vct \theta}(\mat D)
\]
\State \textbf{Loop}: \textbf{for} {$m = 1, 2, \ldots, M$} \textbf{do}
\State \qquad generate sample $\hat{\mat D}^{(m)} \sim P_{\hat{\vct\theta}^{(*)}}$; \quad compute $\hFDR^{(m)} = \hFDR(\hat{\mat D}^{(m)})$

\State \textbf{Output}: Sample standard error
\[
\hse = \sqrt{\frac{1}{M-1} \sum_{m=1}^M \pth{\hFDR^{(m)} - \frac 1 M \sum_{l=1}^M \hFDR^{(l)} }^2}
\]
\end{algorithmic}
\end{algorithm}

In our implementation, we construct $\widehat{\cH_0}$ as the set of variables neglected by the best model in cross-validation. More precisely, our cross-validation scheme is built to select $\widehat{\cH_0}$ such that the constrained MLE model $P_{\hat{\vct\theta}(\widehat{\cH_0})}$ is close to the true distribution $P$. It estimates the constrained MLE with the candidate $\widehat{\cH_0}$ on the training set and employs the negative log-likelihood on the validation set as the error measurement. A formal description of our cross-validation algorithm is in Appendix~\ref{app:cv_bs}. In Section \ref{sec:introduction}, we demonstrated that the model that minimizes the cross-validation error tends to over-select variables moderately: it contains most of the signal variables and a bunch of noise variables. This is exactly what we need here -- to include most of the signals while roughly keeping the signal sparsity level. We give examples of the constrained MLE problem next.

\begin{example}[Gaussian linear model] \label{ex:linear_model_bootstrap}
    In the Gaussian linear model, the negative log-likelihood $-\log P_{\theta}(\mat D)$ is proportional to the mean-squared error. Therefore, $\vct {\hat\theta}(\widehat{\cH_0})$ is the least square coefficients regressing $Y$ on variables $X_j$ for $j \in \widehat{\cH_0}^\setcomp$.
\end{example}

\begin{example}[Gaussian graphical model]
    In the Gaussian graphical model, we parameterize the model using the symmetric positive definite (SPD) precision matrix $\mat \Theta$. Let $\widehat{\mat \Sigma}$ denote the sample covariance matrix. The log-likelihood is
    \[
    \log P_{\mat \Theta}(\mat D) = \log \det \mat \Theta - \textnormal{trace}(\widehat{\mat \Sigma} \mat \Theta).
    \]
    The constrained MLE $\mat {\hat\Theta}$ is obtained by maximizing $P_{\mat \Theta}$ under the constraint that $\mat \Theta_{jk} = 0$ for pairs $(j,k)$ in $\widehat{\cH_0}$. While there are many algorithms that can solve this problem \citep[e.g.][]{yuan2007model}, we apply a simple heuristic algorithm described in Appendix~\ref{app:glasso_mle} to save computation. 
    % It can be solved using the interior point method for the ``maxdet'' problem \citep{yuan2007model}.
\end{example}

\begin{algorithm}
\caption{$\hFDR$ standard error estimation (non-parametric, model-X)} \label{alg:nonpara_BS}
\begin{algorithmic}
\State \textbf{Input}: The observed data matrix $\mat D$ whose $n$ rows are independent observations $(y_i, x_{i,1}, \ldots, x_{i,d})$; an index set $\widehat{\cH_0}$ that we believe $H_j:\, \theta_j = 0$ is true for $j \in \widehat{\cH_0}$.

\State \textbf{Preparation}: For simplicity, denote $x_{i, \widehat{\cH_0}} = (x_{i,j}, \; j \in \widehat{\cH_0})$ as the $i$th observation of $X_{\widehat{\cH_0}}$. Let $\widehat P_-$ denote the empirical distribution of $(Y, X_{-\widehat{\cH_0}})$
\[
\widehat P_-(\,\cdot\,) = \frac 1n \sum_{i=1}^n \one\set{(y_i, x_{i, -\widehat{\cH_0}}) \in \,\cdot\,}.
\]
\State \textbf{Loop}: \textbf{for} {$m = 1, 2, \ldots, M$} \textbf{do}
\State \qquad \textbf{for} {$i = 1, \ldots, n$} \textbf{do}: Generate sample $(\hat y_i, \hat x_{i, -\widehat{\cH_0}}) \sim \widehat P_-$ and $\hat x_{i, \widehat{\cH_0}} \sim X_{\widehat{\cH_0}} \mid X_{-\widehat{\cH_0}} = \hat x_{i, -\widehat{\cH_0}}$.

Let $\hat{\mat D}^{(m)}$ be a data matrix whose $i$th row is $(\hat y_i, \hat x_{i,1}, \ldots, \hat x_{i,d})$.

Compute $\hFDR^{(m)} = \hFDR(\hat{\mat D}^{(m)})$

\State \textbf{Output}: Sample standard error
\[
\hse = \sqrt{\frac{1}{M-1} \sum_{m=1}^M \pth{\hFDR^{(m)} - \frac 1 M \sum_{l=1}^M \hFDR^{(l)} }^2}
\]

\end{algorithmic}
\end{algorithm}

For the model-X setting, we can employ an analogous method by using a variant of the i.i.d. bootstrap that constrains $H_j:\; Y \independent X_j \mid X_{-j}$ to be true for all $j \in \widehat{\cH_0}$. Specifically, we first draw bootstrap samples of $(Y, X_{-\widehat{\cH_0}})$ from the empirical distribution of the data. Then, we sample $X_{\widehat{\cH_0}}$ from the conditional distribution $X_{\widehat{\cH_0}} \mid (Y, X_{-\widehat{\cH_0}})$, which is equal to $X_{\widehat{\cH_0}} \mid X_{-\widehat{\cH_0}}$ under the intersection null $\bigcap_{j \in \widehat{\cH_0}} H_j$. We formally describe our bootstrap scheme in the model-X nonparametric setting in Algorithm \ref{alg:nonpara_BS}.

If the researchers make model assumptions, the set $\widehat{\cH_0}$ can be constructed in the same way as in the parametric settings. If no model assumption is made and the likelihood is unavailable, researchers can choose $\widehat{\cH_0}$ via p-value-based approaches. For example, we can calculate the conditional randomization test p-values \citep{candes2018panning} for each variable and define $\widehat{\cH_0}$ as all variables whose p-values are larger than $0.1$. This choice lets the majority of null variables remain null in $\widehat P$ while retaining the signal for strong signal variables.

% We can always use parametric/non-parametric BS est of variance to assess the uncertainty of $\hFDR$. However, since $\hFDR$ is nonlinear in the data, BS may need a pretty large $p$ to converge. For example, consider a simple model: independent p-values, select by thresholding each $p_j$, large $p$, $\psi_j = 1$. Then $\hFDR \approx p \alpha / R$, $R$ binomial. Not hard to see $\Var(\hFDR)$ depends on the signal strength: All strong signal, var smaller; none signal, var larger. Assume all null variables, 
% BS var est is biased.

% To improve the performance of BS, we recommend parametric BS assisted by the selection procedure if applicable: MLE after selection, then bootstrap from MLE. example: OLS after lasso.

\subsection{Theory of  parametric bootstrap variance estimation} 
\label{sec:bootstrap_consistency}

It is difficult to study the theoretical property of the bootstrap variance estimator in full generality. Nevertheless, we make substantial progress for Gaussian linear models with Lasso selection. In particular, we revisit the stylized block-orthogonal Gaussian linear model used in our theoretical variance analysis (Section~\ref{sec:se_bound}) with known error variance $\sigma^2$ for simplicity and consider the Lasso selection with a appropriately chosen penalty level in an  asymptotic regime where the true and estimated FDR are bounded away from $0$ and $1$ in the limit. 

For variance estimation, we consider the parametric bootstrap in Algorithm~\ref{alg:para_BS} with $M=\infty$, which assumes away Monte Carlo error, and the estimated null set
\[
\widehat{\cH_0}=\{j:\ \ltheta_j(\bar\lambda)=0\},
\]
where $\bar\lambda>\lambda$ is a slightly inflated penalty level used only for estimating the bootstrap model.  
We then choose $\hat{\vct\theta}^{(*)}$ as the OLS estimator restricting to $\widehat{\cH_0}$, as described in Example~\ref{ex:linear_model_bootstrap}. We call the resulting model an OLS-after-Lasso model. Define $\hFDR^{(*)}$ as a generic bootstrap sample of $\hFDR$ computed under OLS-after-Lasso model.

In this setting, we prove that the asymptotic variance of the parametric bootstrap FDR estimates conditional on data matches the asymptotic variance of our FDR estimator. We elaborate the setting and technical details in Appendix \ref{sec:asymptotic-setup}. 

\begin{thm}
Consider a Gaussian linear model with known $\sigma^2$ that satisfies the Assumption \ref{ass:block-orth}. In the setting of Theorem \ref{thm:bootstrap} in Appendix \ref{subsec:bootstrap_main}, there exist deterministic sequences $\nu_d = o(1)$ and $\tilde{\alpha}_d$ that are bounded away from $0$ and $1$, and a data-dependent sequence $\alpha_d^{(*)} = \alpha_d^{(*)}(\vct X, \vct Y)$ such that 
\[\sup_{t\in\RR}
\biggl|
\PP\!\biggl(
  \frac{1}{\nu_d}
  \bigl(\hFDR - \tilde{\alpha}_d\bigr)
  \le t
\biggr)
- \Phi(t)
\biggr| = o(1),\]
and 
\[\sup_{t\in\RR}
\biggl|
\PP\!\biggl(
  \frac{1}{\nu_d}
  \bigl(\hFDR^{(*)} - \alpha_d^{(*)}\bigr)
  \le t
  \,\Big|\, \mat X, \vct Y
\biggr)
- \Phi(t)
\biggr| = o_\PP(1),\]
where $\Phi$ denotes the CDF of $\cN(0, 1)$.
\end{thm}

\section{Real world examples}
\label{sec:real_data}

% help comparing methods via ROC-like curve. Assess FDR when FDR control is not preferred. More informative variable selection, used with CV, need concrete examples.

In this section we illustrate our method on real scientific data in two different modeling settings. In Section~\ref{sec:real_HIV}, we analyze an HIV drug resistance data set using Lasso regression, and in Section~\ref{sec:real_protein} we analyze a protein signaling network using the graphical Lasso.

\subsection{HIV drug resistance studies} \label{sec:real_HIV}

We first re-analyze the Human Immunodeficiency Virus (HIV) drug resistance data set of \citet{rhee2006genotypic}, who found that regularized regression algorithms were effective at predicting which {\em in vitro} virus samples would be resistant to each of 16 antiretroviral drugs, based on which genetic mutations were present in the samples. In this problem it is of scientific interest not only to {\em predict} from the virus's genetic signature whether it would be resistant to various drugs --- the focus of the original study --- but also to {\em identify} which mutations best explain drug resistance, which was the focus of the re-analysis in \citet{barber2016knockoff}. As a result, both model fit and variable selection accuracy are highly relevant, so it may be scientifically attractive to tackle both problems with a unified analysis.

Our initial processing of the data follows that of \citet{barber2016knockoff}, yielding a different regression data set for each of 16 drugs in three different drug classes. In the data set for a given drug, the response $y_i$ is a measure of sample $i$'s resistance to that drug, and the explanatory variables are binary indicators of whether each of several hundred mutations was present or absent in that sample. As in \citet{barber2016knockoff}, we adopt the Gaussian linear model assumption in this study.
 
As a measure of replicability, \citet{barber2016knockoff} assessed whether detected genes were also present in a list of treatment-selected mutations (TSM), which a previous study had found were enriched in patients who had been treated by drugs in each of the three classes \citep{rhee2005hiv}. \citet{rhee2006genotypic} also found that excluding {\em a priori} all but the TSM list of mutations was an effective form of variable selection. Following \citet{barber2016knockoff}, we treat the TSM list, which comes from an independent data set, as an imperfect approximation to the ``ground truth'' set of signal mutations.\footnote{In addition to the TSM list, we include a small number of variables with extreme $t$-statistics as ``ground truth'' signals in each regression. Specifically, we consider all variables rejected by the Bonferroni procedure at level $\alpha = 0.01$ to be ``ground truth'' signal variables.} Having an estimate of the ground truth from another data set allows us to assess our estimator's effectiveness in measuring the true FDR in this example; in most scientific problems, the ground truth is unavailable and we would have to rely on our method to measure FDR on its own.

\begin{figure}[tbp]
    \centering
    \begin{subfigure}[b]{0.45\linewidth}
      \centering
      \includegraphics[width=\linewidth]{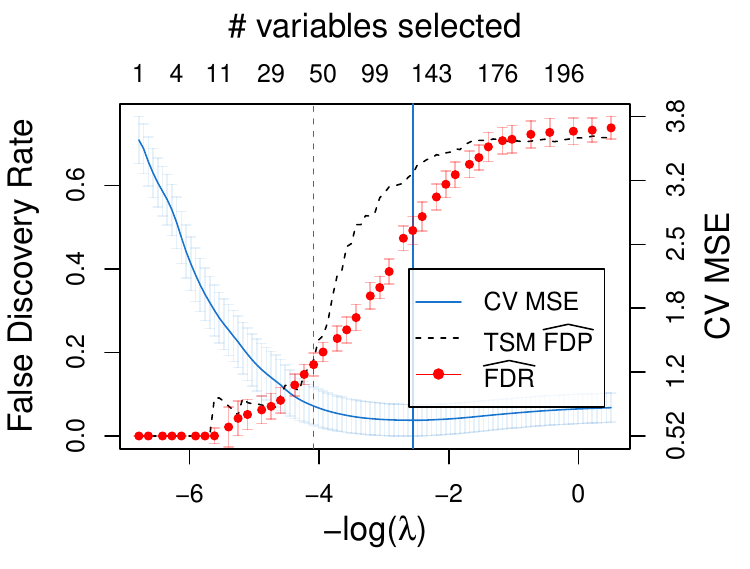}
      \caption{NFV drug, $d = 207$, $n = 842$.}
      \label{fig:hiv_nfv}
    \end{subfigure}    
    \begin{subfigure}[b]{0.45\linewidth}
      \centering
      \includegraphics[width=\linewidth]{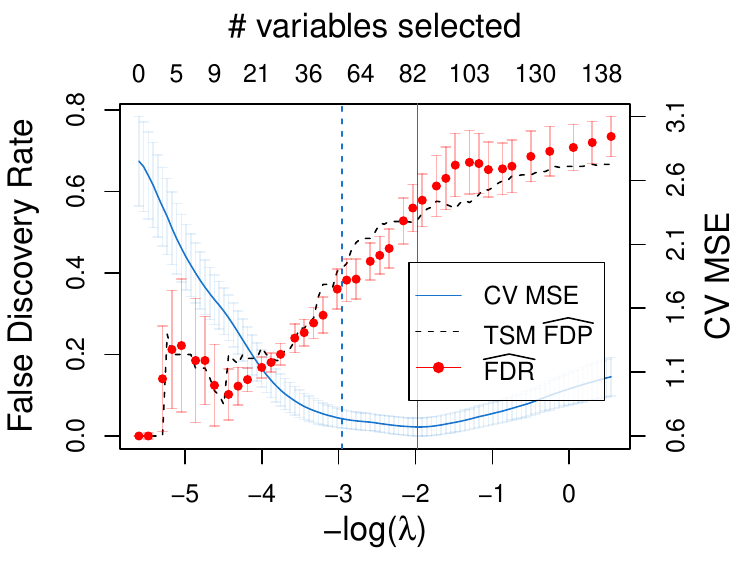}
      \caption{ATV drug, $d = 147$, $n = 328$.}
      \label{fig:hiv_atv}
    \end{subfigure}
    \par\bigskip
    \begin{subfigure}[b]{0.45\linewidth}
      \centering
      \includegraphics[width=\linewidth]{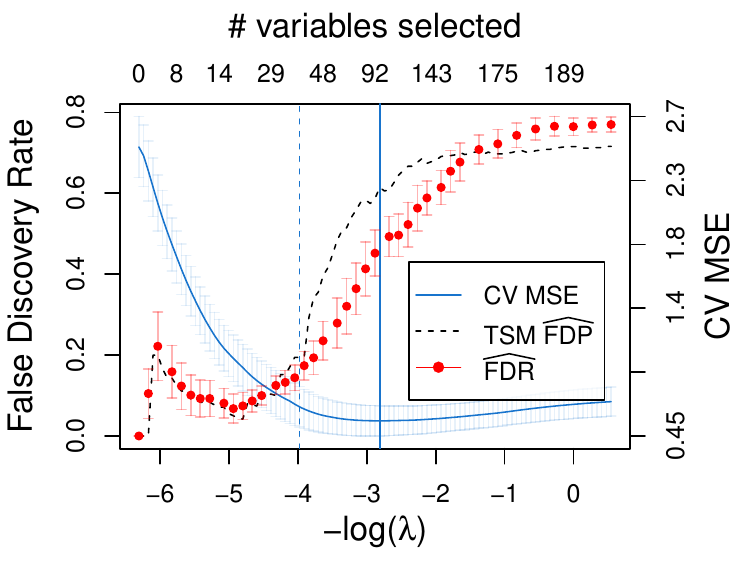}
      \caption{APV drug, $d = 201$, $n = 767$.}
      \label{fig:hiv_apv}
    \end{subfigure}    
    \begin{subfigure}[b]{0.45\linewidth}
      \centering
      \includegraphics[width=\linewidth]{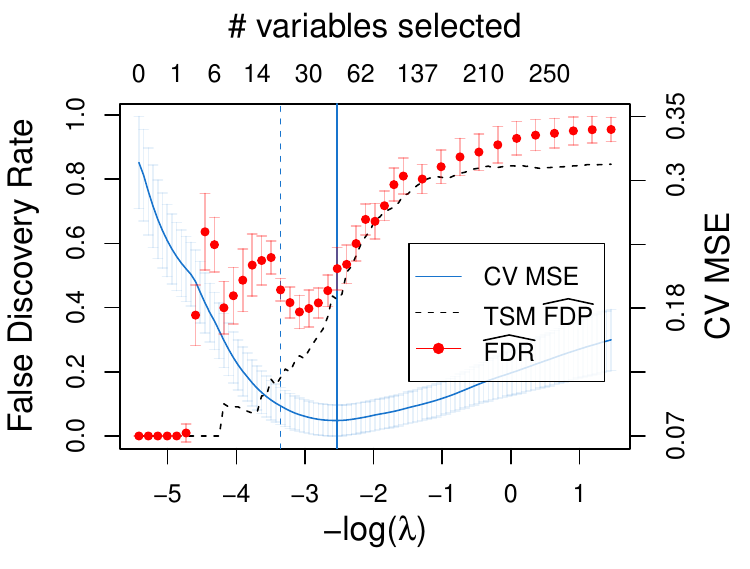}
      \caption{DDI drug, $d = 283$, $n = 628$.}
      \label{fig:hiv_ddi}
    \end{subfigure}
    \caption{$\hFDR$ along with estimated one-standard-error bars by bootstrap, for four representative experiments selected from among the 16 HIV experiments in \citet{rhee2006genotypic}.
    % In each panel, the $\hFDP$ curve is calculated using the TSM set from \citet{rhee2005hiv}, and can be viewed as an independent estimate for the unknown FDR curve in each panel.
    The vertical lines show the locations of the complexity parameters selected by CV and the one-standard-error rule.}
    \label{fig:hiv}
\end{figure}

Figure \ref{fig:hiv} shows our FDR estimates alongside the CV estimate of MSE, for four selected experiments representing a variety of statistical scenarios. Results for all drugs are available in Appendix~\ref{app:extended_results}. The error bars for both CV and $\hFDR$ our estimator show $\pm$one standard error of $\hFDR$ estimated by bootstrap. The $\hFDP$ curves in the figures are estimated using the TSM signals as an approximation to the ground truth.

The four panels of Figure~\ref{fig:hiv} illustrate how different variable selection problems can have very different trade-offs between variable selection accuracy and prediction accuracy. These scenarios would be impossible to distinguish from inspecting the CV curves, but our $\hFDR$ estimator is quite effective in distinguishing them from each other. Panel (a) shows an experiment in which variable selection appears to be easy: the one-standard-error rule simultaneously achieves good model fit and low FDR. Panel (b), by contrast, shows a regime in which no lasso estimator can simultaneously achieve competitive model fit and low FDR. Panel (c) shows an experiment in which the estimated FDR is non-monotone in the complexity parameter $\lambda$: both our estimator and the TSM $\hFDP$ judge the fifth variable to enter the model, gene P71.V, to be a likely noise variable, because its $t$-statistic is $-0.38$, corresponding to a $p$-value of $0.70$, and it is also absent from the TSM list (the $t$-statistics for the first four variables to enter all exceed 5). Panel (d) shows one of the three experiments in which our $\hFDR$ departs substantially from the TSM $\hFDP$. In this experiment, our $\hFDR$ estimator is large mainly because three of the first five selected variables have $p$-values equal to $0.77$, $0.53$, and $0.18$, but the TSM $\hFDP$ curve is small because all three variables appear in the TSM list.

\subsection{Protein-signaling network}\label{sec:real_protein}

To test our method in a realistic Gaussian graphical model setting, we use the flow cytometry data from \citet{sachs2005causal}, which \citet{friedman2008sparse} employed as an illustrative example in their graphical Lasso paper.  The data set records levels of $11$ proteins in individual cells under $14$ different perturbation conditions, in order to estimate relationships in the protein signaling network. For each experimental condition, a sample of cells received a different biological stimulation, and the cells' protein levels were observed.

% The problem of detecting the protein-signaling network is an example different from traditional variable selection. It was an illustrative example in \citet{friedman2008sparse} for the graphical Lasso algorithm. The dataset \citep{sachs2005causal} records levels of $11$ proteins in individual cells under different perturbation conditions respectively. Each perturbation would stimulate the cells and activate certain proteins. Scientists are interested in the relationship between these proteins. For example, the activation of protein PIP3 could lead to the activation of Akt. The goal is to pick up those pairs of proteins that are connected in the signaling network. Since we aim to reveal the signaling mechanism instead of making predictions, the type I error of the selected pairs is a more natural concern than the prediction error. Formally, we adopt the Gaussian graphical model described in Section~\ref{sec:gauss_grapic} and apply the graphical Lasso algorithm to select pairs of proteins that are conditionally dependent given the other proteins.

\begin{figure}[tbp]
    \centering
    \begin{subfigure}[b]{0.45\linewidth}
      \centering
      \includegraphics[width=\linewidth]{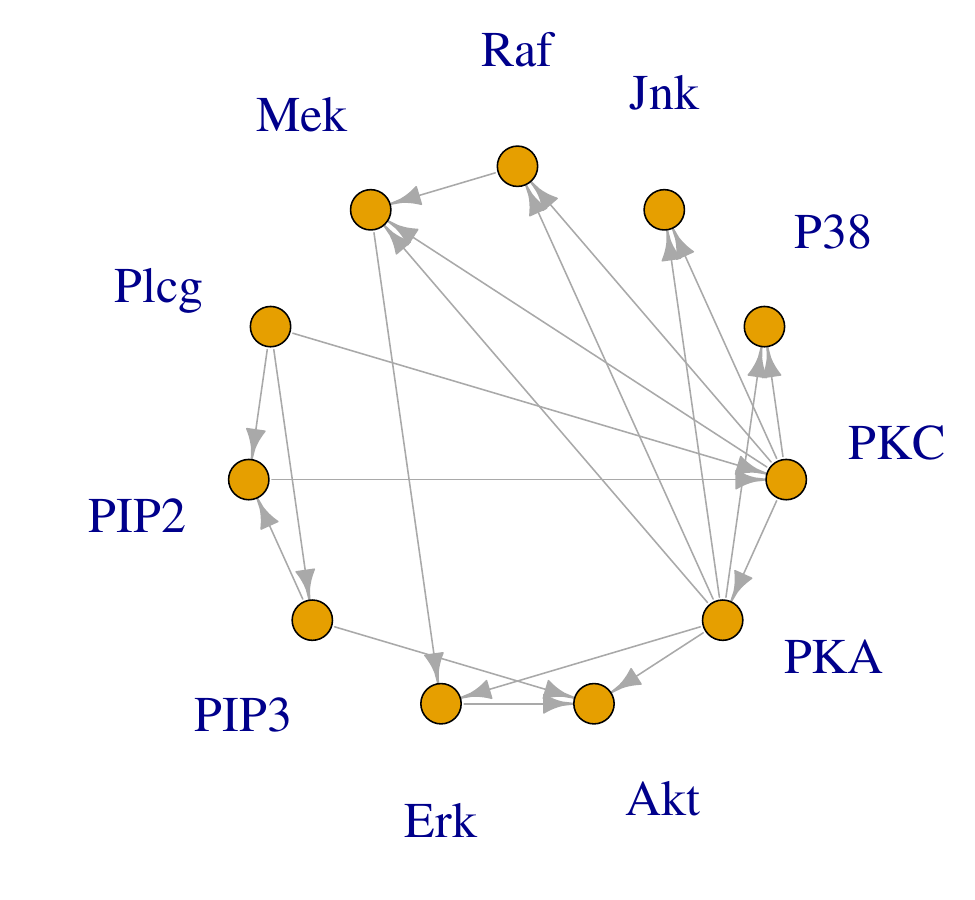}
      \caption{Protein causal DAG.}
      \label{fig:protein_DAG}
    \end{subfigure}
    \begin{subfigure}[b]{0.45\linewidth}
      \centering
      \includegraphics[width=\linewidth]{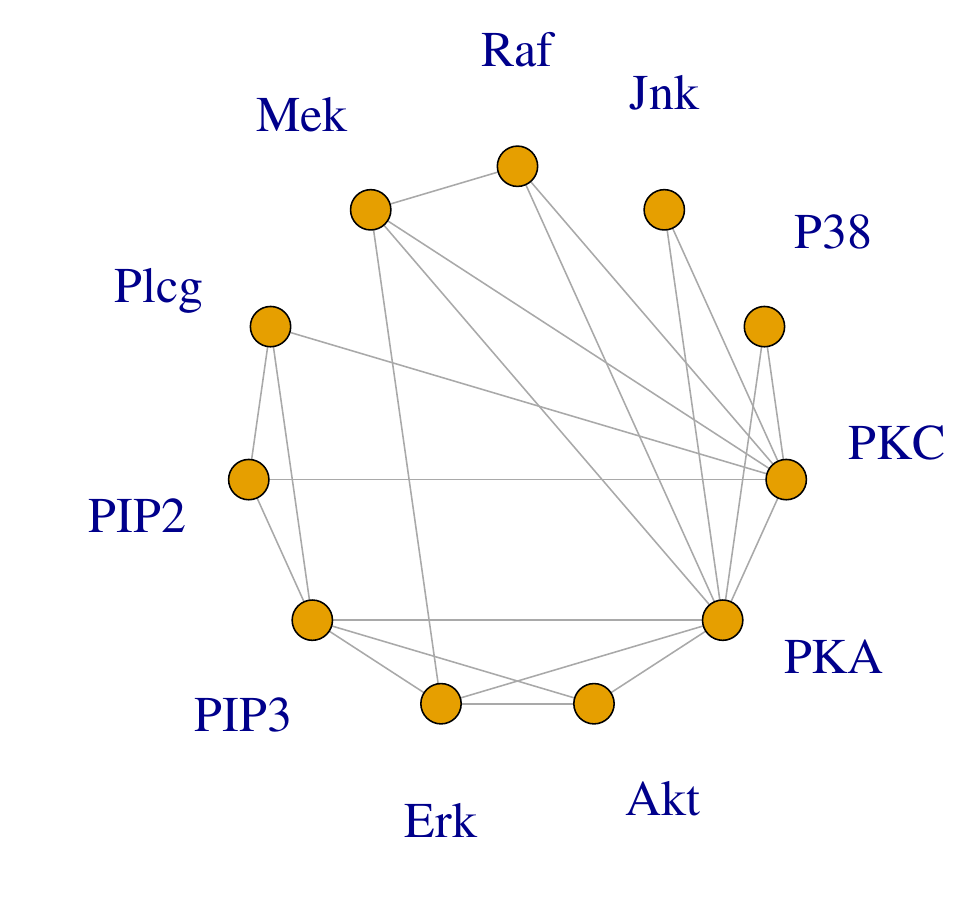}
      \caption{Induced undirected graph.}
      \label{fig:protein_UDG}
    \end{subfigure}
    \caption{Causal DAG supplied by \citet{sachs2005causal} and its induced undirected graph, which we use for independent validation of our FDR estimator. A pair of proteins have an edge in the undirected graph if they are not conditionally independent given the other proteins. There are $22$ edges in the induced undirected graph, out of $55$ possible edges.}
    \label{fig:protein_true}
\end{figure}

\citet{sachs2005causal} supplied a causal directed acyclic graph (DAG) on the proteins, shown here in Figure~\ref{fig:protein_DAG}. The causal graph was based on theoretical and empirical knowledge from previous literature and validation experiments that were independent from the flow cytometry dataset analyzed by \citet{friedman2008sparse}. After converting the directed causal graph to its induced undirected conditional dependence graph, shown in Figure~\ref{fig:protein_UDG}, we obtain an estimate of the set of true signals. As in the HIV data set, we also include as ground truth signals in each experiment a few additional pairs for which the $p$-values pass the Bonferroni threshold at level $\alpha = 0.01$; i.e., those whose $t$-statistic $p$-values are below $1.8 \times 10^{-4}$. Similarly to the TSM list in the HIV data set, the set of edges obtained above functions as an approximation to the ground truth, giving us a point of reference to help evaluate our FDR estimator.

% Therefore, just like the TSM set in the HIV drug resistance studies, we employ the induced conditional dependence undirected graph (Figure \ref{fig:protein_UDG}) as a good approximation to the ground truth\footnote{Besides the induced graph, we treat the pairs with extremely small p-values as signals as well. Specifically, these are pairs that pass the Bonferroni test at level $\alpha = 0.01$: $\set{(X_i, X_j):\, p_{ij} < \alpha / \text{\#pairs}}$. The exact value of the threshold is $1.8 \times 10^{-4}$.}, hence providing a reference for the performance of our $\hFDR$.

\begin{figure}[tbp]
    \centering
    \begin{subfigure}[b]{0.45\linewidth}
      \centering
      \includegraphics[width=\linewidth]{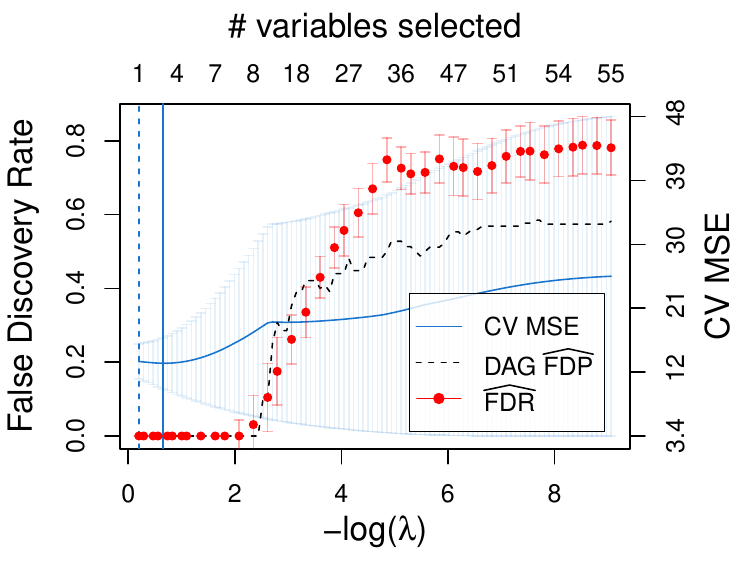}
      \caption{Experiment 1: $d = 11$, $n = 853$.}
      %\caption{Stimulation (CD3, CD28), $d = 11$, $n = 853$.}
      \label{fig:protein_1}
    \end{subfigure}
    \begin{subfigure}[b]{0.45\linewidth}
      \centering
      \includegraphics[width=\linewidth]{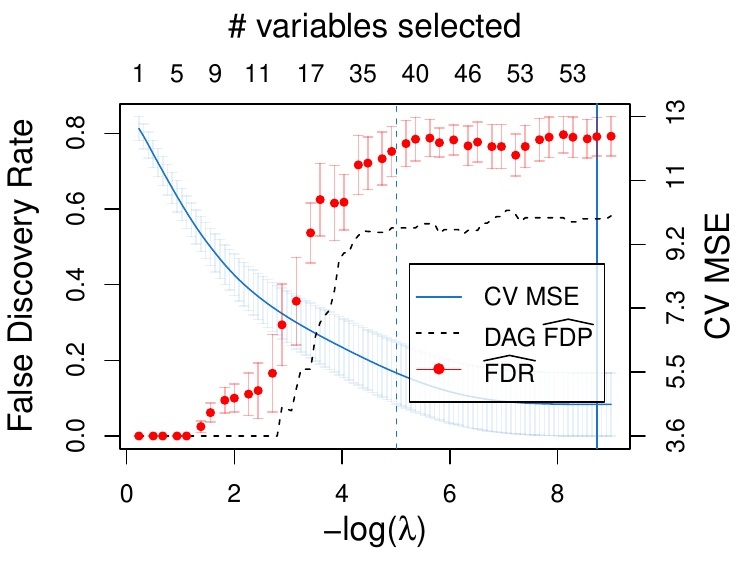}
      \caption{Experiment 11: $d = 11$, $n = 799$.}
      %\caption{Stimulation (CD3, CD28, ICAM-2, G0076), $d = 11$, $n = 799$.}      \label{fig:protein_2}
    \end{subfigure}
    \caption{$\hFDR$ along with estimated one-standard-error bars by bootstrap. The $\hFDP$ curves are calculated by with a ``ground truth'' based on the dependence graph in Figure \ref{fig:protein_UDG}. Cross-validation alone tells little about the right way to select variables (pairs) while our $\hFDR$ is informative.}
    \label{fig:protein}
\end{figure}

Figure~\ref{fig:protein} shows two representative examples of the results from running the graphical Lasso on the $14$ experimental settings. Results for all experiments are available in Appendix~\ref{app:extended_results}. As in the HIV study, we show our $\hFDR$ and its estimated standard error, alongside an estimate of FDP calculated using the causal DAG supplied by \citet{sachs2005causal} (denoted as the DAG $\hFDP$ in Figure~\ref{fig:protein}), and the CV estimate of the predictive log-likelihood. In both panels of Figure~\ref{fig:protein}, the models that we would select by cross-validation are nearly useless: one selects almost no edges, while the other selects every edge. \citet{friedman2008sparse} analyzed the data somewhat differently than we did, pooling all $14$ experiments together, but likewise reported that the CV curve continuously decreased until all edges were present in the model. By contrast, our $\hFDR$ gives useful information about the selection accuracy, and suggests in each experiment a natural stopping point.

The main way in which the $\hFDR$ curves in all $14$ experiments depart from the DAG $\hFDP$ estimates is that our method consistently estimates the FDR to be a bit larger across the board. The reason for this departure is that the evidence for a number of the edges identified in the DAG appears to be very weak in the data sets we analyze. That is, some of the ``ground truth'' edges identified by \citet{sachs2005causal} appear to be absent in the actual data. This fact does not imply that those edges are truly absent in the data set, but if they are present their signals are quite weak, so our estimator treats them as likely nulls and thereby ``overestimates'' the $\FDR$.
% their signals may be too weak for our bias-correction factors to reliably recognize, leading to systematic upward bias in $\hFDR$.

% An interesting observation is that when many variables get selected, $\hFDR$ is consistently larger than the estimated FDP among most of the 14 experiments. This does not necessarily mean $\hFDR$ works poorly here. The ``ground truth'' in Figure \ref{fig:protein_UDG} is an aggregated result from many other studies. It is possible that in one stimulation experiment that we consider, only part of the pairs of proteins are activated by the stimulation. As a result, some pairs in the ``ground truth'' may show no significance in the dataset. As they get selected when the model size increases, the computed FDP becomes consistently smaller than $\hFDR$. In other words, possibly the actual ``ground truth'' in each experiment should be a sub-network of Figure \ref{fig:protein_UDG}. Otherwise, the estimated FDP becomes smaller when the model size is large.

\section{Simulation studies}\label{sec:simulations}

\subsection{Gaussian linear model} \label{sec:simu_linear}

In this section, we evaluate the performance of $\hFDR$ in simulation environments where the ground truth is known. We begin with the Gaussian linear model and employ Lasso as the variable selection procedure. 

The explanatory variables are generated under four primary scenarios, which cover independent, positively correlated, and negatively correlated explanatory variables as well as sparse $\mat X$ matrices:
\begin{enumerate}
    \item \textbf{IID normal}: independent Gaussian random variables $X_j \simiid \cN(0, 1)$.
    \item \textbf{Auto-regression (X-AR)}: $X_j$ ($j = 1, \ldots, d)$ is an AR(1) Gaussian random process with $\text{cov}(X_j, X_{j+1}) = 0.5$.
    \item \textbf{Regression coefficients auto-regression (Coef-AR)}: The OLS coefficients $\hat\theta_j$ ($j = 1, \ldots, d)$ follow an AR(1) Gaussian random process with $\text{cov}(\hat\theta_j, \hat\theta_{j+1}) = 0.5$. That is, $\textnormal{Cov(X)}$ is banded with negative off-diagonal entries.
    \item \textbf{Sparse}: Independent sparse binary observations $X_{ij} \simiid \text{Bernoulli}(0.05)$.
\end{enumerate}
The signal variable set $\cH_0^\setcomp$ is a random subset of $[d]$ that has cardinality $d_1$, uniformly distributed among all such subsets. To mimic realistic scenarios where there is a mix between weak and strong signal variables, we generate nonzero signals with random signal strength 
\[\theta_j \simiid \theta^* \cdot (1 + \textnormal{Exp}(1))/2,\] for all $ j \in \cH_0^\setcomp$, where $\textnormal{Exp}(1)$ represents a rate-$1$ exponentially distributed random variable and $\theta^*$ is a fixed value controlling the average signal strength. We calibrate a moderate signal strength $\theta^*$ such that the selection procedure will have type II error $\FPR = 0.2$ when $\FDR = 0.2$, where FPR stands for the \emph{False Positive Rate}
\[
\FPR = 1 - \EE \br{\frac{|\cR \cap \cH_0^\setcomp|}{d_1}}.
\]
Figure~\ref{fig:hFDR-fixedX} shows the performance of $\hFDR$ and the bootstrap $\hse$ with $d = 500$, $d_1 = 30$, and $n = 1500$. Both our estimator $\hFDR$ and our standard error estimator $\hse$ perform reasonably well across the board. Our $\hFDR$ shows a small non-negative bias in estimating FDR, and its uncertainty is about the same magnitude as the oracle $\FDP$. Notably, the uncertainty of $\hFDR$ is even smaller than the uncertainty of $\FDP$ when FDR is not too large, which is the region we care most about. As for the standard error estimation, $\hse$ is somewhat noisy as an estimator of $\se(\hFDR)$ but roughly captures the magnitude of uncertainty of $\hFDR$ and shows relatively small bias.
  
\begin{figure}[!tb]
    \centering
    \includegraphics[width = 0.95\linewidth]{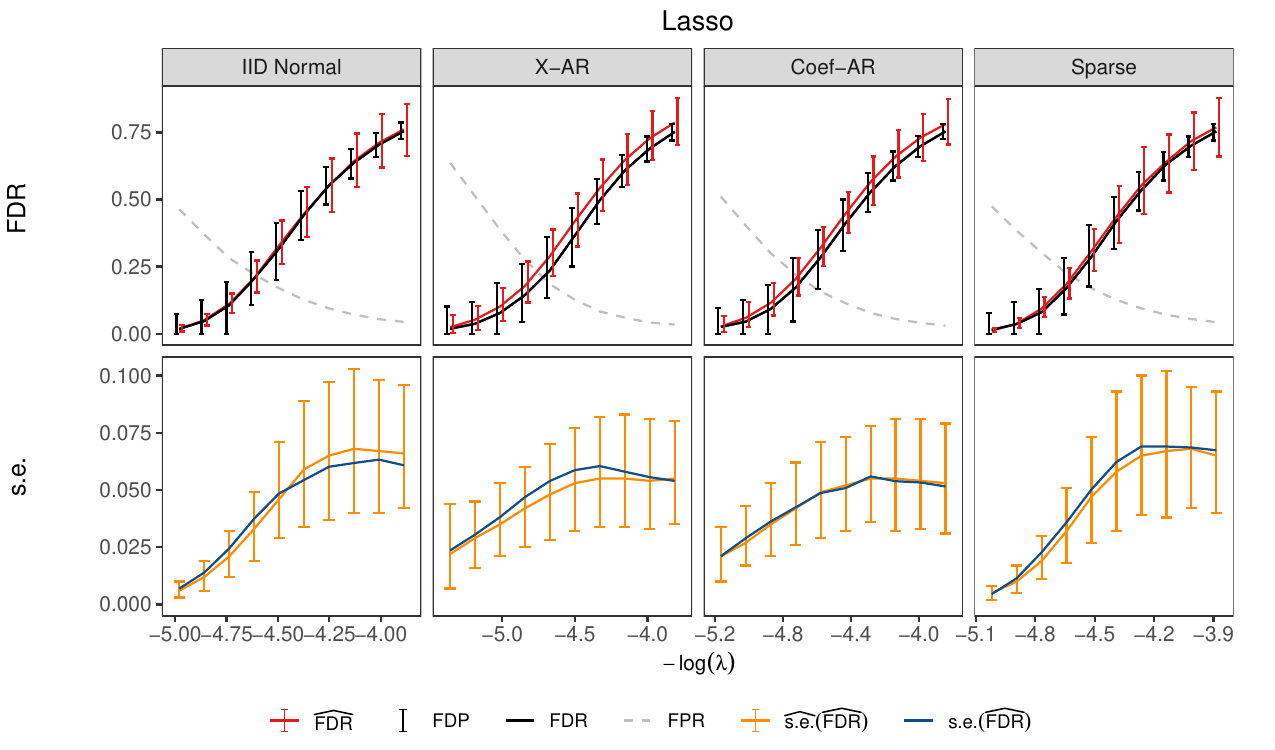}
    \caption{Performance of $\hFDR$ and $\hse$ in the Gaussian linear model. The error bars show the $5\%$ to $95\%$ quantiles of the empirical distributions and the solid lines show the sample average. $\hFDR$ shows a small non-negative bias in estimating FDR and $\hse(\hFDR)$ successfully captures the magnitude of uncertainty of $\hFDR$.}
    \label{fig:hFDR-fixedX}
\end{figure}

Next, we stress-test the performance of $\hFDR$ in several variations on the \textbf{IID normal} scenario above. The other settings remain the same except for those described in each of the four scenarios below. The first three scenarios are intended to illustrate cases where $\hFDR$ may have a substantial bias and the last scenario to illustrate a case where $\hFDR$ could have a high variance.
\begin{enumerate}
    \item \textbf{Weaker Signals}: A weaker signal strength $\theta^*$ such that Lasso has type II error $\FPR = 0.7$ when $\FDR = 0.2$. In this setting, only about $50\%$ of the signal variables have their $\hFDR_j$ contribution killed by the bias correction $\phi$, which could potentially lead to significant bias.
    \item \textbf{Dense Signals}: We have $d = 60$ variables and half of them are signals ($d_1 = 30$). The number of observations is $n = 180$ (same aspect ratio). This could lead to a higher bias because signal variables' contribution is the source of bias.
    \item \textbf{Aspect Ratio Close to 1}: We have $n = 600$ observations for $d = 500$ variables. The IID Gaussian observations are thus closer to collinear. This could lead to higher bias because the $p$-values we use in the bias correction are weak.
    \item \textbf{Very Sparse}: We have only $d_1 = 10$ signal variables, so typical model sizes are small, leading to high variance in $\FDP$ and potentially also in $\hFDR$.
\end{enumerate}

\begin{figure}[!tb]
    \centering
    \includegraphics[width = 0.95\linewidth]{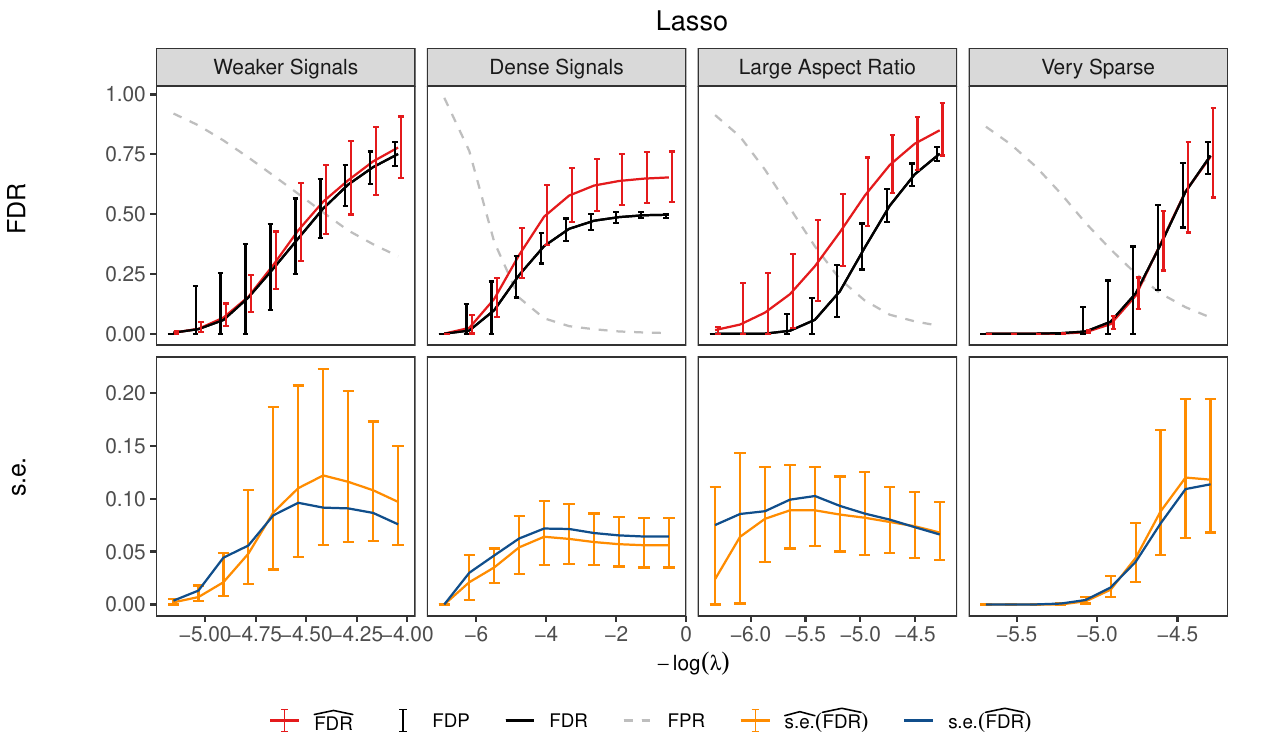}
    \caption{Performance of $\hFDR$ and $\hse$ in additional settings as an edge-case test for bias (column 1-3) and variance (column 4).}
    \label{fig:hFDR-fixedX-extend}
\end{figure}

In Figure~\ref{fig:hFDR-fixedX-extend}, $\hFDR$ does not show a substantial bias in the \textbf{Weaker Signals} setting, but does in the \textbf{Dense Signals} and \textbf{Large Aspect Ratio} settings when $\lambda$ gets small. In the \textbf{Weaker Signals} setting, although many signal variables have $\phi_j > 0$, each of them contributes very little to $\hFDR$ since their probability of being selected is small under the null $H_j: \vct\theta_j = 0$; moreover, because there are many more noise variables than signal variables, the total contribution from the former is still much larger than the contribution from the latter. However, these excess contributions do cause a bias in the \textbf{Dense Signals} setting, because the signal variables' total contributions is not offset by the noise variables' total contribution. The story is a bit different in the \textbf{Large Aspect Ratio} setting. In this scenario, even though the signals are strong enough to be reliably detected by the Lasso, the $p$-values are nearly powerless because the variables are nearly collinear. Hence many signal variables have $\phi_j > 0$. Unlike in the \textbf{Weaker Signals} setting, a signal variable still has a substantial chance to be selected by Lasso under $H_j$ due to the strong correlation structure. Therefore they make a substantial contribution to the bias. Finally, in the \textbf{Very Sparse} setting, we see a higher variance for $\hFDR$ owing to the typically small size of $\cR$, but the $\FDP$ also has a larger variance.

\subsection{Model-X setting} \label{sec:sim_modelX}

Next, we test the model-X version of our procedure in two linear modeling scenarios. The response $Y$ is either a continuous response with Gaussian errors, where we employ Lasso, or a binary response that follows a logistic regression model:
% We test out $\hFDR$ on classical models and selection procedures. For regression problems, we generate the response $Y$ by the Gaussian linear model
% \[ Y = \sum_{j=1}^d X_j \theta_j + \ep, \quad \ep \sim \cN(0, 1),\]
% and employ the Lasso (Example~\ref{ex:coef_selection}) and FS (Example~\ref{ex:opt_selection}) procedures. For classification problems, we adopt the logistic regression model that generates $Y \in \set{0, 1}$ with
\[ \ln \pth{\frac{\PP(Y = 1)}{1 - \PP(Y = 1)}} = \sum_{j = 1}^d X_j \theta_j, \]
where we employ $\ell_1$-penalized logistic regression for variable selection. The other settings are the same as in Section~\ref{sec:simu_linear}.

\begin{figure}[!tb]
    \centering
    \includegraphics[width = 0.95\linewidth]{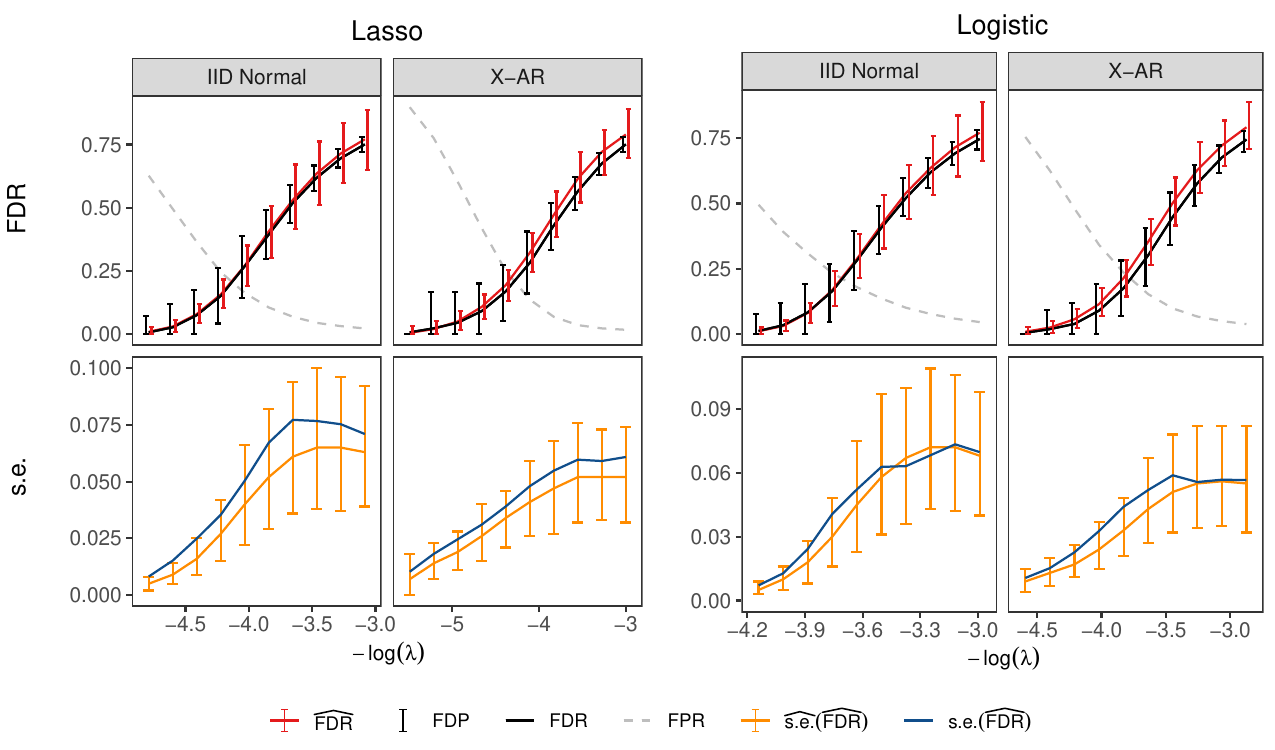}
    \caption{Performance of $\hFDR$ and $\hse$ in the model-X settings. The error bars show the $5\%$ to $95\%$ quantiles of the empirical distributions and the solid lines show the sample average. $\hFDR$ shows a small non-negative bias in estimating FDR and $\hse(\hFDR)$ successively captures the magnitude of uncertainty of $\hFDR$.}
    \label{fig:hFDR-modelX}
\end{figure}

Figure \ref{fig:hFDR-modelX} shows the performance of our estimators where the regression problem (left panel) is high-dimensional with $d = 500$ and $n = 400 < d$ and the classification problem (right panel) has $d = 500$ and $n = 1500$. As in Section~\ref{sec:simu_linear}, our estimators perform well.
% In the figures, the error bars show the $5\%$ to $95\%$ quantiles of the empirical distributions and the solid lines show the sample average.
% the dots show the mean of the outliers (not covered by error bars) if the outliers mean deviates more than three standard errors from the sample mean, which is shown in solid lines.
% As expected, our $\hFDR$ shows a small non-negative bias in estimating FDR. Moreover, its uncertainty is about the same magnitude as the oracle $\FDP$. Especially, the uncertainty of $\hFDR$ is even smaller than the one of $\FDP$ when FDR is not too large, which is the region we most care about. Note that the oracle $\FDP$ is naively the best estimator for the FDR. This shows our $\hFDR$ works well. As for the standard error estimation, $\hse$ is not a perfect estimator for $\se(\hFDR)$ but successively captures the magnitude of uncertainty of $\hFDR$.

\subsection{Gaussian graphical model} \label{sec:simu_graph}

Next, we test our method in the problem of detecting a sparse dependence graph as described in Example~\ref{ex:gauss_grapic}. We generate independent samples of $d$ variables
\[ (X_1, \ldots, X_d) \sim \cN(\vct 0, \mat \Sigma), \quad i = 1, \ldots, n, \]
where $\mat \Sigma^{-1} = \mat \Theta = \mat I_d + (\mat B + \mat B^\tran) / 2$ and $\mat B \in \RR^{d \times d}$ has all elements $0$ except for $d_1$ randomly selected upper-triangular entries. The values at these entries are again generated independently from $\theta^* \cdot (1 + \textnormal{Exp}(1))/2$, with $\theta^*$ calibrated such that the graphical Lasso has type II error $\FPR = 0.2$ when $\FDR = 0.2$, and subject to the constraint that $\mat \Sigma$ is positive semidefinite. We set $d = 50$, $n = 2500$, and $d_1 = 30$. Therefore, there are $1225$ possible edges in the dependence graph and we apply the graphical Lasso to estimate the true edge set. Figure \ref{fig:hFDR-glasso} shows the results. The bias of $\hFDR$ is very small and $\se(\hFDR)$ successfully captures the magnitude of uncertainty of $\hFDR$.

\begin{figure}[!tb]
    \centering
    \includegraphics[width = 0.8\linewidth]{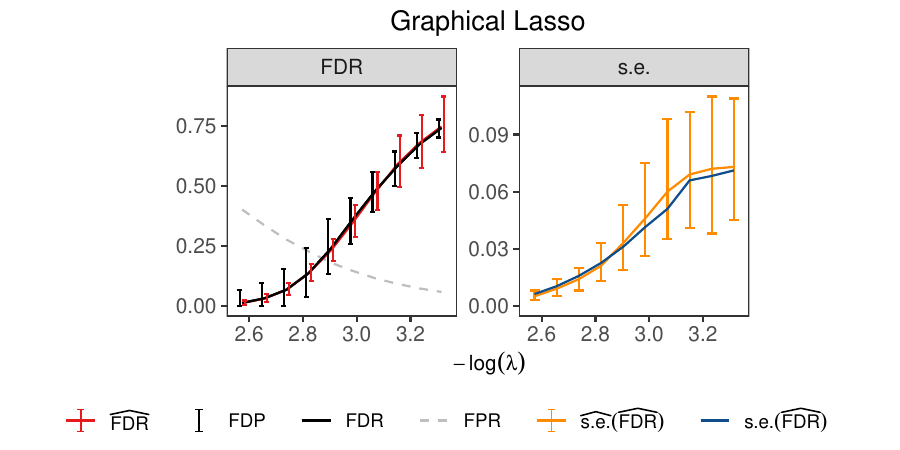}
    \caption{Performance of $\hFDR$ and $\hse$ in the Gaussian graphical model. The error bars show the $5\%$ to $95\%$ quantiles of the empirical distributions and the solid lines show the sample average. $\hFDR$ shows a small non-negative bias in estimating FDR and $\hse(\hFDR)$ successively captures the magnitude of uncertainty of $\hFDR$.}
    \label{fig:hFDR-glasso}
\end{figure}

\subsection{Unfavorable case: equicorrelated Gaussian} \label{sec:simu_unfavorable}

In Section \ref{sec:se_bound}, we show that $\hFDR$ has a small standard error if the variables have a block-independent structure. 
% We demonstrate a toy unfavorable example where $\hFDR$ has a large variance when the observed variables are strongly positively equicorrelated.
In this section, we examine the performance of $\hFDR$ on a practical problem with an unfavorable structure -- the Multiple Comparisons to Control (MCC) problem \citep{dunnett1955multiple}, where we estimate the effects of $d$ treatments by comparing treatment group for each one to a common control group. We assume all $d+1$ groups have the same sample size. Mathematically, the problem is equivalent to considering $d$-dimensional Gaussian vector $\vct Z \sim \cN(\vct \theta, \mat \Sigma)$, where $\mat \Sigma \in \RR^{d \times d}$ has diagonal $1$ and off-diagonal $\rho = 0.5$. So all $Z_j$ variables are positively equicorrelated: $\textnormal{cov}(Z_j, Z_k) = \rho$ for $j \neq k$. We are interested in selecting the signal variables with $\theta_j \neq 0$. As discussed in Section \ref{sec:se_bound}, this is an unfavorable case for $\hFDR$.

The problem can be equivalently formalized as variable selection in a Gaussian linear model whose OLS coefficients are $\vct Z$. We employ the Lasso selection procedure and test $\hFDR$ in this unfavorable case. Figure \ref{fig:hFDR-unfavorable} shows the performance of $\hFDR$ and $\hse$ on the MCC problem with $d = 500$, $n = 1500$, and $d_1 = 30$ signal variables. As expected, $\hFDR$ shows a nonnegative bias and an uncertainty much larger than the oracle FDP. But like in previous simulations, the $\hse$ tells the magnitude of uncertainty of $\hFDR$ and can warn the researchers that $\hFDR$ is not very reliable.

\begin{figure}[!tb]
    \centering
    \includegraphics[width = 0.7\linewidth]{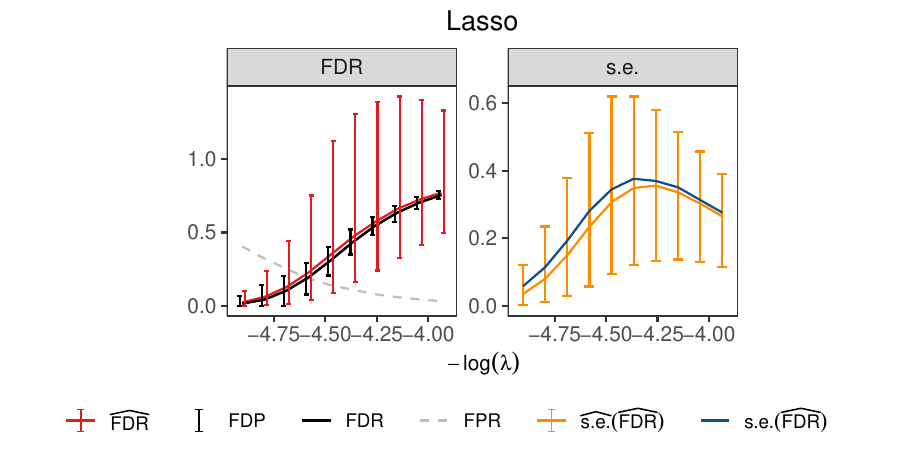}
    \caption{Performance of $\hFDR$ and $\hse$ on the MCC problem. The error bars show the $5\%$ to $95\%$ quantiles of the empirical distributions and the solid lines show the sample average. $\hFDR$ shows a large variance but $\hse(\hFDR)$ successively captures the magnitude of uncertainty of $\hFDR$.}
    \label{fig:hFDR-unfavorable}
\end{figure}

\subsection{Robustness to model misspecification} \label{sec:simu_sensitivity}

The robustness of $\hFDR$ is closely related to the robustness of the conditional test on $H_j$. Take the Gaussian linear model, for example. Figure \ref{fig:showcalc} demonstrates the calculation of
\[
\hFDR^*_j = \EE_{H_j} \br{\frac{\one \set{j \in \cR}}{R} \mid \vct S_j }
\]
on the Lasso selection procedure. Conditioning on $\vct S_j$, there is a one-to-one map between the t-statistic for $H_j$
% \[
% t_j = \frac{\pth{(\mat X^\tran \mat X)^{-1} \mat X^\tran \vct Y}_j}{\sqrt{\pth{(\mat X^\tran \mat X)^{-1}}_{jj} \cdot \RSS / (n-d)}}
% \]
and the sufficient statistic $\mat X^\tran \vct y$. So this conditional expectation can be calculated as a one-dimensional integral of ${\one \set{j \in \cR}}/{R}$ (shown in red curve) against the $t$-distribution density (shown in blue curve). And the calculated $\hFDR^*_j$ is the area in purple. Now consider the Gaussian linear model is misspecified. For example, the noise $\ep_i$ is not Gaussian but follows a t-distribution with $5$ degrees of freedom. The true conditional density of the t-statistic, assuming $\theta_j = 0$, is shown by the orange dashed curve. As expected, the conditional t-test is robust: the true density of the $t$-statistic is not far from the theoretical $t$-distribution density. As a result, the calculated $\hFDR^*_j$ does not change much with our misspecified model. $\phi_j$ is robust for the same reason if it is based on the t-test p-value. Therefore, we expect $\hFDR$ to be robust as long as we have confidence in the robustness of the conditional t-test on each null hypothesis.
\begin{figure}[!tb]
    \centering
    \includegraphics[width = 0.7\linewidth]{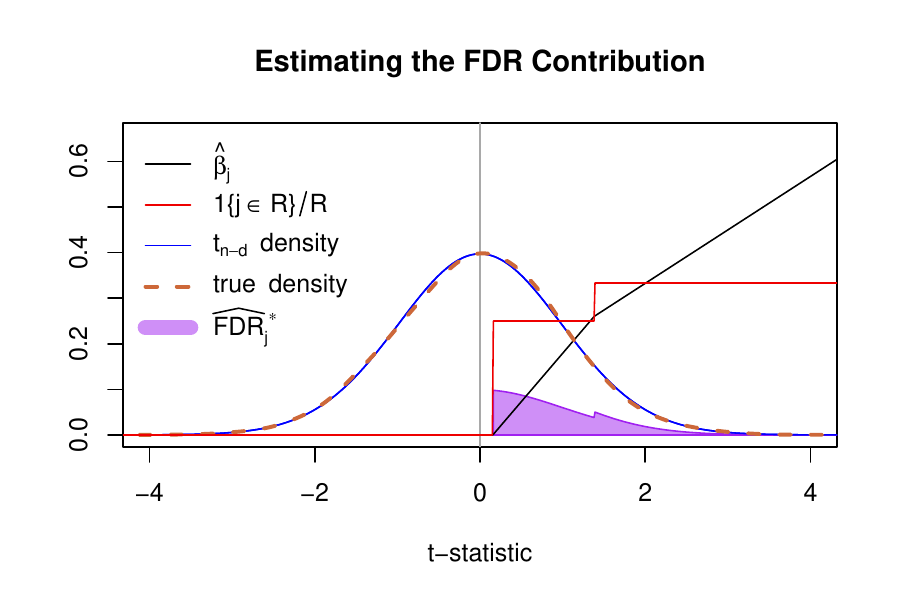}
    \caption{A demonstration of computing $\hFDR^*_j$. The condition expectation is essentially a one-dimensional integral of ${\one \set{j \in \cR}}/{R}$ (shown in red) against the $t$-distribution density (shown in blue). And the calculated $\hFDR^*_j$ is the area in purple. With model misspecification, if the true conditional density (shown in dashed orange) is not far from the misspecified $t$-distribution density, then the calculation of $\hFDR^*_j$ is approximately correct.}
    \label{fig:showcalc}
\end{figure}
The reasoning can be extended to more general scenarios. For example, in the Gaussian graphical model with graphical Lasso selection, we expect $\hFDR$ to be robust to model misspecification if the conditional t-test on $H_{jk}$ is robust. In the model-X settings, consider the selection procedures where $\mat X^\tran \vct Y$ is sufficient. These include Lasso and penalized logistic regression selection. We expect $\hFDR$ to be robust if the conditional distribution of $\mat X_j^\tran \vct Y$ is robust.

\begin{figure}[!tb]
    \centering
    \includegraphics[width = 0.95\linewidth]{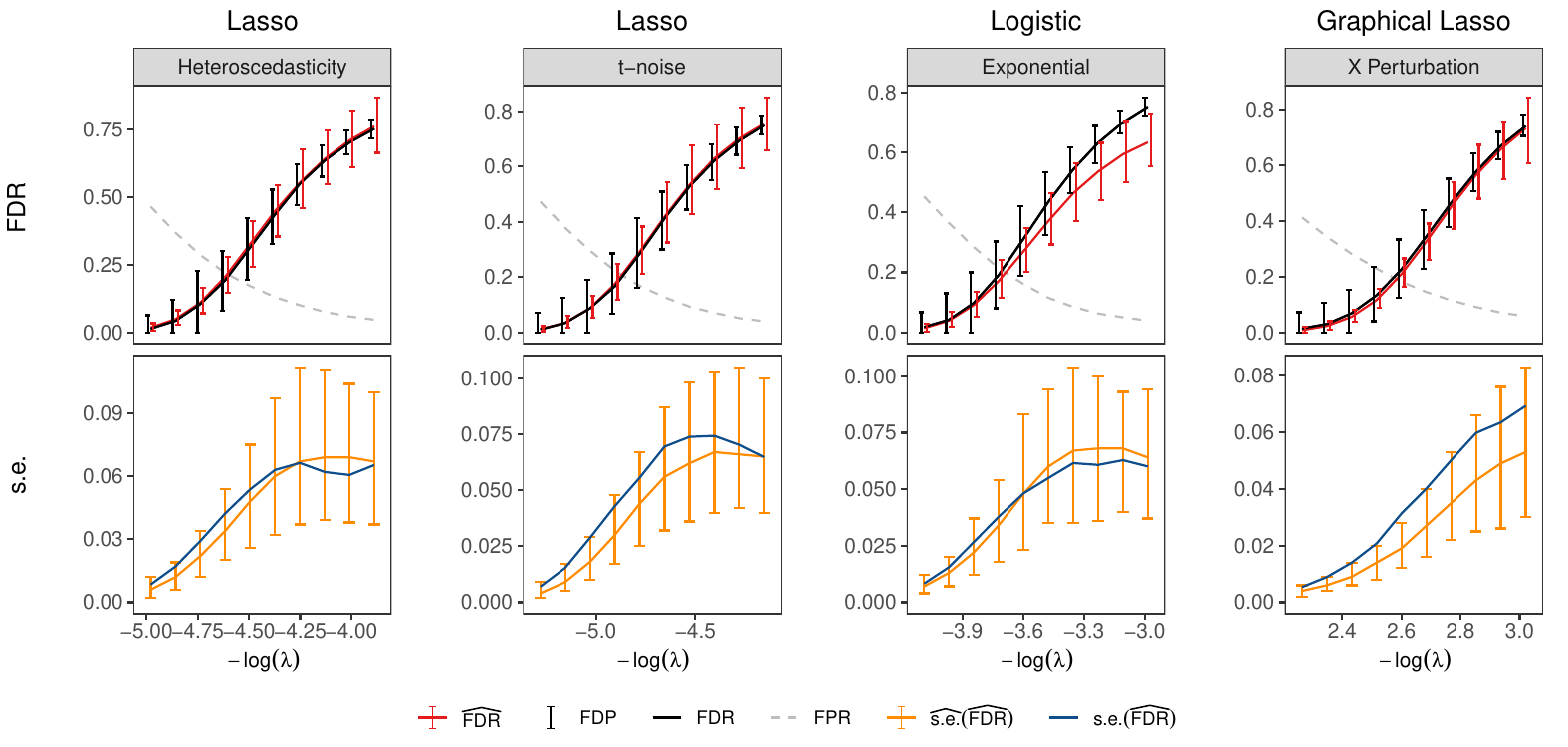}
    \caption{Performance of $\hFDR$ and $\hse$ with model misspecification. The error bars show the $5\%$ to $95\%$ quantiles of the empirical distributions and the solid lines show the sample average.}
    \label{fig:hFDR-robust}
\end{figure}

In Section \ref{sec:real_data}, we have seen $\hFDR$ works well on real-world problems where the model assumptions are never exactly true. Here in Figure~\ref{fig:hFDR-robust}, we test the robustness of $\hFDR$ in a clean simulation study environment. We consider the same settings as previous simulations but with different types of model misspecifications. In the ``Heteroscedasticity'' setting, we introduce heteroscedasticity to the Gaussian linear model. That is, we generate
\[ Y = \sum_{j=1}^d X_j \theta_j + \exp \pth{\sum_{j=1}^d |X_j \theta_j|} \cdot \ep, \quad \ep \sim \cN(0, 1).\]
In the ``t-noise'' setting, we have the noise $\ep$ in the Gaussian linear model follow the t-distribution with a degree of freedom $5$ instead of the standard normal distribution. The ``Exponential'' setting has the variables follow the heavy-tailed exponential distribution $X_j \simiid \textnormal{Exp(1)}$ but are mistaken as Gaussian in the model-X assumption. Finally, In the ``X Perturbation'' setting, we multiply each observed variable by an IID uniformly distributed perturbation $U([0,2])$. As the figure shows, our $\hFDR$ performs well in these scenarios.

\section{Computational details}
\label{sec:compuation}

The major computational task in our method is to calculate conditional expectations
\[
\hFDR^*_j = \EE_{H_j} \br{\frac{\one \set{j \in \cR}}{R} \mid \vct S_j }
\]
for $j = 1, \ldots, d$ with $\psi_j \neq 0$. These calculations can be easily parallelized. In a general setting, $\hFDR^*_j$ can be approximated via Monte Carlo integration: generate $M$ samples of the data $\vct D$ from the known conditional distribution, compute the integrand ${\one\set{j \in \cR}}/{R}$ for each sample, and take the average over them. The worst-case time complexity for computing $\hFDR$ is $O(d \cdot M \cdot T)$, where $T$ is the running time of the selection procedure. We note that calculating the integrand of $\hFDR_j^*$ for each Monte-Carlo sample has the same computation complexity as calculating one placebo statistic in the conditional randomization test (CRT, \citet{candes2018panning}) with the same selection procedure for variable $j$. 
% If we equate the number of Monte-Carlo samples in $\hFDR$ and in the conditional randomization test \citep{candes2018panning}, the two methods will have the same running time. 
% It is worth noticing that the number of Monte-Carlo samples required by $\hFDR$ is typically much smaller. This is because we compute an expectation instead of fine-grid p-values and the accuracy of the Monte-Carlo calculation does not affect the conservative bias property of $\hFDR$. 
In our implementation, we use $20$ Monte-Carlo samples for computing each $\hFDR^*_j$ and $10$ Bootstrap samples.

While the Monte-Carlo procedure can work with any selection procedure, for commonly used ones like Lasso and FS under the Gaussian linear model, we can develop exact and fast algorithms by leveraging the special structures. In particular, $\hFDR^*_j$ reduces to a one-dimensional integral. We briefly describe the algorithms below and leave details in Appendix~\ref{app:fast_alg}. 

\begin{example}[fast and exact calculation of $\hFDR_j^*$ for Lasso] \label{ex:algo_lasso}
    \citet{lei2019fast} proposed a fast homotopy algorithm that computes the entire Lasso solution path when the sufficient statistic $\mat X^\tran \vct Y$ varies in a given direction in the context of conformal prediction. With simple modifications, we apply their algorithm to obtain the Lasso selection set $\cR$ for all possible values of $\vct X_j^\tran \vct Y$ given $\vct S_j$. This provides a full characterization of the piecewise constant function $\one \set{j \in \cR}/{R}$ and hence yields the exact value of $\hFDR_j^*$.

    The runtime of the homotopy algorithm scales linearly in the number of knots along the homotopy path. At each knot, the algorithm  applies a handful of $d$-dimensional matrix-vector multiplications. As a result, the algorithm is faster for a reasonably large $\lambda$ where $\cR$ varies less over the range of $\vct X_j^\tran \vct Y$ given $\vct S_j$.
    % is not sensitive to the perturbation of $\vct X_j^\tran \vct Y$ given $\vct S_j$. 
    % As a specific example, the evaluation of $\hFDR$ in example~\ref{?} takes XX and XX minutes with and without employing the fast and exact homotopy algorithm.
\end{example}

\begin{example}[fast and exact calculation of $\hFDR_j^*$ for FS] \label{ex:algo_fs}
    At the $k$-th step, the FS procedure for the Gaussian linear model selects a variable $j_k$ that has the largest $|\vct X_j^{(k)} \cdot \vct Y^{(k)}|$ where $\vct X^{(k)}, \vct Y^{(k)}$ can be obtained by applying a sequence of normalization and orthogonalization steps. When the number of selected variables $k$ is fixed, $\hFDR^*_j$ simplifies into $\PP_{H_j}(j \in \cR \mid S_j) / k$. The numerator can be decomposed into $\sum_{\ell=1}^{k}\PP_{H_j}(j\text{ is selected at the }\ell\text{-th step}\mid S_j)$. We can calculate each term by using an auxiliary FS procedure that has the same computation cost as the original one. Again, each step only involves a handful of $d$-dimensional matrix-vector multiplications. Thus, the algorithm is faster when the number of selections $k$ is small. 
    
    % This procedure would tell us the threshold of $|\vct X_j \cdot \vct Y|$ if $X_j$ is to get selected. We then compute $\PP_{H_j}(j \in \cR \mid S_j)$ using the known conditions distribution of $|\vct X_j \cdot \vct Y|$ given $\vct S_j$. Finally $\hFDR^*_j$ follows as $\PP_{H_j}(j \in \cR \mid S_j) / \lambda$.
    
    % The running time of the algorithm depends on the number of steps in FS. At each step, the algorithm essentially applies a handful of $d$-dimensional matrix-vector products. As a result, the algorithm is faster when the number of steps $\lambda$ is small. 
    % Let $L$ denote the running time of the algorithm considering one variable. In the worst case, the time complexity of $\hFDR$ employing the algorithm is $O(d \cdot (T + L))$. 
    % As a specific example, the evaluation of $\hFDR$ in example~\ref{?} takes XX and XX minutes with and without employing the fast and exact algorithm.
\end{example}

The running time of evaluating $\hFDR$ (without bootstrap s.e. estimation) at all $10$ different tuning parameter values in the simulation problems in Section~\ref{sec:simulations} is typically several minutes on a single-cored computer. We report the exact running time in Appendix~\ref{app:runtime}.

\section{Summary and discussion}
\label{sec:discussion}

In this work we present a new approach for estimating FDR of any variable selection procedure. Our estimator can be applied in common statistical settings including the Gaussian linear model, the Gaussian graphical model, and the nonparametric model-X framework, requiring no additional statistical assumptions in these settings beyond those required for standard hypothesis tests. Our estimator can be evaluated alongside the cross-validation curve to inform the analyst about the tradeoff between variable selection accuracy and prediction accuracy over the entire regularization path. The bias of the estimator is non-negative in finite samples and typically small. The standard error can be assessed using the specialized bootstrap and it is typically vanishing when there are many variables with only local dependency. Examples with real-world and simulation data demonstrate our estimator works well in both bias and standard error. \texttt{R} package is available online at the Github repository \texttt{https://github.com/yixiangLuo/hFDR}.

While this paper is focused on three settings, our method can be applied insofar as a sufficient statistic $\vct S_j$ exists under the null submodel $H_j$, i.e., the data $\vct D$ has a known distribution conditional on $\vct S_j$ under $H_j$. Consider a common setting in the multiple testing literature where a test statistic $T_j$ is available for $H_j$ and $T_1, \ldots, T_d$ are independent. If $H_j$ are all simple hypotheses (e.g., $H_j: \theta_j = 0$), we can simply set $\vct S_j = (T_1, \ldots, T_{j-1}, T_{j+1}, \ldots, T_d\}$. When the $H_j$ are composite but simple conditional on some auxiliary statistic $V_j$, we can set $\vct S_j = (T_1, \ldots, T_{j-1}, V_j, T_{j+1}, \ldots, T_d\}$. For example, if the researcher has two samples $X_{j1}, \ldots, X_{jn_{j_1}}\stackrel{i.i.d.}{\sim} P$ and $Z_{j1}, \ldots, Z_{jn_{j_2}}\stackrel{i.i.d.}{\sim} Q$ for the hypothesis $H_j: P = Q$, $\vct S_j$ can be always chosen as the unordered set of $\{X_{j1}, \ldots, X_{jn_{j_1}}, Z_{j1}, \ldots, Z_{jn_{j_2}}\}$. If $P = F_{\theta_1}$ and $Q = F_{\theta_2}$ for some exponential family $F$ with natural parameter $\theta$ and sufficient statistic $T(\cdot)$, then we can choose $\vct S_j = \sum_{i=1}^{j_1}T(X_{ji}) + \sum_{k=1}^{j_2}T(Z_{jk})$. Therefore, our method can also be applied to a large class of multiple testing problems with independent test statistics.

\section*{Reproducibility}
Our FDR estimator is implemented in an \texttt{R} package available online at the Github repository \texttt{github.com/yixiangLuo/hFDR}.
The \texttt{R} code to reproduce all simulations and figures is available at \texttt{github.com/yixiangLuo/hFDR\_expr}.

\bibliographystyle{plainnat}
\bibliography{biblio.bib}

\appendix

\section{Additional details in experiments and methodologies}
\subsection{Simulation settings for the illustrative examples} \label{app:intro_sim}

In Figure \ref{fig:illustration-indpt}, the explanatory variables are independent Gaussian -- $X_j \simiid \cN(0, 1)$. We have $d = 200$ variables, $n = 600$ observations, and $d_1 = 20$ variables are signal variables with random signal strength $\theta_j = 0.25 \cdot (1 + \textnormal{Exp}_j(1))/2$, where $\textnormal{Exp}_j(1)$ with $j \in \cH_0^\setcomp$ are independent exponentially distributed random variable with rate $1$. The response $Y$ is generated by the Gaussian linear model with $\sigma = 1$.

In Figure \ref{fig:illustration-corr}, the explanatory variables $X_j$ ($j = 1, \ldots, d)$ form an AR(1) Gaussian random process with $\text{cov}(X_j, X_{j+1}) = 0.8$. The other settings are the same as above.

\subsection{Cross-validation for constructing $\widehat{\cH_0}$ in $\hse$} \label{app:cv_bs}

In this section, we describe the cross-validation algorithm we use to construct the set $\widehat{\cH_0}$ in computing $\hse(\hFDR)$. Formally, let $\mat D_{\text{trn}, k}$ and $\mat D_{\text{val}, k}$ be the $k$th training and validation data set in a $K$-folds cross-validation, respectively. Throughout the section, we write $\vct {\hat\theta}^{(*)}(\widehat{\cH_0}(\lambda); \mat D_{\text{trn}, k})$ with $\widehat{\cH_0}(\lambda) = \cR^\setcomp(\lambda; \mat D_{\text{trn}, k})$ as $\vct {\hat\theta}^{(*)}(\lambda; \mat D_{\text{trn}, k})$ for simplicity of notation. We have
\[
\vct {\hat\theta}^{(*)}(\lambda; \mat D_{\text{trn}, k}) = \underset{\theta_j = 0,\, j \not\in \cR(\lambda; \mat D_{\text{trn}, k})}{\argmax}\; P_{\vct \theta}(\mat D_{\text{trn}, k}).
\]
Then we find the best $\lambda$ that achieves the largest averaged likelihood on the validation sets 
\[
\lambda_{\text{cv}} = \underset{\lambda}{\argmax}\; \frac 1K \sum_{k=1}^K P_{\vct {\hat\theta}^{(*)}(\lambda; \mat D_{\text{trn}, k})}(\mat D_{\text{val}, k})
\]
and define
\[
\widehat{\cH_0} = \set{j: \; j \not\in \cR(\lambda_{\text{cv}}; \mat D)}.
\]

\subsection{Bias contribution from signal variables}
\label{app:signal_bias}

\begin{thm} \label{thm:strong_signal_bias}
    In the Gaussian linear model \eqref{eq:linear-model}, consider the canonical version \eqref{eq:hFDR_p_def} of $\hFDR_j$:
    \[
    \hFDR_j = \EE_{H_j} \br{\frac{\one \set{j \in \cR}}{R} \mid \vct S_j } \cdot \frac{\one\{p_j > \zeta\}}{1-\zeta},
    \]
    with $p_j$ being the two-sided t-test p-value for $H_j$. Define $T_j$ as the t-test statistic with noncentrality 
    \[
    \mu_j = \frac{\theta_j}{\sigma\sqrt{(\mat X^\tran \mat X)^{-1})_{jj}}}.
    \]
    Let $t_{1-\zeta/2}$ denote the $1-\zeta/2$ upper quantile of the central t-distribution with degree of freedom $n-d$. Then the contribution to the bias of $\hFDR$ from variable $j$, when the signal strength $|\mu_j| > t_{1-\zeta/2}$, has upper bound
    \[
    \EE[\hFDR_j] \le 
    \frac{1}{(1-\zeta)} \, \br{2 \exp\!\left(-\frac{\bigl(|\mu_j|-\sqrt{2}\,t_{1-\zeta/2}\bigr)_+^2}{2}\right)
    \;+\;
    \exp\!\left(-\frac{n-d}{16}\right)}.
    \]
\end{thm}
\begin{proof}
Note
\[
    T_j \;\stackrel{d}{=}\; \frac{Z+\mu_j}{\sqrt{U/(n-d)}},
    \qquad \textnormal{where } Z\sim \cN(0,1),\ U\sim\chi^2_{n-d},\ Z\perp U.
\]
Let $A:=\{U\le 2 (n-d)\}$. Split
\[
\PP(|T|<t_{1-\zeta/2})\le \PP(|T|<t_{1-\zeta/2},\ A)+\PP(A^c).
\]
On $A$ we have $\sqrt{U/(n-d)}\le \sqrt{2}$, hence
\[
\{|T|<t_{1-\zeta/2}\}\cap A
\subseteq
\{|Z+\mu_j|<\sqrt{2}\,t_{1-\zeta/2}\}.
\]
Therefore
\begin{equation}\label{eq:main-split}
\PP(|T|<t_{1-\zeta/2})\le \PP(|Z+\mu_j|<\sqrt{2}\,t_{1-\zeta/2})+\PP(U>2(n-d)).
\end{equation}

For the Gaussian term, with $a:=\sqrt{2}\,t_{1-\zeta/2}$,
\[
\PP(|Z+\mu_j|<a)
=
\PP(-a-\mu_j<Z<a-\mu_j)
\le
\PP(Z<a-\mu_j)
\le
\exp\!\left(-\frac{(\mu_j-a)_+^2}{2}\right),
\]
and applying the same bound to $-\mu_j$ gives
\[
\PP(|Z+\mu_j|<a)\le
2\exp\!\left(-\frac{(|\mu_j|-a)_+^2}{2}\right).
\]

For the chi-square term, use a Chernoff bound,
\[
\PP(U>2(n-d))\le \exp\!\left(-\frac{n-d}{16}\right).
\]
Therefore
\[
\PP(p_j>\zeta)
=
\PP(|T_j|<t_{1-\zeta/2})
\;\le\;
2\exp\!\left(-\frac{\bigl(|\mu_j|-\sqrt{2}\,t_{1-\zeta/2}\bigr)_+^2}{2}\right)
\;+\;
\exp\!\left(-\frac{n-d}{16}\right).
\]
Note
\[
\EE_{H_j} \br{\frac{\one \set{j \in \cR}}{R} \mid \vct S_j } \leq 1 \qquad a.s.
\]
Combining all the components together, the inequality is proved.
\end{proof}

Theorem~\ref{thm:strong_signal_bias} states that the contribution to the bias of $\hFDR$ from strong signal variables decays to $0$ exponentially with the signal strength. In the proof, the bound
\[
\EE_{H_j} \br{\frac{\one \set{j \in \cR}}{R} \mid \vct S_j } \leq 1 \qquad a.s.
\]
is rather loose. In practice, we expect $R$ to be not much smaller than the number of strong signals for a reasonable variable selection procedure. Hence the bound established in Theorem~\ref{thm:strong_signal_bias} is also a practical estimate for the bias coming from all the strong signals.

\subsection{Sparse model estimation with graphical Lasso} \label{app:glasso_mle}

We start with the unconstrained MLE
\[
\mat {\tilde\Theta} := \underset{\mat \Theta}{\argmax} \; \log P_{\mat \Theta}(\mat D) = \log \det \mat \Theta - \textnormal{trace}(\widehat{\mat \Sigma} \mat \Theta)
\]
and then enforce the entries in $\widehat{\cH_0}$ to be zero
\[
\mat {\hat\Theta}_{jk} \leftarrow 
\begin{cases}
    0, & \text{if } (j,k) \in \widehat{\cH_0} \text{ or } (k,j) \in \widehat{\cH_0}\\
    \mat {\tilde\Theta}_{jk},              & \text{otherwise}
\end{cases}.
\]
It is possible that the resulting $\mat {\hat\Theta}$ is not positive definite. If that is the case, we increase the eigenvalues of $\mat {\hat\Theta}$ by the same amount until it is positive definite.

\section{Example where $\hFDR$ has non-vanishing variance} \label{app:asymp_noisy}

We describe an example where $\hFDR$ has a non-vanishing standard error as $d$ increases.

Consider a $d$-dimensional Gaussian vector $\vct Z \sim \cN(\vct \theta, \mat \Sigma)$, where $\mat \Sigma \in \RR^{d \times d}$ has diagonal $1$ and off-diagonal $\rho = 0.8$. So all $Z_j$ variables are positively equicorrelated: $\textnormal{cov}(Z_j, Z_k) = \rho$ for $j \neq k$. To select the signal variables with $\theta_j \neq 0$, we do a simple one-sided thresholding, 
\[
\cR = \set{j: \; Z_j \geq 0.5 \cdot \sqrt{2 \log d}},
\]
and set $\psi_j = \one \set{p_j > 0.5}$, where $p_j$ is the one-sided marginal z-test p-value. Let $\vct S_j = \vct Z_{-j} - \mat \Sigma_{-j,j} \cdot Z_j$ and $\vct \theta = \vct 0$. Figure \ref{fig:std_conterexample} shows the standard deviation of $\hFDR$ does not converge to $0$ as the number of variables $d$ approach infinity. 
\begin{figure}[!tb]
    \centering
    \includegraphics[width = 0.6\linewidth]{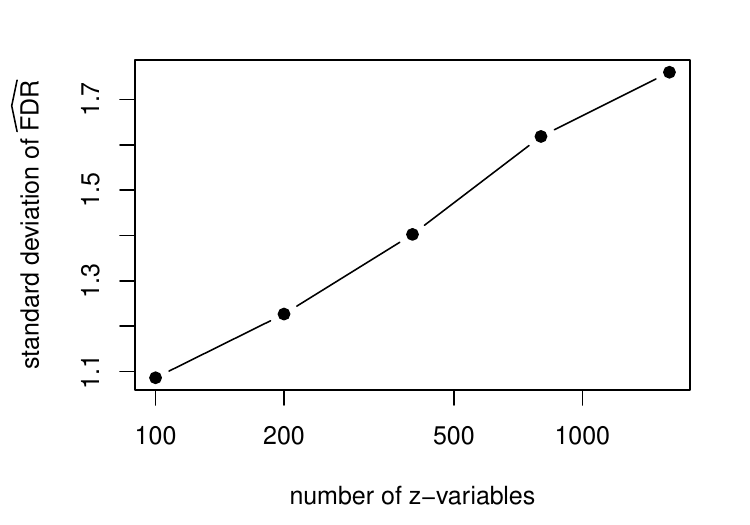}
    \caption{For variable selection on strongly positively equicorrelated $Z$-variables, standard deviation of $\hFDR$ does not converge to $0$ as the number of variables approach infinity.}
    \label{fig:std_conterexample}
\end{figure}
The noisy behavior of $\hFDR$ is a consequence of the strong equicorrelation between the $Z$ variables. Due to the correlation, perturbing $Z_j$ strongly affects all other $Z$ variables and hence the entire selection set. As a result, $\hFDR_j$ are highly positively correlated. The backbone idea for the concentration laws, which builds upon the cancellation of uncorrelated/independent noise, does not apply.

\section{Fast algorithms for computing $\hFDR^*_j$} \label{app:fast_alg}

In this section, we describe the fast and exact algorithms for computing $\hFDR^*_j$ for the Lasso (Example~\ref{ex:algo_lasso}) and FS (Example~\ref{ex:algo_fs}) procedures in the Gaussian linear model.
Our goal is to compute
\[
\hFDR^*_j
= \EE_{H_j} \br{\frac{\one \set{j \in \cR}}{R} \mid \vct S_j }.
\]
To avoid confusion, we use $j$ to refer to the particular index for which we want to compute $\hFDR^*_j$ and use $i$ to refer to a generic variable index
% \[
% \hFDR^*_j
% = \EE_{H_j} \br{\frac{\one \set{j \in \cR}}{R} \mid \vct S_j }
% = \int_{\vct Z} \frac{\one \set{j \in \cR(\lambda; \vct Z)}}{R(\lambda; \vct Z)} \cdot \PP(d \vct Z \mid \vct S_j(\vct Y)),
% \]
% where $\vct Y$ denotes the observed response vector and $\vct Z$ denotes a generic response vector drawn from the conditional distribution.
% To avoid confusion, throughout this section,
% \begin{enumerate}
%     \item we use $j$ to refer to the particular index for which we want to compute $\hFDR^*_j$ and use $i$ to refer to a generic variable index;
%     \item we view the observed $\vct Y$, and so does $\vct S_j(\vct Y)$, as given and fixed, and the value of the generic response $\vct Z$ can vary.
% \end{enumerate}

\subsection{fast and exact calculation for Lasso}
The idea is inspired by \citet{lei2019fast}. Recall the Lasso coefficients are
\[
\vct \ltheta(\lambda) = \argmin_{\vct t} \; \frac{1}{2} \, \norm{\vct Y - \mat X \vct t}_2^2 + \lambda \norm{\vct t}_1.
\]
The first-order optimal condition is
\[
- \mat X^\tran (\vct Y - \mat X \cdot \vct \ltheta) + \vct \nu = \vct 0,
\]
where $\vct \nu \in \RR^d$ is the dual variable satisfying $\nu_i = \sgn(\ltheta_i) \cdot \lambda$ for $\ltheta_i \neq 0$ and $\nu_i \in [-\lambda, \lambda]$ otherwise. Suppose we have computed the selection set $\cR$. Then by the optimal condition, $\vct \ltheta$ has a explicit expression
\begin{equation} \label{eq:lasso_path}
    \widehat{\vct \theta}_\cR^{\textnormal{lasso}} = \pth{\mat X_\cR^\tran \mat X_\cR}^{-1} \pth{\mat X_\cR^\tran \vct Y - \sgn(\widehat{\vct \theta}_\cR^{\textnormal{lasso}}) \cdot \lambda}, \quad
    \widehat{\vct \theta}_{-\cR}^{\textnormal{lasso}} = \vct 0.
\end{equation}
Now condition on $\vct S_j$ and let the value of $\vct X_j^\tran \vct Y$ vary. We can compute $\vct \ltheta$ by \eqref{eq:lasso_path} as long as $\cR$ does not change. The change happens when the computed $\vct \ltheta$ becomes critical:
\begin{enumerate}
    \item an already selected variable $i \in \cR$ gets kicked out of the selection set if and only if its computed coefficient $\ltheta_i$ become zero;
    \item an unselected variable $i \not\in \cR$ gets selected if and only if the absolute value of its dual variable $\nu_i = \mat X_i^\tran (\vct Y - \mat X \vct \ltheta)$ touches $\lambda$.
\end{enumerate}
Readers may refer to \citet{lei2019fast} for a more detailed discussion on these two conditions.

Combine all the above. We start from a computed $\cR$, let the value of $\vct X_j^\tran \vct Y$ vary (say increase), and calculate $\vct \ltheta$ by \eqref{eq:lasso_path}. Once either of the two conditions is satisfied, we update $\cR$ and move forward. We perform a low-rank update on $\pth{\mat X_\cR^\tran \mat X_\cR}^{-1}$ in \eqref{eq:lasso_path} by the Sherman–Morrison–Woodbury formula. This algorithm reveals the exact Lasso solution path. To compute $\hFDR^*_j$, we only care about the points of change of $\cR$, which can be found by solving simple linear equations derived from the two conditions. The calculation of $\hFDR^*_j$ then becomes a one-dimensional integral of the known piecewise constant function $\one \set{j \in \cR}/{R}$ on a known probability measure.

\subsection{fast and exact calculation for FS}
In the Gaussian linear model, FS first centers and normalizes the observed variables
\[
\vct Y^{(0)} \leftarrow \vct Y - \frac 1n (\vct Y \cdot \vct 1) \cdot \vct 1, \quad
\vct X_j^{(0)} \leftarrow \vct X_j - \frac 1n (\vct X_j \cdot \vct 1) \cdot \vct 1, \quad
\vct X_j^{(0)} \leftarrow \vct X_j^{(0)} / \norm{\vct X_j^{(0)}}_2, \quad \forall j.
\]
Then do the following operations in each step.
\begin{enumerate}
    \item Selection: $J^{(\lambda)} = {\argmax}_{j \not\in \cR(\lambda-1)} \; |\vct X_j^{(\lambda-1)} \cdot \vct Y^{(\lambda-1)}|$;
    \item Gram-Schmidt orthogonalization: $\vct Y^{(\lambda)} \leftarrow \vct Y^{(\lambda-1)} - (\vct Y^{(\lambda-1)} \cdot \vct X_{j^{(\lambda)}}^{(\lambda-1)}) \cdot \vct X_{j^{(\lambda)}}^{(\lambda-1)}$ and $\vct X_j^{(\lambda)} \leftarrow \vct X_j^{(\lambda-1)} - (\vct X_j^{(\lambda-1)} \cdot \vct X_{j^{(\lambda)}}^{(\lambda-1)}) \cdot \vct X_{j^{(\lambda)}}^{(\lambda-1)}$ for all $j \not\in \cR(\lambda)$.
    \item renormalization: $\vct X_j^{(\lambda)} \leftarrow \vct X_j^{(\lambda)} / \norm{\vct X_j^{(\lambda)}}_2$ for all $j \not\in \cR(\lambda)$.
\end{enumerate}    
Since we select one variable at each step, the number of selections $R$ is known. Therefore, to compute $\hFDR^*_j$, it suffices to calculate $\PP_{H_j}(j \in \cR \mid S_j)$.

Now condition on $\vct S_j$ and let the value of $\vct X_j \cdot \vct Y$ vary. Consider an auxiliary FS procedure where we enforce that variable $X_j$ is not to be selected:
\[
J^{(\lambda)} = {\argmax}_{l \not\in \cR(\lambda-1) \cup \set{j}} \; |\vct X_l^{(\lambda-1)} \cdot \vct Y^{(\lambda-1)}|.
\]
This auxiliary procedure is determined by the value of $\vct S_j$ without looking at $\vct X_j \cdot \vct Y$. At each step, if
\[
|\vct X_j^{(\lambda-1)} \cdot \vct Y^{(\lambda-1)}| > |\vct X_{J^{(\lambda)}}^{(\lambda-1)} \cdot \vct Y^{(\lambda-1)}|.
\]
then the variable with index $j$, instead of $J^{(\lambda)}$, will be selected. By this inequality and certain linear transforms, we obtain the set of values of $\vct X_j \cdot \vct Y$ such that variable $X_j$ is selected at each step. Taking a union of them and using the known distribution of $\vct X_j \cdot \vct Y$, we can compute $\PP_{H_j}(j \in \cR \mid S_j)$. The algorithm requires one auxiliary FS process (shared among many different $\hFDR^*_j$) and produces a fast and exact calculation. 
Similar ideas apply to the model-X settings.

\section{Variance of FDR estimation} \label{app:proofs}

In this section, we assume Assumption~\ref{ass:block-orth} and prove the propositions in Section~\ref{sec:se_bound}. Throughout the section, we use $\cR(\lambda)$ to denote the selection set of Lasso. Define $\cR_b = \cR \cap (b)$ as the selected variables in block $(b)$ and $R_b = |\cR_b|$ as its cardinality. Throughout the section, we consider any pair of $(j, b)$ such that $j \in (b)$. We start with showing two technical lemmas.

\begin{lemma} \label{lem:block-suff}
$\cR_b$ depends on $\vct Y$ only through $\mat X_{(b)}^\tran \vct Y$. The p-value $p_j$ depends on $\vct Y$ only through $(\mat X_{(b)}^\tran \vct Y, \RSS)$, where $\RSS = \norm{\vct Y}^2 - \norm{\mat X (\mat X^\tran \mat X)^{-1} \mat X^\tran \vct Y}^2$.
\end{lemma}
\begin{proof}
Let $\vct Y^{(b)}$ be the projection of $\vct Y$ onto the subspace spanned by the columns of $\mat X_{(b)}$, and denote the projection of $\vct Y$ onto the columns of $\mat X$ as $\widehat{\vct Y} = \mat X (\mat X^\tran \mat X)^{-1} \mat X^\tran \vct Y$. Then
\[
\widehat{\vct Y} = \sum_{b=1}^B \vct Y^{(b)}
\]
and the $\vct Y^{(b)}$s are orthogonal to each other under Assumption~\ref{ass:block-orth}. We rewrite the Lasso coefficients
\begin{align*}
    \vct \ltheta(\lambda) 
    &= \argmin_{\vct t} \; \frac{1}{2} \, \norm{\vct Y - \mat X \vct t}_2^2 + \lambda \norm{\vct t}_1 \\
    &= \argmin_{\vct t} \; \frac{1}{2} \, \norm{\widehat{\vct Y} - \mat X \vct t}_2^2 + \lambda \norm{\vct t}_1 \\
    &= \argmin_{\vct t} \; \frac{1}{2} \, \norm{\sum_{b=1}^B \vct Y^{(b)} - \mat X_{(b)} \vct t_{(b)}}_2^2 + \lambda \sum_{b=1}^B \norm{\vct t_{(b)}}_1 \\
    &= \argmin_{\vct t} \; \sum_{b=1}^B \frac{1}{2} \, \norm{\vct Y^{(b)} - \mat X_{(b)} \vct t_{(b)}}_2^2 + \lambda \norm{\vct t_{(b)}}_1.
\end{align*}
We see solving for $\ltheta$ is equivalent to solving for Lasso coefficients for $B$ subproblems of regressing $\vct Y$ on $\mat X_{(b)}$. So $\ltheta_{(b)}(\lambda)$ depends on $\vct Y$ only through $\mat X_{(b)}^\tran \vct Y^{(b)}$, which is equal to $\mat X_{(b)}^\tran \vct Y$. Therefore, $\cR_b = \set{j \in (b):\; \ltheta_j(\lambda) \neq 0}$ depends on $\vct Y$ only through $\mat X_{(b)}^\tran \vct Y$.

Recall the p-value $p_j$ is a function of the t-statistic
\[
\frac{\pth{(\mat X^\tran \mat X)^{-1}}_j \cdot \mat X^\tran \vct Y}{\sqrt{\pth{(\mat X^\tran \mat X)^{-1}}_{jj} \cdot \RSS / (n-d)}}.
\]
The denominator depends on $\vct Y$ only through $\RSS$. For the numerator, note $\mat X^\tran \mat X$ is block-diagonal by the orthogonality in Assumption~\ref{ass:block-orth}. We have
\[
\pth{(\mat X^\tran \mat X)^{-1}}_j \cdot \mat X^\tran \vct Y = \pth{(\mat X_{(b)}^\tran \mat X_{(b)})^{-1}}_{i(j,b)} \cdot \mat X_{(b)}^\tran \vct Y,
\]
where $i(j,b)$ is the index of variable $X_j$ in block $(b)$. Therefore $p_j$ depends on $\vct Y$ only through $(\mat X_{(b)}^\tran \vct Y, \RSS)$.
\end{proof}

\begin{lemma}[mean value] \label{lem:mean-value}
Suppose $f(x) \in [a, b] \subseteq \RR^+$ for $x \in D$ and $1/f$ is integrable with respect to the measure $\mu$. Then there exists some $c \in [a, b]$ such that
\[ \int_{x \in D} 1/f(x) \, \mu(d x) = 1/c \cdot \mu(D). \]
\end{lemma}
\begin{proof}
Define 
\[
c = \frac{\mu(D)}{\int_{x \in D} 1/f(x) \, \mu(d x)}.
\]
Since $f(x) \in [a, b] \subseteq \RR^+$, we have $c \in [a, b]$ satisfies the conditions.
\end{proof}

Now we are ready to show the propositions in Section~\ref{sec:se_bound}.

\begin{proof}[Proof of Proposition \ref{prop:signals-lasso}]
Recall in the Gaussian linear model, we condition on the explanatory variables $\mat X$ and treat them as fixed. By Lemma~\ref{lem:block-suff} and the block-orthogonality, $R_b$ are independent for $b=1, \ldots, B$. Our goal is to show $R = \sum_b R_b$ is stochastically larger than the sum of certain independent Bernoulli random variables. Then the proposition follows by the concentration laws.

\emph{Part 1.} By Lemma~\ref{lem:block-suff} and the first order condition, $R_b \geq 1$ if and only if there exist some $j \in (b)$ such that $|\vct X_j^\tran \vct Y| \geq \lambda$. We have
\[
\PP(R_b \geq 1) =
\PP(\bigcup_{j \in (b)} \set{|\vct X_j^\tran \vct Y| \geq \lambda})
\geq \PP(|\vct X_j^\tran \vct Y| \geq \lambda)
% \geq 2 (1 - \Phi(\lambda)).
\geq 2 (1 - \Phi(\lambda / \sigma)).
\]
The last inequality is because that $\vct X_j^\tran \vct Y$ is a Gaussian random variable with variance $\sigma^2$ (note $\mat X$ is normalized). Direct calculation shows the probability achieves minimum when $\vct X_j^\tran \vct Y$ has mean zero.

As a consequence, $R = \sum_{b=1}^B R_b$ is stochastically larger than $\sum_{b=1}^B A_b$, where $A_b$ are $B$ independent Bernoulli random variables with success probability $\delta$. By the multiplicative Chernoff bound, we have
\[
\PP(R \leq B q)
\;\leq\; \PP \pth{\sum_{b=1}^B A_b \leq \pth{1 - \frac{\delta-q}{\delta}} B \delta}
\;\leq\; \exp \pth{- \frac{(\delta - q)^2}{2 \delta} \cdot B}.
\]

\emph{Part 2.} To avoid triviality, assume $B \xi$ is an integer. Let $b$ denote any one of the $B \xi$ blocks that contains a variable with $|\vct X_j \cdot \vct \mu| \geq \lambda$. Note $\vct X_j^\tran \vct Y$ is Gaussian with mean $\vct X_j \cdot \vct \mu$. By a similar reasoning as in Part 1,
\[
\PP(R_b \geq 1) \;\geq\; \PP(|\vct X_j^\tran \vct Y| \geq \lambda) \;\geq\; \frac 12.
\]
As a consequence, $R = \sum_{b=1}^B R_b$ is stochastically larger than $\sum_{b=1}^{B\xi} A_b$, where $A_b$ are $B\xi$ independent Bernoulli random variables with success probability $1/2$. By the multiplicative Chernoff bound, we have
\[
\PP(R \leq B q)
\;\leq\; \PP \pth{\sum_{b=1}^{B\xi} A_b \leq \pth{1 - \frac{\xi-2q}{\xi}} \frac{B \xi}{2}}
\;\leq\; \exp \pth{- \frac{(\xi - 2q)^2}{4 \xi} \cdot B}.
\]
\end{proof}

\begin{proof}[Proof of Theorem \ref{thm:variance}]

Notice in \eqref{eq:lasso}, if we scale both $\vct Y$ and $\lambda$ by $1/\sigma$, the optimal solution $\vct \ltheta$ will be scaled by $1/\sigma$ as well, but the selected variables do not change. Hence in the proof, we assume $\sigma^2 = 1$ without loss of generality. Dependence on $\sigma^2$ can be restored by rescaling $\vct Y$ and $\lambda$.

Let $\mat X_{(b)}$ be the columns of $\mat X$ whose indices are in $(b)$ and $\mat X_{-(b)}$ be its complement. Let $\mat X_{(b); -j}$ be the same matrix as $\mat X_{(b)}$ but with $\vct X_j$ excluded.

Under Assumption \ref{ass:block-orth},
\begin{align*}
    \hFDR^*_j &= \EE_{H_j} \br{\frac{\one \set{j \in \cR}}{|\cR \cup \set j|} \mid \vct S_j }  %  \cdot \frac{\one \set{p_j > 0.5}}{0.5}
    = \EE_{H_j} \br{\frac{\one \set{j \in \cR_b}}{|\cR_b \cup \set j| + \sum_{l \neq b} R_l} \mid \mat X_{(b); -j}^\tran \vct Y, \mat X_{-(b)}^\tran \vct Y, r_j^2 },
\end{align*}
where $r_j^2 = \RSS + u_j^2$ and $u_j$ is the norm of the projection of $\vct Y$ on to the direction of $(\mat I - \mat X_{-j} (\mat X_{-j}^\tran \mat X_{-j})^{-1} \mat X_{-j}^\tran) \vct X_j$. By Lemma \ref{lem:block-suff}, $\sum_{l \neq b} R_l$ depends on $\vct Y$ only through $\mat X_{-(b)}^\tran \vct Y$ and $\cR_b$ depends on $\vct Y$ only through $\mat X_{(b)}^\tran \vct Y$. 
% Note one can recover $\mat X_{(b)}^\tran \vct Y$ from $\mat X_{(b); -j}^\tran \vct Y$ and $u_j$, 
The conditional expectation $\hFDR^*_j$ is a one-dimensional integration over $u_j$.
To avoid confusion, we use $u_j$ to denote the value corresponding to the observed $\vct Y$ and $v_j$ to denote a placeholder for $u_j$ in the integral. We write
\[ \hFDR^*_j = \int_{\set{v_j:\; j \in \cR_b}} \frac{1}{f(v_j ; \mat X_{(b); -j}^\tran \vct Y, \mat X_{-(b)}^\tran \vct Y)} \, \mu_j(d v_j \mid \vct S_j), \]
where
\[
f(v_j; \mat X_{(b); -j}^\tran \vct Y, \mat X_{-(b)}^\tran \vct Y) = |\cR_b \cup \set j| + \sum_{l \neq b} R_l
\]
and
\[
\mu_j(A \mid \vct S_j) = \PP_{H_j} (v_j \in A \mid \mat X_{(b); -j}^\tran \vct Y, \mat X_{-(b)}^\tran \vct Y, r_j^2) = \PP_{H_j} (v_j \in A \mid \mat X_{(b); -j}^\tran \vct Y, r_j^2).
\]
We suppress $\vct S_j$ and write $\mu_j(A)$ for simplicity of notation. Note the measure $\mu_j$ depends on $\vct Y$ only through $(\mat X_{(b); -j}^\tran \vct Y, r_j^2)$.
Note 
\[ f(v_j) \in \br{1 + \sum_{l \neq b} R_l, \; m + \sum_{l \neq b} R_l} \subseteq \br{1 \mam \left(-m +  \sum_{b=1}^B R_b\right), \; m+ \sum_{b=1}^B R_b}. \]
By Lemma \ref{lem:mean-value}, for some $\gamma_j \in [-m, m]$, we have
\[\hFDR^*_j = \frac{\mu_j(j \in \cR_b)}{1\vee (R + \gamma_j)}.
\]
Therefore
\[
\abs{\hFDR^*_j - \frac{\mu_j(j \in \cR_b)}{1 \mam R}}
= \abs{\frac{\mu_j(j \in \cR_b)}{1\mam (R + \gamma_j)} - \frac{\mu_j(j \in \cR_b)}{1 \mam R}}
\leq \frac{m}{1 \mam R^2} \cdot \frac{1 \mam R}{1\mam (R - m)}
\leq \frac{4 m^2}{1 \mam R^2}.
\]

For notational convenience, we suppress taking maximum with $1$ for $R$. Throughout the section, we have $R$ represents $1 \mam R$ and $\sum_b R_b$ represents $1 \mam \sum_b R_b$.
Then we can write
\[
\hFDR^*_j = \frac{\mu_j(j \in \cR_b)}{\sum_b R_b} + O \pth{\frac{m^2}{R^2}}.
\]
Define
\[ T_b = \sum_{j \in (b)} \mu_j(j \in \cR_b) \cdot \frac{\psi_j(p_j)}{\EE [\psi_j(p_j) \mid \vct S_j]}. \]
We have
\[ \sum_{j \in (b)} \hFDR_j = \frac{T_b}{\sum_b R_b} + O \pth{\frac{c_\phi m^3}{R^2}} \]
and
\[ \hFDR = \sum_{b=1}^B \sum_{j \in (b)} \hFDR_j = \frac{\sum_b T_b}{\sum_b R_b} + O \pth{\frac{B c_\phi m^3}{R^2}} \]
Note we can compute $(\mat X_{(b); -j}^\tran \vct Y, r_j^2)$ from $(\mat X_{(b)}^\tran \vct Y, \RSS)$ for any $j \in (b)$. Moreover, by Lemma \ref{lem:block-suff}, for any $j \in (b)$, $p_j$ depends on $Y$ only through $(\mat X_{(b)}^\tran \vct Y, \RSS)$. So $T_b$ depends on $Y$ only through $(\mat X_{(b)}^\tran \vct Y, \RSS)$. Therefore, condition on $\RSS$, both $\set{T_b:\; b \in [B]}$ and $\set{R_b:\; b \in [B]}$ are sets of conditional independent random variables respectively. We denote $T = \sum_b T_b$ for simplicity.

Next we prove the theorem in three steps.

\emph{Step 1.} $\Var\pth{{\sum_b T_b}/{\sum_b R_b} \mid \RSS} \le C(c_\phi^2 m^4 / B q^4 + c_\phi^2 m^2 B^2 \rho)$ almost surely for some $C > 0$ that does not depend on $\RSS$.

For simplicity of notation, we denote $\Var(\cdot | \cdot, \RSS)$ as $\Var_\RSS(\cdot | \cdot)$ and $\EE(\cdot | \cdot, \RSS)$ as $\EE_\RSS(\cdot | \cdot)$. Let $q$ be the number in Proposition \ref{prop:signals-lasso}. We have
\begin{align*}
    \Var_\RSS \br{\frac{T}{R}}
    =& \frac 12 \, \EE_\RSS \pth{\frac{T}{R} - \frac{T'}{R'}}^2 \\
    =& \frac 12 \, \EE_\RSS \br{ \pth{\frac{T}{R} - \frac{T'}{R'}}^2 \mid R,R' > B q } \cdot \PP_\RSS (R, R' > B q) \\
    &+ \frac 12 \, \EE_\RSS \br{ \pth{\frac{T}{R} - \frac{T'}{R'}}^2 \mid R \mim R' \leq B q } \cdot \PP_\RSS (R \mim R' \leq B q),
\end{align*}
where $(T', R')$ is an IID copy of $(T, R)$ conditional on $\RSS$.
% By the law of total variance,
% \begin{align*}
%     \Var_\RSS \br{\frac{\sum_b T_b}{\sum_b R_b}}
%     =& \Var_\RSS \br{\frac{\sum_b T_b}{\sum_b R_b} \mid R > B q} \PP_\RSS (R > B q) \\
%     &+ \Var_\RSS \br{\frac{\sum_b T_b}{\sum_b R_b} \mid R \leq B q} \PP_\RSS (R \leq B q) \\
%     &+ \Var_\RSS \br{\EE_\RSS \br{\frac{\sum_b T_b}{\sum_b R_b} \mid \one \set{R > B q}}}.
% \end{align*}
We will control each term separately. First,
\begin{align*}
    & \frac 12 \, \EE_\RSS \br{ \pth{\frac{T}{R} - \frac{T'}{R'}}^2 \mid R,R' > B q } \cdot \PP_\RSS (R, R' > B q) \\
    = & \frac 12 \, \EE_\RSS \br{ \pth{\frac{T R' - T'R}{R R'}}^2 \mid R,R' > B q } \cdot \PP_\RSS (R, R' > B q) \\
    \leq & \frac{1}{2 B^4 q^4} \cdot \EE_\RSS \br{ \pth{T R' - T'R}^2 \mid R,R' > B q } \cdot \PP_\RSS (R, R' > B q) \\
    \leq & \frac{1}{2 B^4 q^4} \cdot \EE_\RSS \br{ \pth{T R' - T'R}^2} \\
    = & \frac{1}{2 B^4 q^4} \cdot \pth{2 \Var_\RSS \br{T R'} - 2 \pth{\EE_\RSS^2 \br{T R} - \EE_\RSS^2 \br{T} \EE_\RSS^2 \br{R}}} \\
    \leq & \frac{1}{B^4 q^4} \cdot \pth{\Var_\RSS \br{T R'} + 2 \EE_\RSS \br{T} \EE_\RSS \br{R} \cdot |\text{Cov}_\RSS(T, R)|}
\end{align*}
The last inequality is obtained by considering cases where $\text{Cov}_\RSS(T, R)$ is positive or negative separately. Note $T$ and $R'$ are conditionally independent given $\RSS$. We have
\[
\Var_\RSS \br{T R'}
= \Var_\RSS(T) \Var_\RSS(R') + \Var_\RSS(T) (\EE_\RSS R')^2 + \Var_\RSS(R') (\EE_\RSS T)^2
\]
Note $T_b \in [0, c_\phi m]$, we have $\EE_\RSS [T] \leq B c_\phi m$ and $\Var_\RSS [T] = \sum_b \Var_\RSS [T_b] \leq B c_\phi^2 m^2$. Similarly, $\EE_\RSS [R'] \leq B m$ and $\Var_\RSS [R'] \leq B m^2$. Thus 
\[ \Var_\RSS \br{T R'} \leq B^2 c_\phi^2 m^4 + B^3 c_\phi^2 m^4 = O(m^4 B^3)\]
and
\[
|\text{Cov}_\RSS(T, R)| \leq \sqrt{\Var_\RSS [T] \Var_\RSS [R]} \leq B m^2 c_\phi.
\]
Therefore the first term is bounded by $O(m^4 / B)$. For the second term,
\begin{align*}
    & \frac 12 \, \EE_\RSS \br{ \pth{\frac{T}{R} - \frac{T'}{R'}}^2 \mid R \mim R' \leq B q } \\
    \leq& \frac 12 \, \EE_\RSS \br{ \pth{\frac{T}{R} + \frac{T'}{R'}}^2 \mid R \mim R' \leq B q } \\
    \leq& \frac 12 \, \EE_\RSS \br{ \pth{T + T'}^2 \mid R \mim R' \leq B q } \\
    \leq& 2 B^2 c_\phi^2 m^2.
\end{align*}
The last inequality is because $T \leq B c_\phi m$. By Proposition \ref{prop:signals-lasso}, we have
\[
\PP_\RSS (R \mim R' \leq B q)
\leq 2 \PP_\RSS (R \leq B q)
\leq 2 \rho.
\]
Combine the results above, we have
\[
\Var_\RSS \br{\frac{T}{R}}
\leq c_\phi^2  \cdot \br{\frac{B^2 m^4 + 5 B^3 m^4}{B^4 q^4} + 2 B^2 m^2 \cdot 2 \rho}.
\]
The bound is of $O \pth{{m^4}/{B}}$ if $\rho = O(B^{-3})$.

\emph{Step 2.} $\Var\br{\EE \pth{ {\sum_b T_b}/{\sum_b R_b} \mid \RSS }} = O(1/(n-d))$.

We start by bounding the perturbation of $\mu_j(j \in \cR_b)$ when we perturb $\RSS$ while holding $\mat X^\tran \vct Y$ fixed. We write 
\[
\mu_j(j \in \cR_b) = \PP_{H_j} (v_j \in A \mid \mat X_{(b); -j}^\tran \vct Y, r_j^2)
= \int_{-r_j}^{r_j} I_j(v_j; \mat X_{(b); -j}^\tran \vct Y) \cdot q_j(v_j) \, d v_j,
\]
where $r_j = \sqrt{\RSS + u_j^2}$, $I_j(v_j; X_{(b); -j}^\tran \vct Y) = \one \set{j \in \cR_b}$, and $q_j = d\mu_j / dv_j$ is the probability density. Direct calculation shows
\[
q_j(v_j) = c_{n-d} \cdot \sqrt{n-d} \cdot \pth{1 - \frac{v_j^2}{r_j^2}}^{\frac{n-d-2}{2}}\cdot \frac{1}{r_j}, \quad\text{where}\quad 
c_{n-d} = \frac{\Gamma\left(\frac{n-d+1}{2}\right)}{\sqrt{(n-d) \pi} \, \Gamma\left(\frac{n-d}{2}\right)}.
\]
Note that $q_j(r_j) = q_j(-r_j) = 0$. By the bounded convergence theorem,
\[
\partial_\RSS \, \mu_j(j \in \cR_b) = 
\int_{-r_j}^{r_j} I_j(v_j; X_{(b); -j}^\tran \vct Y) \cdot \partial_\RSS \, q_j(v_j) \, d v_j,
\]
where
\[
\partial_\RSS \, q_j = ((n-d-1) v_j^2 - r_j^2) \cdot \frac{c_{n-d} \sqrt{n-d}}{2} \cdot (r_j^2 - v_j^2)^{\frac{n-d-4}{2}} \cdot (r_j^2)^{-\frac{n-d+1}{2}}.
\]
Note that $\partial_\RSS \, q_j \geq 0$ if and only if $v_j^2 \geq \frac{r_j^2}{n-d-1}$ and 
\[\int_{-r_j}^{r_j}\partial_\RSS q_j dv_j = \partial_\RSS\int_{-r_j}^{r_j} q_j dv_j = \partial_\RSS 1 = 0.\]
Since $I_j(v_j; X_{(b); -j}^\tran \vct Y) \in \set{0, 1}$, we have
\begin{align*}
|\partial_\RSS \, \mu_j(j \in \cR_b)| &\leq 
\max \set{ \int_{r_j^2 \geq v_j^2 \geq \frac{r_j^2}{n-d-1}} \partial_\RSS \, q_j \, d v_j, \; -\int_{v_j^2 \leq \frac{r_j^2}{n-d-1}} \partial_\RSS \, q_j \, d v_j } \\
& = 
\int_{r_j^2 \geq v_j^2 \geq \frac{r_j^2}{n-d-1}} \partial_\RSS \, q_j \, d v_j\\
&= c_{n-d} \cdot \frac{\sqrt{n-d}}{\sqrt{n-d-1}} \pth{1 - \frac{1}{n-d-1}}^{\frac{n-d-2}{2}} \cdot \frac{1}{r_j^2} \\
&\leq \frac{1}{\RSS}.
\end{align*}
The last equation is calculated by \citet{wolfram19math} and the last inequality is by $c_{n-d} \leq 1/2$ and $r_j^2 \geq \RSS$.

Then let us look at the effect of perturbing $\RSS$ on the two-sided t-test $p$-values. We have
\[
p_j = 2 \cdot F_{t_{n-d}}\pth{- \frac{|v_j|}{\sqrt{\RSS/(n-d)}}},
\]
where $F_{t_{n-d}}$ is the cumulative distribution function of the t-distribution of degree of freedom $n-d$. Then
\[
\partial_\RSS \, p_j
= c_{n-d} \cdot \pth{1 + \frac{v_j^2}{\RSS}}^{-\frac{n-d+1}{2}} \cdot \frac{|v_j|}{\sqrt{\RSS/(n-d)}} \cdot \frac{1}{\RSS}.
\]

Now consider the event $A_\RSS = \set{ 0.5 \, (n-d) \leq \RSS \leq 1.5 \, (n-d)}$. Recall we assume $\sigma^2 = 1$. Since $\RSS \sim\chi^2_{n-d}$ is sub-exponential, we have $\PP(A^\setcomp) \leq 2 \exp(-(n-d) / 32)$ by the Bernstein concentration inequalities.

On $A_\RSS$, we have
\[
|\partial_\RSS \, \mu_j(j \in \cR_b)| \leq \frac{2}{n-d}
\]
and
\begin{align*}
\partial_\RSS \, p_j
&\leq c_{n-d} \cdot \pth{1 + \frac{v_j^2 / 3 }{(n-d)/2}}^{-\frac{n-d}{2}} \cdot |v_j| \cdot \sqrt{2} \cdot \frac{1}{n-d} \\
&\leq c_{n-d} \cdot \frac{|v_j|}{1 + v_j^2 / 3} \cdot \sqrt{2} \cdot \frac{1}{n-d} \\
&\leq c_{n-d} \cdot \frac{\sqrt{3}}{2} \cdot \sqrt{2} \cdot \frac{1}{n-d} \\
&\leq \frac{1}{n-d}.
\end{align*}
Note both $\mu_j(j \in \cR_b)$ and $\phi_j(p_j)$ are bounded. Therefore, $\mu_j(j \in \cR_b) \cdot \phi_j(p_j)$ is Lipschitz in $\RSS$ with a Lipschitz constant
\[
\pth{2 c_\phi + L_\phi} \cdot \frac{1}{n-d}.
\]

Moreover, $\sum_b T_b = \sum_{j=1}^{d} \mu_j(j \in \cR_b) \cdot \phi_j(p_j)$ is then Lipschitz in $\RSS$ with a Lipschitz constant $B m \pth{2 c_\phi + L_\phi} / (n-d)$.
A key observation is that ${\sum_b T_b}/{R}$ is Lipschitz in $\RSS$ since $R$ is not a function of $\RSS$.

Now consider the event $A_R = \set{R > B q}$. Since $\RSS$ is independent of $\mat X^\tran \vct Y$, $A_\RSS$ and $A_R$ are independent. In addition, we have $\PP(A_R^\setcomp) \leq \rho$.  % \exp(-2 q^2 B)

On $A_\RSS \cap A_R$, we have ${\sum_b T_b}/{R}$ is $L$-Lipschitz in $\RSS$ with
\[
L = \pth{\frac{2 m c_\phi + m L_\phi}{q}} \cdot \frac{1}{n-d}.
\]
Therefore, for any $\RSS \in A_\RSS$,
\[
\EE \br{\frac{\sum_b T_b}{R} \mid \RSS, A_R}
\]
is $L$-Lipschitz in $\RSS$.

For notational convenience, we denote
\[
E_R(\RSS) = \EE \br{ \frac{\sum_b T_b}{R} \mid \RSS, A_R}, \quad
E_{R^\setcomp}(\RSS) = \EE \br{ \frac{\sum_b T_b}{R} \mid \RSS, A_R^\setcomp},
\]
and
\[
E(\RSS) = \EE \br{ \frac{\sum_b T_b}{R} \mid \RSS}
= E_R(\RSS) \cdot \PP(A_R) + E_{R^\setcomp}(\RSS) \cdot \PP(A_R^\setcomp).
\]
Note all $E_R(\RSS)$, $E_{R^\setcomp}(\RSS)$, and $E(\RSS)$ are bounded by $B c_\phi m$.

Let $\RSS'$ be an independent copy of $\RSS$. Then
\begin{align*}
& (E(\RSS) - E(\RSS'))^2 \\
= & (E_R(\RSS) - E_R(\RSS'))^2 \cdot \PP(A_R)^2 + (E_{R^\setcomp}(\RSS) - E_{R^\setcomp}(\RSS'))^2 \cdot \PP(A_R^\setcomp)^2 + \\
& \quad 2 (E_R(\RSS) - E_R(\RSS')) (E_{R^\setcomp}(\RSS) - E_{R^\setcomp}(\RSS')) \cdot \PP(A_R) \PP(A_R^\setcomp) \\
= & (E_R(\RSS) - E_R(\RSS'))^2 + O(c_\phi^2 m^2 \cdot B^2 \rho )
\end{align*}
Therefore,
\begin{align*}
\Var[E(\RSS)] =& \frac 12 \EE \br{(E(\RSS) - E(\RSS'))^2} \\
=& \frac 12 \EE \br{(E(\RSS) - E(\RSS'))^2 \mid A_\RSS} \cdot \PP(A_\RSS) + \\
& \quad \frac 12 \EE \br{(E(\RSS) - E(\RSS'))^2 \mid A_\RSS^\setcomp} \cdot \PP(A_\RSS^\setcomp) \\
= & \frac 12 \EE \br{(E_R(\RSS) - E_R(\RSS'))^2 \mid A_\RSS} \cdot \PP(A_\RSS) \\
& \quad + O(c_\phi^2 m^2 \cdot B^2 \rho) + O \pth{c_\phi^2 m^2 \cdot B^2 e^{-{(n-d)}/{32}}}
\end{align*}
Since $E_R(\RSS)$ is $L$-Lipschitz in $\RSS$,
\begin{align*}
& \frac 12 \EE \br{(E_R(\RSS) - E_R(\RSS'))^2 \mid A_\RSS} \cdot \PP(A_\RSS) \\
\leq & \frac 12 \EE \br{L^2 (\RSS - \RSS')^2 \mid A_\RSS} \cdot \PP(A_\RSS) \\
\leq & \frac 12 \EE \br{L^2 (\RSS - \RSS')^2} \\
\leq & L^2 \Var(\RSS) \\
= & \pth{\frac{2 m c_\phi + m L_\phi }{q }}^2 \cdot \frac{2}{n-d},
% = & \pth{\frac{2 m c_\phi + m L_\phi \sigma^2}{q \sigma^2}}^2 \cdot \frac{2}{n-d},
\end{align*}
where the last line uses the fact that $\RSS \sim \chi^2(n-d)$ and hence $\Var(\RSS) = 2(n-d)$. Then we have
\[
\Var \br{E(\RSS)} \leq \pth{\frac{2 m c_\phi + m L_\phi}{q}}^2 \cdot \frac{2}{n-d} + O(c_\phi^2 m^2 \cdot B^2 \rho) + O \pth{c_\phi^2 m^2 \cdot B^2 e^{-(n-d)/32}}.
\]
That is
\[
\Var \br{\EE \pth{\frac{\sum_b T_b}{R} \mid \RSS}} \leq \pth{\frac{2 m c_\phi + m L_\phi}{q}}^2 \cdot \frac{2}{n-d} + O(c_\phi^2 m^2 \cdot B^2 (\rho \mam e^{-(n-d)/32})).
\]

\vspace{0.5em}
Combining Step 1 and Step 2, we have
\begin{align*}
\Var \br{\frac{\sum_b T_b}{R}}
&= \EE \br{\Var_\RSS \pth{\frac{T}{R}}} + \Var \br{\EE \pth{\frac{\sum_b T_b}{R} \mid \RSS}} \\
&\leq O \pth{\frac{c_\phi^2  m^4}{B q^4}} + \pth{\frac{2 m c_\phi + m L_\phi}{q }}^2 \cdot \frac{2}{n-d} + O(c_\phi^2 m^2 \cdot B^2 (\rho \mam e^{-(n-d)/32}))
\end{align*}

Finally,
\begin{align*}
\Var \br{\hFDR} =& \Var \br{\frac{\sum_b T_b}{\sum_b R_b} + O \pth{\frac{B c_\phi m^3}{R^2}}} \\
\leq& 2 \Var \br{\frac{\sum_b T_b}{\sum_b R_b}} + O \pth{2 \EE \br{\frac{B^2 c_\phi^2 m^6}{R^4}}} \\
\leq& 2 \Var \br{\frac{\sum_b T_b}{\sum_b R_b}} + O \pth{\frac{c_\phi^2 m^6}{B^2 q^4}} + O(B^2 c_\phi^2 m^6 \rho) \\
\leq& O \pth{c_\phi^2 \pth{\frac{ m^4}{q^4} + \frac{m^6}{B q^4} } \cdot \frac 1 B} + \pth{\frac{2 m c_\phi + m L_\phi}{q }}^2 \cdot \frac{4}{n-d} 
\\ &
+ O(c_\phi^2 m^6 \cdot B^2 (\rho \mam e^{-(n-d)/32})).
\end{align*}
The big-O notation only hides universal constant.

\end{proof}

\begin{proof}[Proof of Corollary \ref{cor:variance}]
Result (1): The setting satisfies the conditions in Part 2 of Proposition~\ref{prop:signals-lasso} with $\xi = B^{-\ep}$ and $q = \xi/3 = O(B^{-\ep})$. Then $\rho = \exp\pth{-B^{1-\ep}/36}$. Plugging them into the bound in Theorem~\ref{thm:variance}, we have
\[
\Var(\hFDR) = O\pth{B^{-1+4\ep} + \frac{B^{2\ep}}{n-d} + B^2 \exp\pth{-B^{1-\ep}/36}} = O(B^{-1+4\ep}).
\]

Result (2): The setting satisfies the conditions in Part 2 of Proposition~\ref{prop:signals-lasso} with $q = \xi/3$, and $\rho = \exp\pth{-B\xi/36}$. Plugging them into the bound in Theorem~\ref{thm:variance}, we have
\[
\Var(\hFDR) = O\pth{B^{-1} + \frac{1}{n-d} + B^2 \exp\pth{-(n-d)/32}} = O(B^{-1}).
\]
\end{proof}

\section{Consistency of Bootstrap variance estimation} \label{app:bootstrap_consistency}

\subsection{Setup and notation}
\label{sec:asymptotic-setup}

We consider the Gaussian linear model with block-orthogonal design as in Assumption~\ref{ass:block-orth}. In particular, the $d$ columns of
$\mat X$ are partitioned into $B$ blocks, specified by the index sets $(b) \subseteq [d]$ for $b = 1, \ldots, B$. The columns in different blocks are orthogonal. Moreover, we assume that $\sigma^2$ is known throughout the section. 

We consider an asymptotic regime in which the number of blocks $B$ tends to infinity. Hence so does the number of variables:
\[
d \;=\; \sum_{b=1}^B m_b \;\to\; \infty.
\]
Formally, consider a sequence of such models indexed by $d$, with design matrices $\mat X = \mat X_d \in \RR^{n_d \times d}$ satisfying $n_d > d$ for all $d$ and $n_d \to \infty$ as $d \to \infty$, and we suppress the index $d$ from the notation when there is no ambiguity.

For each block $b \in [B] := \{1,\ldots,B\}$, let $\mat X_{(b)} \in \RR^{n \times m_b}$ denote the
submatrix of $\mat X$ formed by the columns in block $b$, let
$\vct\theta_{(b)} \in \RR^{m_b}$ be the corresponding coefficient vector, and let
\[
\Sigma_{(b)} = \mat X_{(b)}^\tran \mat X_{(b)}
\]
be the within-block Gram matrix. 
We write the sufficient statistics within block $b$ as
\[
\vct Z_{(b)} = \mat X_{(b)}^\tran \vct Y
\;\sim\; \cN\big(\Sigma_{(b)} \vct\theta_{(b)},\; \sigma^2 \Sigma_{(b)}\big).
\]

In this section, the notation becomes simpler if we index variables within a given block. 
Thus, we make a slight notational adjustment from the main text, where $j$ denotes a global index across all variables. 
Here, we fix a block $b$ and let $j \in [m_b]$ denote the \emph{within-block} index. 
For example, $\vct Z_{(b); j}$ refers to the $j$th coordinate of the vector $\vct Z_{(b)}$.

Similarly, we write $\Sigma_{(b); -j, j} \in \mathbb{R}^{m_b-1}$ for the $j$th column of $\Sigma_{(b)}$ with its $j$th entry removed,  
and $\Sigma_{(b); j, -j} = \Sigma_{(b); -j, j}^\top$ for the corresponding row.  
We also use
\[
\Sigma_{(b); -j, -j} \in \mathbb{R}^{(m_b - 1)\times (m_b - 1)}
\]
to denote the principal submatrix obtained by removing the $j$th row and $j$th column of $\Sigma_{(b)}$.

% Throughout, we assume the columns of $\mat X$ have been scaled so that
% \[
% \Sigma_{(b);jj} = 1, \qquad \forall\, b \in [B],\; j \in [m_b],
% \]

We make the following standing assumptions. In particular, the Assumption A1 is the same as Assumption 3.1 in the main text. We state it here to make the discussion self-contained. 

\begin{enumerate}[{A}1)]
\item (Column normalization and block orthogonality) $\|\vct X_{(b), j}\| = 1$, and hence $\Sigma_{(b);jj} = 1$, for all $b\in [B], j\in [m_b]$ for any $b_1 \neq b_2$,  
\[
\mat X_{(b_1)}^\tran \mat X_{(b_2)} = \mat 0_{m_{b_1}\times m_{b_2}}.
\]
% \item (Known noise level) The Gaussian noise variance $\sigma^2 \in (0,\infty)$ is known.
% For notational convenience, one may take $\sigma^2 = 1$; this can always be achieved by rescaling.
\item (Bounded block sizes)
$m_b \;\le\; m_{\max}$ for all $b\in[B]$ 
and some constant $m_{\max}=O(1)$.
\item (Within-block regularity) There exist constants
$\tau_{\min}^2,\tau_{\max}^2,\rho_1,\rho_\infty \in (0,1)$ such that, for every block
$b \in [B]$ and every $j \in [m_b]$,
\[
\lambda_{\min}(\Sigma_{(b)}) > 0,
\]
\[
\tau_{\min}^2
\;<\;
1 - \Sigma_{(b);j,-j}\,\Sigma_{(b);-j,-j}^{-1}\,\Sigma_{(b);-j,j}
\;<\;
\tau_{\max}^2,
\]
and
\[
\max_{b\in [B]}\;\max_{j\in (b)}
\big\|\Sigma_{(b);-j,-j}^{-1}\Sigma_{(b);-j,j}\big\|_1 \;<\; \rho_1,
\qquad
\max_{b\in [B]}\;\max_{\substack{j,k\in [m_b]\\ j\neq k}}
\big|\Sigma_{(b);j,k}\big| \;<\; \rho_\infty.
\]
\end{enumerate}

The next lemma provides a simple example in which Assumption~A3 holds.

\begin{lemma}\label{lem:equicorrelated}
If $\Sigma_{(b)} = (1 - \rho_{b})I_{m_b} + \rho_{b} 1_{m_b}1_{m_b}^\tran$ for some $|\rho_{b}| < 1/(m_b - 1)$, then Assumption A3 holds with any 
\[\tau_{\min}^2 < \min_{b\in [B]}(1 - \rho_{b})\frac{(m_b-1)\rho_{b} + 1}{(m_b-2)\rho_{b}+1} < \max_{b\in [B]}(1 - \rho_{b})\frac{(m_b-1)\rho_{b} + 1}{(m_b-2)\rho_{b}+1} < \tau_{\max}^2.\]
and
\[\rho_1 > \max_{b\in [B]}\bigg|\frac{(m_b - 1) \rho_{b}}{(m_b - 2)\rho_{b} + 1}\bigg|, \quad \rho_\infty > \max_{b\in [B]}|\rho_{b}|.\]
\end{lemma}

We now define the Lasso estimator. Let $\lambda = \lambda_d \to \infty$ be a regularity level that may depend on $d$. The (global) Lasso solution is
\[
\hat{\vct\theta}^\lambda
= \argmin_{\vct\gamma \in \RR^d}
\sum_{b=1}^B \bigg\{
\frac{1}{2}\,\vct\gamma_{(b)}^\tran \Sigma_{(b)} \vct\gamma_{(b)} - \vct\gamma_{(b)}^\tran \vct Z_{(b)}
\bigg\}
+ \lambda \|\vct\gamma\|_1.
\]
Equivalently, by block orthogonality,
\[
\hat{\vct\theta}^\lambda
= \argmin_{\vct\gamma \in \RR^d}
\sum_{b=1}^B \bigg\{
\frac{1}{2}\,\vct\gamma_{(b)}^\tran \Sigma_{(b)} \vct\gamma_{(b)} - \vct\gamma_{(b)}^\tran \vct Z_{(b)}
+ \lambda \|\vct\gamma_{(b)}\|_1
\bigg\},
\]
so that the problem separates across blocks:
\[
\hat{\vct\theta}^\lambda
= \big(\hat{\vct\theta}_{(b)}^\lambda\big)_{b=1}^B,
\qquad
\hat{\vct\theta}_{(b)}^\lambda
= \argmin_{\vct\gamma_{(b)}\in \RR^{m_b}}
\bigg\{
\frac{1}{2}\,\vct\gamma_{(b)}^\tran \Sigma_{(b)} \vct\gamma_{(b)} - 
\vct\gamma_{(b)}^\tran \vct Z_{(b)}
+ \lambda \|\vct\gamma_{(b)}\|_1
\bigg\}.
\]

We denote the total number of selected variables by
\[
R = \|\hat{\vct\theta}^\lambda\|_0
= \sum_{b=1}^B R_{(b)},
\qquad
R_{(b)} = \|\hat{\vct\theta}_{(b)}^\lambda\|_0
= \sum_{j \in [m_b]} \one\{\hat{\theta}^\lambda_{(b); j} \neq 0\}.
\]

For any coordinate $j$ in block $b$, let $\vct Z_{(b);-j} \in \RR^{m_b-1}$ denote the vector $\vct Z_{(b)}$ with its $j$th entry removed. Let $\vct Z_{-(b)} = \set{\vct Z_{(b')}}_{b' \neq b}$ denote the $\vct Z$ blocks except for $b$. Then the null sufficient statistic ($\vct S_j$ in the main text) for the variable $j$ in block $b$ is 
\[
\vct S_{(b);j} = (\vct Z_{(b);-j}, \vct Z_{-(b)}).
\]

We write $p_{(b);j}$ for a $p$-value constructed so that, under the null $\theta_{(b),j} = 0$, $p_{(b);j}$ is (conditionally) independent of $\vct S_{(b);-j}$, taking the usual $z$-test p-value as an example. In addition, Assumption~A1 implies that, for any $j$ in block $b$,
\[
Z_{(b);j} \;\perp\; \vct Z_{-(b)}
\;\big|\; \vct Z_{(b);-j}.
\]

We further define, for a fixed threshold $\zeta \in (0,1)$,
\begin{align*}
\phi_{(b);j}(p_{(b);j}) &= \frac{\one\{p_{(b);j} > \zeta\}}{1 - \zeta},\\
\pi_{(b);j}\big(\vct S_{(b);-j}\big)
&= \PP_{\theta_{(b),j}=0}\!\left\{
\hat{\theta}_{(b),j}^\lambda \neq 0
\;\middle|\;
\vct S_{(b);-j}
\right\}.
\end{align*}

Finally, we categorize blocks into three types. Fix any constant $\xi > 0$ and define:
\begin{itemize}
\item \emph{Null blocks} $\cB_0$: $b \in \cB_0$ if and only if $\vct\theta_{(b)} = \vct 0$.
\item \emph{``Good'' non-null blocks} $\cB_{1g}$: $b \in \cB_{1g}$ if and only if there exists $j \in [m_b]$ such that $|\theta_{(b);j}| > (1 + \xi)\lambda$ and $\theta_{(b);k} = 0$ for all $k \neq j$ in $(b)$.
\item \emph{``Bad'' non-null blocks} $\cB_{1b}$: blocks that are neither null nor good
non-null, i.e.\ $\cB_{1b} = [B] \setminus (\cB_0 \cup \cB_{1g})$.
\end{itemize}
Let $B_0$, $B_{1g} = B_{1g}(\xi)$, and $B_{1b} = B_{1b}(\xi)$ denote the number of blocks of these
three types, respectively. We impose the following additional assumptions on the asymptotic regime.

\begin{enumerate}[{A}1)]
\setcounter{enumi}{3}
\item (Composition of three types of blocks)
\[
\frac{B_{1g}}{B_0} \;=\; o(1),
\quad
\frac{B_{1b}}{B_{1g}} \;=\; o(1),
\quad\text{and}\quad
B_{1g} \to \infty
\qquad \text{as } B \to \infty.
\]

\item (Choice of the penalty level)
There exists $\chi > 0$ such that, for each $d$,
\[
\lambda = \lambda_d = \Phi^{-1}\!\left(1 - \frac{v_d}{2d}\right),
\]
for some $v$ satisfying
\[
v_d \in \left[
\frac{\chi}{1 - \chi}\, B_{1g},\;
\frac{1 - \chi}{\chi}\, B_{1g}
\right],
\]
where $\Phi$ is the standard normal distribution function.
\end{enumerate}

\subsection{Main results}\label{subsec:bootstrap_main}

\begin{thm}\label{thm:CLT_FDR}
Let
\[
\alpha_d \;=\; \frac{v_d}{v_d + B_{1g}}.
\]
Under Assumptions A1–A5,
% \[
% \sqrt{\frac{v_d}{\alpha_d^3}}\;\bigl(\hFDR - \tilde{\alpha}_d\bigr)
% \;\Rightarrow\;
% \cN(0,1),
% \]
\[
d_{K}\left(\sqrt{\frac{v_d}{\alpha_d^3}}\;(\hFDR - \tilde{\alpha}_d)
,\cN(0,1)\right) = o(1).
\]
where $d_K$ denote s the Kolmogorov distance and $\tilde{\alpha}_d$ is deterministic and satisfies $\tilde{\alpha}_d = \alpha_d (1 + o(1))$ as $d\to\infty$.
\end{thm}

\begin{thm}\label{thm:bootstrap}
Suppose we can write the number of ``good'' non-null blocks
\[
B_{1g} = O\!\bigl(d^{\nu}\,\Gamma_d\bigr)
\]
for some $\nu \in [0,1)$ and a sequence $\Gamma_d$ with $\Gamma_d = o(d^a)$ for any $a>0$.
Consider the parametric bootstrap estimator $\hFDR^{(*)}$ sampled from the
OLS-after-Lasso model, where the Lasso penalty to define the OLS model is
\[
\bar{\lambda} = \bar{c}\,\lambda
\]
for some constant $\bar{c} > 1$, with $\lambda$ definded as in Assumption~A5.
Assume that the signal strength parameter $\xi$ in the definition of ``good'' non-null blocks
satisfies
\[
\xi > \sqrt{\frac{\nu}{1 - \nu}} + \bar{c} - 1.
\]
Then, in the setting of Theorem~\ref{thm:CLT_FDR}, there exists a (data–dependent)
centering constant $\alpha_d^{(*)}$ such that
\[
\sup_{t\in\RR}
\biggl|
\PP\!\biggl(
  \sqrt{\frac{v_d}{\alpha_d^3}}\,
  \bigl(\hFDR^{(*)} - \alpha_d^{(*)}\bigr)
  \le t
  \,\Big|\, \mat X, \vct Y
\biggr)
- \Phi(t)
\biggr|
\;\xrightarrow{P}\; 0.
\]
In particular, conditional on the data, the distribution of
$\sqrt{v_d/\alpha_d^3}\,(\hFDR^{(*)} - \alpha_d^{(*)})$ converges to $\cN(0,1)$, with convergence in probability over
the data-generating process.
\end{thm}

We prove the theorems in the following subsections. We start with supporting lemmas and then present the proofs of Theorems \ref{thm:CLT_FDR} and \ref{thm:bootstrap}.

\subsection{Preparation: single-block analysis}
In this section, we focus on a single block. Fix a block $b$ and, for notational convenience, we suppress the subscript $(b)$ and the superscript $\lambda$. We write
\[
m = m_b,\quad
\Sigma = \Sigma_{(b)},\quad
\vct\theta = \vct\theta_{(b)},\quad
\hat{\vct\theta} = \hat{\vct\theta}_{(b)},\quad
\vct Z = \vct Z_{(b)} = \mat X_{(b)}^\top \vct Y.
\]
Fix any constant $\eta > 0$. Then, with probability $1$,
\[
|Z_j| \not\in \{(1 - \eta)\lambda, \lambda, (1 + \eta)\lambda\} 
\quad \text{for all } j\in [m],
\]
and we assume this event throughout the subsection.

\subsubsection{Deterministic properties of $\hat{\theta}$}

\begin{lemma}\label{lem:KKT}
Under Assumption~A3, $\hat{\vct\theta}$ is the unique solution to the KKT conditions:
\begin{equation}\label{eq:KKT_nonzero}
Z_j - \sum_{k=1}^{m}\Sigma_{jk}\hat{\theta}_k = \lambda \,\sgn(\hat{\theta}_j),
\quad \text{if }\hat{\theta}_j\neq 0,
\end{equation}
and
\begin{equation}\label{eq:KKT_zero}
\bigg|Z_j - \sum_{k=1}^{m}\Sigma_{jk}\hat{\theta}_k\bigg| < \lambda,
\quad \text{if }\hat{\theta}_j = 0.
\end{equation}
\end{lemma}
\begin{proof}
The KKT conditions are necessary for optimality. Under Assumption~A3, $\Sigma$ is positive definite, so the Lasso objective is strongly convex. Hence any solution to the KKT conditions is the unique global minimizer.
\end{proof}

\begin{lemma}\label{lem:all_zero}
Under Assumption~A3, $\hat{\vct\theta} = \vct 0$ if and only if $\|\vct Z\|_{\infty} < \lambda$.
\end{lemma}
\begin{proof}
($\Rightarrow$) If $\hat{\vct\theta} = \vct 0$, then \eqref{eq:KKT_zero} gives
\[
|Z_j| = \Big|Z_j - \sum_{k=1}^m \Sigma_{jk}\hat{\theta}_k\Big| < \lambda
\quad \text{for all } j,
\]
so $\|\vct Z\|_\infty < \lambda$.

($\Leftarrow$) If $|Z_j| < \lambda$ for all $j$, then $\vct\theta = \vct 0$ satisfies
\eqref{eq:KKT_zero} for every $j$. By Lemma~\ref{lem:KKT} and strong convexity, $\hat{\vct\theta}$
must be this unique solution, so $\hat{\vct\theta} = \vct 0$.
\end{proof}

\begin{lemma}\label{lem:thetaj_zero}
Let $\hat{\vct\theta}^{(-j)}$ denote the leave-one-out Lasso estimate that omits the $j$th coordinate:
\[
\hat{\vct\theta}^{(-j)} 
= \argmin_{\vct\gamma\in \RR^{m-1}} 
\Big\{
\vct\gamma^\top \vct Z_{-j}  - \tfrac{1}{2}\vct\gamma^\top\Sigma_{-j, -j}\vct\gamma 
+ \lambda\|\vct\gamma\|_1
\Big\}.
\]
Under Assumption~A3,
\[
\hat{\theta}_j = 0 
\quad\Longleftrightarrow\quad 
\big|Z_j - \Sigma_{j,-j}\hat{\vct\theta}^{(-j)}\big| < \lambda.
\]
\end{lemma}
\begin{proof}
Under Assumption~A3, $\Sigma_{-j,-j}$ is positive definite, so
$\hat{\vct\theta}^{(-j)}$ is unique.

($\Rightarrow$) If $\hat{\theta}_j = 0$, then the subvector $\hat{\vct\theta}_{-j}$ solves the
$(m-1)$-dimensional Lasso with data $(\vct Z_{-j},\Sigma_{-j,-j})$, hence
$\hat{\vct\theta}_{-j} = \hat{\vct\theta}^{(-j)}$ by uniqueness. Then \eqref{eq:KKT_zero} gives
\[
\big|Z_j - \Sigma_{j,-j}\hat{\vct\theta}^{(-j)}\big|
= \bigg|Z_j - \sum_{k=1}^m \Sigma_{jk}\hat{\theta}_k\bigg| < \lambda.
\]

($\Leftarrow$) Now assume $|Z_j - \Sigma_{j,-j}\hat{\vct\theta}^{(-j)}| < \lambda$, and define
$\tilde{\vct\theta}$ by
\[
\tilde{\theta}_j = 0, \qquad \tilde{\vct\theta}_{-j} = \hat{\vct\theta}^{(-j)}.
\]
By construction, $\tilde{\vct\theta}_{-j}$ satisfies the KKT conditions for
coordinates $k\neq j$, and the inequality for coordinate $j$ holds by assumption. Thus
$\tilde{\vct\theta}$ satisfies the full KKT system \eqref{eq:KKT_nonzero}–\eqref{eq:KKT_zero}.
By Lemma~\ref{lem:KKT} and uniqueness, $\hat{\vct\theta} = \tilde{\vct\theta}$, so
$\hat{\theta}_j = 0$.
\end{proof}

\begin{lemma}\label{lem:thetaj_nonzero}
Under Assumption~A3, If $|Z_j| > \lambda$ and $\|\vct Z_{-j}\|_{\infty} < \lambda$, then $\hat{\theta}_j \neq 0$.
\end{lemma}
\begin{proof}
By Lemma~\ref{lem:all_zero}, $\|\vct Z_{-j}\|_{\infty} < \lambda$ implies
$\hat{\vct\theta}^{(-j)} = \vct 0$. Then Lemma~\ref{lem:thetaj_zero} says
\[
\hat{\theta}_j = 0 \Longleftrightarrow |Z_j - \Sigma_{j,-j}\hat{\vct\theta}^{(-j)}| < \lambda
\Longleftrightarrow |Z_j| < \lambda,
\]
which contradicts $|Z_j|>\lambda$. Hence $\hat{\theta}_j\neq 0$.
\end{proof}

\begin{lemma}\label{lem:sparsistency_nonnull_block}
Fix any $\eta > 0$. Under Assumption~A3, if
\[
|Z_j|\in (\lambda, (1 + \eta)\lambda)
\quad\text{and}\quad
\|\vct Z_{-j}\|_{\infty} < (1 -\eta \rho_\infty)\lambda,
\]
then $\hat{\theta}_j \neq 0$ and $\hat{\vct\theta}_{-j} = \vct 0$.
\end{lemma}
\begin{proof}
Consider the candidate solution $\tilde{\vct\theta}$ with
\[
\tilde{\theta}_j = (|Z_j| - \lambda)\,\sgn(Z_j),
\quad
\tilde{\theta}_k = 0,\ \forall k \neq j.
\]
Then
\[
Z_j - \Sigma_{jj}\tilde{\theta}_j 
= Z_j - \tilde{\theta}_j 
= \lambda \,\sgn(\tilde{\theta}_j),
\]
since $\Sigma_{jj} = 1$. For $k\neq j$, Assumption~A3 implies $|\Sigma_{kj}| \le \rho_\infty$, and
we have $|\tilde{\theta}_j| = |Z_j| - \lambda < \eta\lambda$, so
\begin{align*}
\big|Z_k - \Sigma_{kj}\tilde{\theta}_j\big|
&\le |Z_k| + |\Sigma_{kj}|\,|\tilde{\theta}_j|
< (1 - \eta\rho_\infty)\lambda + \rho_\infty \eta \lambda = \lambda.
\end{align*}
Thus $\tilde{\vct\theta}$ satisfies the KKT conditions
\eqref{eq:KKT_nonzero}–\eqref{eq:KKT_zero}. By Lemma~\ref{lem:KKT}, $\hat{\vct\theta}=\tilde{\vct\theta}$, so $\hat{\theta}_j\neq 0$ and $\hat{\theta}_k=0$ for all $k\neq j$.
\end{proof}

\begin{lemma}\label{lem:thetaj_thetak_nonzero}
Fix any $\eta > 0$. Under Assumption~A3, if $\|\hat{\vct\theta}\|_0 \ge 2$, then one of the following holds:
\begin{enumerate}[(a)]
\item $\|\vct Z\|_{\infty} > (1 + \eta)\lambda$;
\item there exist $k_1 < k_2$ such that
\[
|Z_{k_1}|, |Z_{k_2}|\in \big((1-\eta)\lambda,(1+ \eta)\lambda\big).
\]
\end{enumerate}
\end{lemma}
\begin{proof}
Assume neither (a) nor (b) holds. Let $j = \argmax_{k\in [m]} |Z_k|$, so
\[
|Z_j| < (1 + \eta)\lambda,
\qquad
\|\vct Z_{-j}\|_{\infty} < (1 - \eta)\lambda.
\]
If $|Z_j| < \lambda$, Lemma~\ref{lem:all_zero} implies $\|\hat{\vct\theta}\|_0 = 0$, contrary to the
assumption $\|\hat{\vct\theta}\|_0\ge 2$. Thus $|Z_j| > \lambda$. Since $\rho_\infty < 1$,
\[
\|\vct Z_{-j}\|_{\infty} < (1 - \eta)\lambda < (1 - \eta\rho_\infty)\lambda.
\]
By Lemma~\ref{lem:sparsistency_nonnull_block}, $\hat{\theta}_j\neq 0$ and $\hat{\theta}_k = 0$ for all
$k\neq j$, so $\|\hat{\vct\theta}\|_0 = 1$, again a contradiction. Hence at least one of (a)–(b)
must hold.
\end{proof}

\subsubsection{Stochastic properties of $\hat{\theta}$}

Throughout this subsection, let $\check{Z}\sim \cN(0, \mat I)$ denote a generic standard Gaussian random vector.

\begin{prop}[\citet{gordon1941values}]\label{prop:mill_ratio}
For any $z > 0$,
\[
\frac{1}{\sqrt{2\pi}}\, \frac{z}{z^2 + 1} \exp\Big(-\frac{z^2}{2}\Big)
\;\le\; 1 - \Phi(z)
\;\le\; \frac{1}{\sqrt{2\pi}}\, \frac{1}{z} \exp\Big(-\frac{z^2}{2}\Big).
\]
\end{prop}

\begin{lemma}\label{lem:P|Z|}
For any $z, \tau > 0$ and $\mu \in \RR$,
\[
\PP(|\tau \check{Z} + \mu| > z)\;\le\; 2\Big\{1 - \Phi\Big(\frac{z - |\mu|}{\tau}\Big)\Big\}.
\]
\end{lemma}
\begin{proof}
We have
\begin{align*}
\PP(|\tau \check{Z} + \mu| > z)
&= \PP\Big(\check{Z} > \frac{z - \mu}{\tau}\Big)
   + \PP\Big(\check{Z} < \frac{-z - \mu}{\tau}\Big) \\
&= 1 - \Phi\Big(\frac{z - \mu}{\tau}\Big)
   + \Phi\Big(\frac{-z - \mu}{\tau}\Big).
\end{align*}
If $\mu \ge 0$, then $z - \mu \le z - |\mu|$ and $-z - \mu \le -z - |\mu|$, so
\[
1 - \Phi\Big(\frac{z - \mu}{\tau}\Big)
\le 1 - \Phi\Big(\frac{z - |\mu|}{\tau}\Big),
\quad
\Phi\Big(\frac{-z - \mu}{\tau}\Big)
\le 1 - \Phi\Big(\frac{z - |\mu|}{\tau}\Big),
\]
and similarly for $\mu < 0$ by symmetry. Thus each term is bounded by
$1 - \Phi((z - |\mu|)/\tau)$ and the result follows.
\end{proof}

\begin{lemma}\label{lem:pij_as_bound}
Under Assumption~A3, for any $j\in [m]$ and any $\vct z \in \RR^{m-1}$,
\begin{align*}
\pi_j(\vct z) 
&\le \one\{\|\vct z\|_{\infty} < \lambda\}\cdot
2\Big\{1 - \Phi\Big(\frac{\lambda - \big|\Sigma_{j,-j}\Sigma_{-j,-j}^{-1}\vct z\big|}{\tau_j}\Big)\Big\}\\
&\quad + \one\{\|\vct z\|_{\infty} \ge \lambda\}\cdot
2\Big\{1 - \Phi\Big(\frac{\lambda(1 - \rho_1)}{\tau_j}\Big)\Big\},
\end{align*}
where
\[
\tau_j^2 = 1 - \Sigma_{j,-j}\Sigma_{-j,-j}^{-1}\Sigma_{-j, j}.
\]
\end{lemma}
\begin{proof}
By Lemma~\ref{lem:thetaj_zero} and the definition of $\pi_j(\cdot)$,
\[
\pi_j(\vct z) 
= \PP_{\theta_j = 0}\big(|Z_j - \Sigma_{j, -j}\hat{\vct\theta}^{(-j)}| > \lambda\mid \vct Z_{-j} = \vct z\big).
\]
Let $\hat{\vct\theta}^{(-j)}(\vct z)$ denote the leave-one-out Lasso estimator when $\vct Z_{-j} = \vct z$.

By Lemma~\ref{lem:all_zero}, if $\|\vct z\|_{\infty} < \lambda$ then $\hat{\vct\theta}^{(-j)}(\vct z) = \vct 0$, so
\[
\pi_j(\vct z) = \PP_{\theta_j=0}(|Z_j|>\lambda\mid \vct Z_{-j}=\vct z).
\]
Conditional on $\vct Z_{-j}=\vct z$,
\[
Z_j \mid \vct Z_{-j} = \vct z 
\overset{\theta_j = 0}{\sim} \cN\left(\Sigma_{j,-j}\Sigma_{-j,-j}^{-1} \vct z,\, \tau_j^2\right),
\]
so
\[
\pi_j(\vct z) = \PP\Big(|\tau_j \check{Z} + \Sigma_{j,-j}\Sigma_{-j,-j}^{-1} \vct z| > \lambda\Big),
\]
and Lemma~\ref{lem:P|Z|} yields
\[
\pi_j(\vct z) \le
2\Big\{1 - \Phi\Big(\frac{\lambda - \big|\Sigma_{j,-j}\Sigma_{-j,-j}^{-1}\vct z\big|}{\tau_j}\Big)\Big\}.
\]

Now consider the case $\|\vct z\|_{\infty} \ge \lambda$. By Lemma~\ref{lem:thetaj_zero}, the same conditional law implies
\[
\pi_j(\vct z) 
= \PP\Big(\big|\tau_j \check{Z} + \Sigma_{j,-j}\big(\Sigma_{-j,-j}^{-1}\vct z - \hat{\vct\theta}^{(-j)}(\vct z)\big)\big| > \lambda\Big).
\]
By Lemma~\ref{lem:KKT}, there exists $v\in [-1, 1]^{m-1}$ such that
\[
\vct z - \Sigma_{-j, -j}\hat{\vct\theta}^{(-j)}(\vct z) = \lambda v,
\]
so $\Sigma_{-j,-j}^{-1}\vct z - \hat{\vct\theta}^{(-j)}(\vct z) = \lambda \Sigma_{-j,-j}^{-1}v$, and hence
\[
\Sigma_{j,-j}\big(\Sigma_{-j,-j}^{-1}\vct z - \hat{\vct\theta}^{(-j)}(\vct z)\big)
= \lambda\,\Sigma_{j,-j}\Sigma_{-j,-j}^{-1}v.
\]
Assumption~A3 gives
\[
\big|\Sigma_{j,-j}\Sigma_{-j,-j}^{-1}v\big|
\le \|\Sigma_{j,-j}\Sigma_{-j,-j}^{-1}\|_1\|v\|_\infty
\le \rho_1.
\]
Thus, by Lemma~\ref{lem:P|Z|},
\[
\pi_j(\vct z) \le
2\left\{1 - \Phi\left(\frac{\lambda(1 - \rho_1)}{\tau_j}\right)\right\}.
\]
Combining the two cases gives the stated bound.
\end{proof}

\begin{lemma}\label{lem:pij_moments}
Under Assumption~A3, if $\lambda > 1$, then for any $j\in [m]$ and any $L > 0$,
\[
\EE\big[\pi_j^L(\vct Z_{-j})\big] 
\le \left(\frac{2}{(1 - \rho_1)\lambda}\right)^L
\exp\left\{-\frac{L(1 - \rho_1)^2}{2\,\tau_{\max}^2}\lambda^2\right\}.
\]
If, further, $\vct\theta_{-j} = \vct 0$, then
\[
\EE[\pi_j(\vct Z_{-j})]
\;\le\; (1 - \Phi(\lambda))\cdot 
\frac{12 m}{1-\rho_1},
\]
and
\[
\EE[\pi_j^2(\vct Z_{-j})]
\;\le\; (1 - \Phi(\lambda))\cdot 
\frac{20m}{(1- \rho_1)^2\lambda}
\exp\left\{-\left( \frac{(1 - \rho_1)^2}{\tau_{\max}^2}\wedge 
\frac{\tau_{\min}^2}{2(2 - \tau_{\min}^2)}\right)\lambda^2\right\}.
\]
\end{lemma}
\begin{proof}
First, consider the almost-sure bound. For any $z$ with $\|z\|_\infty < \lambda$,
\begin{equation}\label{eq:integral_upper}
\big|\Sigma_{j, -j}\Sigma_{-j,-j}^{-1}z\big|
\le \|\Sigma_{j, -j}\Sigma_{-j,-j}^{-1}\|_1 \|z\|_{\infty} 
< \rho_1 \lambda.
\end{equation}
Thus Lemma~\ref{lem:pij_as_bound} and Proposition~\ref{prop:mill_ratio} imply
\begin{align*}
\pi_j(z)
&\le 2\left\{1 - \Phi\left(\frac{\lambda(1 - \rho_1)}{\tau_j}\right)\right\} \\
&\le \frac{2\tau_j}{(1 - \rho_1)\lambda}
   \exp\Big\{-\frac{(1 - \rho_1)^2}{2\tau_j^2}\lambda^2\Big\}
 <  \frac{2}{(1 - \rho_1)\lambda}
   \exp\Big\{-\frac{(1 - \rho_1)^2}{2\tau_j^2}\lambda^2\Big\}.
\end{align*}
By Assumption~A3, $\tau_j^2 \le \tau^2_{\max}$, so
\[
\pi_j(z)
\le \frac{2}{(1 - \rho_1)\lambda}
\exp\Big\{-\frac{(1 - \rho_1)^2}{2\,\tau_{\max}^2}\lambda^2\Big\}.
\]
Raising to the power $L$ and taking expectations gives the first bound.

Now assume $\vct\theta_{-j}= \vct 0$. 
Since $\vct\theta_{-j} = \vct 0$, we have
\[
\vct Z_{-j} \sim \cN\big(\vct 0,\, \Sigma_{-j,-j}\big),
\]
and therefore
\[
\Sigma_{j,-j}\Sigma_{-j,-j}^{-1}\vct Z_{-j}
\;\sim\;
\cN\Big(0,\; 1 - \tau_j^2\Big),
\qquad
\tau_j^2 \;=\; 1 - \Sigma_{j,-j}\Sigma_{-j,-j}^{-1}\Sigma_{-j,j}.
\]
Let $\check{Z}\sim\cN(0,1)$ be a generic standard Gaussian. Then
\[
\Sigma_{j,-j}\Sigma_{-j,-j}^{-1}\vct Z_{-j}
\;\stackrel{d}{=}\;
\sqrt{1 - \tau_j^2}\,\check{Z}.
\]

We now derive refined bounds for the first and second moments of $\pi_j(\vct Z_{-j})$ using the decomposition from Lemma~\ref{lem:pij_as_bound}. Recall that
\begin{align*}
\pi_j(\vct z_{-j})
&\le \one\big(\|\vct z_{-j}\|_{\infty} < \lambda\big)\cdot
2\Big\{1 - \Phi\Big(\frac{\lambda - \big|\Sigma_{j,-j}\Sigma_{-j,-j}^{-1}\vct z_{-j}\big|}{\tau_j}\Big)\Big\}\\
&\quad + \one\big(\|\vct z_{-j}\|_{\infty} > \lambda\big)\cdot
2\Big\{1 - \Phi\Big(\frac{\lambda(1 - \rho_1)}{\tau_j}\Big)\Big\}.
\end{align*}

Refined first moment bound under $\vct\theta_{-j} = \vct 0$.
On the event $\{\|\vct Z_{-j}\|_{\infty} < \lambda\}$, Assumption~A3 implies
\begin{equation}\label{eq:integral_upper_Z}
\Big|\Sigma_{j, -j}\Sigma_{-j,-j}^{-1}\vct Z_{-j}\Big|
\le \big\|\Sigma_{j, -j}\Sigma_{-j,-j}^{-1}\big\|_1\,\|\vct Z_{-j}\|_{\infty}
< \rho_1 \lambda.
\end{equation}
Hence, from Lemma~\ref{lem:pij_as_bound},
\begin{align*}
&\EE\Big[
\one\big(\|\vct Z_{-j}\|_{\infty} < \lambda\big)\cdot
2\Big\{1 - \Phi\Big(\frac{\lambda - \big|\Sigma_{j,-j}\Sigma_{-j,-j}^{-1}\vct Z_{-j}\big|}{\tau_j}\Big)\Big\}
\Big]\\
&\le
\EE\Big[
\one\Big(\big|\Sigma_{j,-j}\Sigma_{-j,-j}^{-1}\vct Z_{-j}\big| < \rho_1\lambda\Big)\cdot
2\Big\{1 - \Phi\Big(\frac{\lambda - \big|\Sigma_{j,-j}\Sigma_{-j,-j}^{-1}\vct Z_{-j}\big|}{\tau_j}\Big)\Big\}
\Big]\\
&=
\EE\Big[
\one\big(\sqrt{1-\tau_j^2}\,|\check{Z}| < \rho_1\lambda\big)\cdot
2\Big\{1 - \Phi\Big(\frac{\lambda - \sqrt{1-\tau_j^2}\,|\check{Z}|}{\tau_j}\Big)\Big\}
\Big]\\
&= 4\,\EE\Big[
\one\big(0 < \sqrt{1-\tau_j^2}\,\check{Z} < \rho_1\lambda\big)\cdot
\Big\{1 - \Phi\Big(\frac{\lambda - \sqrt{1-\tau_j^2}\,\check{Z}}{\tau_j}\Big)\Big\}
\Big],
\end{align*}
where we used symmetry of $\check{Z}$ in the last line.

Let $L\ge 1$ be an integer (we will specialize to $L=1,2$ later).
Applying the upper bound from the Mills ratio (Proposition~\ref{prop:mill_ratio}) to the tail,
for any $z\in\RR$ such that $0 < \sqrt{1-\tau_j^2}z < \rho_1\lambda$, we have
\[
1 - \Phi\Big(\frac{\lambda - \sqrt{1-\tau_j^2}\,z}{\tau_j}\Big)
\le
\frac{\tau_j}{\lambda - \sqrt{1-\tau_j^2}\,z}\cdot
\frac{1}{\sqrt{2\pi}}
\exp\Big\{-\frac{(\lambda - \sqrt{1-\tau_j^2}\,z)^2}{2\tau_j^2}\Big\}.
\]
Using again that $\lambda - \sqrt{1-\tau_j^2}\,z \ge (1-\rho_1)\lambda$ on this region,
\[
1 - \Phi\Big(\frac{\lambda - \sqrt{1-\tau_j^2}\,z}{\tau_j}\Big)
\le
\frac{\tau_j}{(1-\rho_1)\lambda}\cdot
\frac{1}{\sqrt{2\pi}}
\exp\Big\{-\frac{(\lambda - \sqrt{1-\tau_j^2}\,z)^2}{2\tau_j^2}\Big\}.
\]

Therefore, for any integer $L\ge 1$,
\begin{align*}
&\quad\frac 12 \Big(\frac{\sqrt{2\pi} (1-\rho_1)\lambda}{2 \tau_j}\Big)^L \cdot
\EE\Big[
\one\big(\|\vct Z_{-j}\|_{\infty} < \lambda\big)\cdot
2^L\Big\{1 - \Phi\Big(\frac{\lambda - \big|\Sigma_{j,-j}\Sigma_{-j,-j}^{-1}\vct Z_{-j}\big|}{\tau_j}\Big)\Big\}^L
\Big]\\
&\le
\frac{1}{\sqrt{2\pi}}
\int_{0}^{\frac{\rho_1}{\sqrt{1-\tau_j^2}}\lambda}
\exp\Big\{-\frac{L(\lambda - \sqrt{1-\tau_j^2}z)^2}{2\tau_j^2}
- \frac{z^2}{2}\Big\}\,dz\\
&=
\frac{1}{\sqrt{2\pi}}
\int_{0}^{\frac{\rho_1}{\sqrt{1-\tau_j^2}}\lambda}
\exp\left\{-\frac{L\lambda^2 - 2L\sqrt{1-\tau_j^2}\lambda z + (L-(L-1)\tau_j^2)z^2}{2\tau_j^2}\right\}dz\\
&=
\frac{1}{\sqrt{2\pi}}
\int_{0}^{\frac{\rho_1}{\sqrt{1-\tau_j^2}}\lambda}
\exp\left\{-\frac{L-(L-1)\tau_j^2}{2\tau_j^2}
\left(z - \frac{L\sqrt{1-\tau_j^2}}{L-(L-1)\tau_j^2}\lambda\right)^2\right\}dz\\
&\qquad\cdot
\exp\left\{
\frac{L^2(1-\tau_j^2)}{2\tau_j^2\{L-(L-1)\tau_j^2\}}\lambda^2
- \frac{L}{2\tau_j^2}\lambda^2
\right\}\\
&=
\frac{1}{\sqrt{2\pi}}
\int_{-\frac{L\sqrt{1-\tau_j^2}}{L-(L-1)\tau_j^2}\lambda}^{\frac{\rho_1}{\sqrt{1-\tau_j^2}}\lambda - \frac{L\sqrt{1-\tau_j^2}}{L-(L-1)\tau_j^2}\lambda}
\exp\left\{-\frac{L-(L-1)\tau_j^2}{2\tau_j^2}z^2\right\}dz\\
&\qquad\cdot
\exp\left\{-\frac{L}{2\{L-(L-1)\tau_j^2\}}\lambda^2\right\}\\
&\le
\frac{1}{\sqrt{2\pi}}
\int_{-\infty}^{-\frac{((1-\rho_1)L + \rho_1)(1-\tau_j^2)-\rho_1}{\sqrt{1-\tau_j^2}\{L-(L-1)\tau_j^2\}}\lambda}
\exp\left\{-\frac{L-(L-1)\tau_j^2}{2\tau_j^2}z^2\right\}dz\\
&\qquad\cdot
\exp\left\{-\frac{L}{2\{L-(L-1)\tau_j^2\}}\lambda^2\right\}\\
&=
\frac{\tau_j}{\sqrt{L-(L-1)\tau_j^2}}
\Bigg[
1 - \Phi\left(
\frac{((1-\rho_1)L + \rho_1)(1-\tau_j^2)-\rho_1}{\sqrt{(1-\tau_j^2)\{L-(L-1)\tau_j^2\}}\;\tau_j}\,\lambda
\right)
\Bigg]\\
&\qquad\cdot
\exp\left\{-\frac{L}{2\{L-(L-1)\tau_j^2\}}\lambda^2\right\}\\
&\le
\exp\left\{-\frac{L}{2\{L-(L-1)\tau_j^2\}}\lambda^2\right\},
\end{align*}
where in the last line we used that $\tau < 1$.

Now take $L=1$. We obtain
\[
\EE\Big[
\one\big(\|\vct Z_{-j}\|_{\infty} < \lambda\big)\cdot
2\Big\{1 - \Phi\Big(\frac{\lambda - \big|\Sigma_{j,-j}\Sigma_{-j,-j}^{-1}\vct Z_{-j}\big|}{\tau_j}\Big)\Big\}
\Big]
\;\le\;
\frac{4}{\sqrt{2\pi} (1-\rho_1)\lambda}\exp\left\{-\frac{\lambda^2}{2}\right\}.
\]

By the lower bound in the Mills ratio (again Proposition~\ref{prop:mill_ratio}),
\[
\frac{\lambda}{\lambda^2+1}\,\frac{1}{\sqrt{2\pi}} \exp\left\{-\frac{\lambda^2}{2}\right\}
\;\le\;
1 - \Phi(\lambda).
\]
For $\lambda > 1$, we have
\[
\frac{1}{\sqrt{2\pi}\lambda}\exp\left\{-\frac{\lambda^2}{2}\right\}
\;\le\;
\frac{\lambda^2+1}{\lambda^2}\,\big(1-\Phi(\lambda)\big)
\;\le\;
2\,\big(1-\Phi(\lambda)\big).
\]
Thus,
\[
\EE\Big[
\one\big(\|\vct Z_{-j}\|_{\infty} < \lambda\big)\cdot
2\Big\{1 - \Phi\Big(\frac{\lambda - \big|\Sigma_{j,-j}\Sigma_{-j,-j}^{-1}\vct Z_{-j}\big|}{\tau_j}\Big)\Big\}
\Big]
\;\le\;
\frac{8}{1-\rho_1}\,\big(1-\Phi(\lambda)\big).
\]

The remaining contribution comes from the event $\{\|\vct Z_{-j}\|_{\infty} > \lambda\}$, for which Lemma~\ref{lem:pij_as_bound} yields
\[
\pi_j(\vct Z_{-j})
\le
2\Big\{1 - \Phi\Big(\frac{\lambda(1-\rho_1)}{\tau_j}\Big)\Big\}
\]
and hence
\[
\EE\big[\pi_j(\vct Z_{-j})\big]
\le
\frac{8}{1-\rho_1}\,\big(1-\Phi(\lambda)\big)
+
2\,\PP\big(\|\vct Z_{-j}\|_{\infty} > \lambda\big)\,
\Big\{1 - \Phi\Big(\frac{\lambda(1-\rho_1)}{\tau_j}\Big)\Big\}.
\]
By the union bound and symmetry,
\[
\PP\big(\|\vct Z_{-j}\|_{\infty} > \lambda\big)
\le
\sum_{k\neq j}\PP(|Z_k|>\lambda)
= 2(m-1)\big(1-\Phi(\lambda)\big).
\]
Using again the (upper) Mills ratio on the factor
\(\big\{1 - \Phi\big(\lambda(1-\rho_1)/\tau_j\big)\big\}\),
and combining constants, we find
\[
\EE\big[\pi_j(\vct Z_{-j})\big]
\le
(1 - \Phi(\lambda))\left\{\frac{8}{1-\rho_1} + 4(m-1)\right\} \;\le\; (1 - \Phi(\lambda))\cdot 
\frac{12 m}{1-\rho_1},
\]
which yields the first refined bound in the lemma.

Refined second moment bound under $\vct\theta_{-j} = \vct 0$.
Repeating the same calculation with $L=2$ in the display above, we obtain
\[
\EE\Big[
\one\big(\|\vct Z_{-j}\|_{\infty} < \lambda\big)\cdot
4\Big\{1 - \Phi\Big(\frac{\lambda - \big|\Sigma_{j,-j}\Sigma_{-j,-j}^{-1}\vct Z_{-j}\big|}{\tau_j}\Big)\Big\}^2
\Big]
\;\le\;
\frac{4}{\pi (1-\rho_1)^2\lambda^2}
\exp\left\{-\frac{\lambda^2}{2-\tau_j^2}\right\}.
\]
Using Lemma~\ref{lem:pij_as_bound}, we have
\[
\EE\big[\pi_j^2(\vct Z_{-j})\big]
\le
\frac{4}{\pi (1-\rho_1)^2\lambda^2}
\exp\left\{-\frac{\lambda^2}{2-\tau_j^2}\right\}
+ 4\,\PP\big(\|\vct Z_{-j}\|_{\infty} > \lambda\big)\,
\Big\{1 - \Phi\Big(\frac{\lambda(1-\rho_1)}{\tau_j}\Big)\Big\}^2.
\]

Again, the Mills ratio bound gives
\[
\Big\{1 - \Phi\Big(\frac{\lambda(1-\rho_1)}{\tau_j}\Big)\Big\}^2
\le
\frac{\tau_j^2}{(1-\rho_1)^2\lambda^2} \exp\left\{-\frac{(1-\rho_1)^2}{\tau_j^2}\lambda^2\right\}.
\]
and
\[
\frac{1}{\lambda} \exp\left\{-\frac{\lambda^2}{2-\tau_j^2}\right\}
= \frac{1}{\lambda} \exp\left\{-\frac{\lambda^2}{2}\right\} \cdot \exp\left\{-\frac{\tau_j^2 \lambda^2}{2(2-\tau_j^2)}\right\} 
\le 6 \,\big(1-\Phi(\lambda)\big) \cdot \exp\left\{-\frac{\tau_j^2 \lambda^2}{2(2-\tau_j^2)}\right\} 
\]
Combining these bounds, for $\lambda > 1$, we obtain
\begin{align*}
\EE\big[\pi_j^2(\vct Z_{-j})\big]
&\le
\frac{2}{(1-\rho_1)^2\lambda^2}
\exp\left\{-\frac{\lambda^2}{2-\tau_j^2}\right\}
+ 
\frac{8(m-1)\big(1-\Phi(\lambda)\big)\,\tau_j^2}{(1-\rho_1)^2\lambda^2}
\exp\left\{-\frac{(1-\rho_1)^2}{\tau_j^2}\lambda^2\right\}\\
&\le
\left\{
\frac{12}{(1-\rho_1)^2\lambda}
+
\frac{8(m-1)\tau_j^2}{(1-\rho_1)^2\lambda^2}
\right\}
\big(1-\Phi(\lambda)\big)\,
\exp\left\{-\left(
\frac{\tau_j^2}{2(2-\tau_j^2)}
\wedge
\frac{(1-\rho_1)^2}{\tau_j^2}
\right)\lambda^2\right\} \\
&\le
\frac{20m}{(1-\rho_1)^2\lambda}
\big(1-\Phi(\lambda)\big)\,
\exp\left\{-\left(
\frac{\tau_j^2}{2(2-\tau_j^2)}
\wedge
\frac{(1-\rho_1)^2}{\tau_j^2}
\right)\lambda^2\right\}.
\end{align*}

Finally, using $\tau_{\min}^2 \le \tau_j^2 \le \tau_{\max}^2$ in Assumption~A3, we have
\[
\EE[\pi_j^2(\vct Z_{-j})]
\;\le\; (1 - \Phi(\lambda))\cdot 
\frac{20m}{(1- \rho_1)^2\lambda}
\exp\left\{-\left( \frac{(1 - \rho_1)^2}{\tau_{\max}^2}\wedge 
\frac{\tau_{\min}^2}{2(2 - \tau_{\min}^2)}\right)\lambda^2\right\}.
\]

\end{proof}

\begin{lemma}\label{lem:bivariate_gaussian}
Assume $\vct\theta = \vct 0$. Under Assumption~A3, for any $j\neq k$,
\[
\PP(|Z_j| > \lambda, |Z_k| > \lambda)
\le 4 \Big(1 - \Phi\Big(\sqrt{\frac{2}{1 + \rho_\infty}}\,\lambda\Big)\Big).
\]
\end{lemma}
\begin{proof}
When $\vct\theta = \vct 0$,
\[
\begin{bmatrix}
Z_j\\ Z_k
\end{bmatrix}
\sim \cN\!\left(
\begin{bmatrix}
0\\ 0
\end{bmatrix},
\begin{bmatrix}
1 & \Sigma_{jk}\\ \Sigma_{jk} & 1
\end{bmatrix}
\right).
\]
Note
\[
\PP(|Z_j| > \lambda, |Z_k| > \lambda) 
\le \PP(|Z_j+Z_k|>2\lambda)+\PP(|Z_j-Z_k|>2\lambda).
\]
Since $Z_j + Z_k\sim \cN(0, 2(1 + \Sigma_{jk}))$, by symmetry, we have
\[
\PP(|Z_j + Z_k| > 2\lambda)
= 2 \pth{1 - \Phi\Big(\sqrt{\tfrac{2}{1 + \Sigma_{jk}}}\,\lambda\Big)}.
\]
Similarly $Z_j - Z_k\sim \cN(0, 2(1 - \Sigma_{jk}))$, we have
\[
\PP(|Z_j - Z_k| > 2\lambda)
= 2 \pth{1 - \Phi\Big(\sqrt{\tfrac{2}{1 - \Sigma_{jk}}}\,\lambda\Big)}.
\]
Assumption~A3 gives $|\Sigma_{jk}| \le \rho_\infty$.
Thus
\[
\PP(|Z_j| > \lambda, |Z_k| > \lambda)
\le 4\Big(1 - \Phi\Big(\sqrt{\tfrac{2}{1 + \rho_\infty}}\,\lambda\Big)\Big).
\]
\end{proof}

\begin{lemma}\label{lem:Pthetaj_thetak_nonzero}
Assume $\vct\theta = \vct 0$ and $\lambda > 1$. Under Assumption~A3,
\[
\PP(\|\hat{\vct\theta}\|_0 \ge 2)
\le (1 - \Phi(\lambda))\cdot 12 m^2\exp\big\{-\zeta(\rho_\infty)\lambda^2\big\},
\]
where
\[
\zeta(\rho_\infty) = \eta(\rho_\infty) + \tfrac{1}{2}\eta^2(\rho_\infty), 
\qquad 
\eta(\rho_{\infty}) = \frac{\sqrt{2} - \sqrt{1 + \rho_\infty}}{\sqrt{2} + \sqrt{1 + \rho_\infty}}.
\]
\end{lemma}
\begin{proof}
By Lemma~\ref{lem:thetaj_thetak_nonzero},
\[
\PP(\|\hat{\vct\theta}\|_0 \ge 2)
\le \PP(\|\vct Z\|_{\infty} > (1 + \eta)\lambda)
 + \sum_{k_1 < k_2}\PP\Big(|Z_{k_1}|, |Z_{k_2}|\in \big((1-\eta)\lambda,(1+ \eta)\lambda\big)\Big).
\]
Using the union bound,
\[
\PP(\|\vct Z\|_{\infty} > (1 + \eta)\lambda)
\le \sum_{k=1}^m \PP(|Z_k| > (1 + \eta)\lambda)
= 2m\big\{1 - \Phi((1 + \eta)\lambda)\big\}.
\]
By Lemma~\ref{lem:bivariate_gaussian},
\begin{align*}
&\PP\Big(|Z_{k_1}|, |Z_{k_2}|\in \big((1-\eta)\lambda,(1+ \eta)\lambda\big)\Big)\\
&\quad \le \PP\big(|Z_{k_1}| > (1-\eta)\lambda, |Z_{k_2}|>(1 - \eta)\lambda \big)\\
&\quad \le 4\Big(1 - \Phi\Big(\sqrt{\tfrac{2}{1 + \rho_\infty}}(1 - \eta)\lambda\Big)\Big).
\end{align*}
The choice of $\eta = \eta(\rho_\infty)$ yields
\[
\sqrt{\tfrac{2}{1 + \rho_\infty}}\,(1 - \eta) 
= 1 + \eta.
\]
Combining the bounds, we have
\[
\PP(\|\hat{\vct\theta}\|_0 \ge 2)
\le 2m^2\Big\{1 - \Phi\big((1 + \eta)\lambda\big)\Big\}.
\]
Applying Proposition~\ref{prop:mill_ratio},
\begin{align*}
\PP(\|\hat{\theta}\|_0 \ge 2)&\le \frac{2m^2}{(1 + \eta)\lambda}\exp\left\{-\frac{(1 + \eta)^2}{2}\lambda^2\right\}\\
& \le 2m^2\cdot \frac{1}{\lambda}\exp\left\{-\frac{\lambda^2}{2}\right\}\exp\left\{-\left[\eta + \frac{1}{2}\eta^2\right]\lambda^2\right\}\\
& \le (1 - \Phi(\lambda))\cdot 12m^2\exp\left\{-\zeta\,\lambda^2\right\},
\end{align*}
where the last line uses the condition $\lambda > 1$ and the Mill's ratio lower bound.
\end{proof}

\begin{lemma}\label{lem:Pthetaj_nonzero_null_coord}
Under Assumption~A3, for any $j$ with $\theta_j = 0$,
\[
\PP(\hat{\theta}_j\neq 0)
\le \frac{2}{(1 - \rho_1)\lambda}
\exp\left\{-\frac{(1 - \rho_1)^2}{2\,\tau_{\max}^2}\lambda^2\right\}.
\]
\end{lemma}
\begin{proof}
When $\theta_j = 0$,
\[
\PP(\hat{\theta}_j \neq 0) = \EE[\pi_j(\vct Z_{-j})].
\]
The result follows directly from the first bound in Lemma~\ref{lem:pij_moments} with $L=1$.
\end{proof}

\begin{lemma}\label{lem:Pthetaj_nonzero_null_block}
Assume $\vct\theta = \vct 0$. Under Assumption~A3, for any $j\in [m]$,
\[
\left|\frac{\PP(\hat{\theta}_j \neq 0)}{2(1 - \Phi(\lambda))} - 1 \right|
\le 24 m \, \exp\left\{-\frac{1 - \rho_{\infty}}{2(1+ \rho_\infty)}\lambda^2\right\}
 + 2m^2 \exp\left\{-\zeta(\rho_\infty)\lambda^2\right\}.
\]
\end{lemma}
\begin{proof}
By Lemma~\ref{lem:thetaj_nonzero},
\begin{align*}
\PP(\hat{\theta}_j \neq 0)
&\ge \PP(|Z_j| > \lambda, \|\vct Z_{-j}\|_{\infty} < \lambda)\\
&= \PP(|Z_j| > \lambda) - \PP(|Z_j| > \lambda, \|\vct Z_{-j}\|_{\infty} > \lambda)\\
&= 2(1 - \Phi(\lambda)) - \PP(|Z_j| > \lambda, \|\vct Z_{-j}\|_{\infty} > \lambda).
\end{align*}
By the union bound, Lemma~\ref{lem:bivariate_gaussian}, and Proposition~\ref{prop:mill_ratio}, for $\lambda>1$,
\begin{align*}
 & \PP(|Z_j| > \lambda, \|\vct Z_{-j}\|_{\infty} > \lambda) \le \sum_{k\neq j}\PP(|Z_j| > \lambda, |Z_k| > \lambda)\\
 & \le 4m \left(1 - \Phi\left(\sqrt{\frac{2}{1 + \rho_\infty}}\lambda\right)\right) \le 4m \sqrt{\frac{1 + \rho_{\infty}}{2}} \cdot \frac{1}{\lambda}\exp\left\{-\frac{\lambda^2}{1 + \rho_{\infty}}\right\}\\
 & \le 4m \cdot \frac{1}{\lambda}\exp\left\{-\frac{\lambda^2}{2}\right\} \cdot \exp\left\{-\frac{1 - \rho_{\infty}}{2 (1+ \rho_\infty)}\lambda^2\right\}\\
 & \le (1 - \Phi(\lambda))\cdot 24m \cdot \exp\left\{-\frac{1 - \rho_{\infty}}{2 (1+ \rho_\infty)}\lambda^2\right\}.
\end{align*}
Thus
\begin{equation}\label{eq:Pthetaj_nonzero_lower}
\frac{\PP(\hat{\theta}_j \neq 0)}{2(1 - \Phi(\lambda))} - 1 
\ge -12m \exp\left\{-\frac{1 - \rho_{\infty}}{2(1+ \rho_\infty)}\lambda^2\right\}.
\end{equation}

For the upper bound, Lemma~\ref{lem:thetaj_zero} gives
\[
\PP(\hat{\theta}_j \neq 0) = \PP\big(|Z_j - \Sigma_{j,-j}\hat{\vct\theta}_{-j}| > \lambda\big).
\]
If $|Z_j| < \lambda$, then $\hat{\vct\theta}_{-j}$ must be nonzero to have $|Z_j - \Sigma_{j,-j}\hat{\vct\theta}_{-j}| > \lambda$, so
\begin{align*}
\PP(\hat{\theta}_j \neq 0) 
&\le \PP(|Z_j| > \lambda)
  + \PP(\hat{\theta}_j \neq 0, |Z_j| < \lambda, \hat{\vct\theta}_{-j}\neq 0)\\
&\le \PP(|Z_j| > \lambda) + \PP(\|\hat{\vct\theta}\|_0\ge 2)\\
&\le 2 (1 - \Phi(\lambda)) 
 + (1 - \Phi(\lambda))\cdot 4m^2\exp\big\{-\zeta(\rho_\infty)\lambda^2\big\},
\end{align*}
where the last step uses Lemma~\ref{lem:Pthetaj_thetak_nonzero}. Thus
\begin{equation}\label{eq:Pthetaj_nonzero_upper}
\frac{\PP(\hat{\theta}_j \neq 0)}{2(1 - \Phi(\lambda))} - 1 
\le 2m^2 \exp\big\{-\zeta(\rho_\infty)\lambda^2\big\}.
\end{equation}
Combining \eqref{eq:Pthetaj_nonzero_lower} and \eqref{eq:Pthetaj_nonzero_upper} yields the claim.
\end{proof}

\begin{lemma}\label{lem:Pthetaj_nonzero_strong_signal}
Fix $\xi > 0$. Assume $|\theta_j| > (1 + \xi)\lambda $ for some $\lambda > 1$ and
$\vct\theta_{-j} = \vct 0$. Under Assumption~A3,
\[
1 - \PP(\hat{\theta}_j\neq 0)
\le (1-\Phi(\xi \lambda)) + 2m(1 - \Phi(\lambda)).
\]
\end{lemma}
\begin{proof}
By Lemma~\ref{lem:thetaj_nonzero},
\begin{align*}
\PP(\hat{\theta}_j \neq 0)
&\ge \PP(|Z_j| > \lambda, \|\vct Z_{-j}\|_{\infty} < \lambda)\\
&= \PP(|Z_j| > \lambda) - \PP(|Z_j| > \lambda, \|\vct Z_{-j}\|_{\infty} > \lambda)\\
&\ge \PP(|Z_j| > \lambda) - \PP(\|\vct Z_{-j}\|_{\infty} > \lambda).
\end{align*}
By the union bound,
\[
\PP(\|\vct Z_{-j}\|_{\infty} > \lambda)
\le \sum_{k\neq j}\PP(|Z_k| > \lambda)
\le 2m(1 - \Phi(\lambda)).
\]
Since $Z_j \sim \cN(\theta_j, 1)$ and $|\theta_j|>(1+\xi)\lambda$,
\[
\PP(|Z_j| < \lambda)
= \Phi(\lambda - \theta_j) - \Phi(-\lambda - \theta_j)
\le 1-\Phi(\xi\lambda).
\]
Combining these bounds yields
\[
1 - \PP(\hat{\theta}_j\neq 0)
\le (1-\Phi(\xi \lambda)) + 2m(1 - \Phi(\lambda)).
\]
\end{proof}

\subsubsection{Bias and variance analysis}

\begin{thm}\label{thm:bias_variance_single_block}
Let $W_j = \phi_j(p_j)\,\pi_j(\vct Z_{-j})$ and $W = \sum_{j=1}^m W_j$. Further let
$R = \sum_{j=1}^m \one\{\hat{\theta}_j \neq 0\}$. Assume $\lambda = \sqrt{c\log d}$, where $d$ is
sufficiently large so that $\lambda > 1$, and $m = O(d^{\omega})$ with
\[
0\le \omega < c\cdot \min\left\{\frac{1 - \rho_{\infty}}{2(1+ \rho_\infty)}, 
\frac{\zeta(\rho_{\infty})}{3}, 
\frac{(1 - \rho_1)^2}{4\,\tau_{\max}^2}, 
\frac{\tau_{\min}^2}{4(2 - \tau_{\min}^2)}\right\}.
\]
Then:
\begin{itemize}
\item For a null block in $\cB_0$,
\[
\EE[W] = \EE[R] = 2m(1 - \Phi(\lambda))(1 + o(1)),
\]
and
\[
\Var(W) = o(\EE[W]), 
\qquad 
\Var(R) = 2m(1 - \Phi(\lambda))(1 + o(1)).
\]
\item For a ``good'' non-null block in $\cB_{1g}$,
\[
\EE[W] = o(1), \quad \EE[R] = 1+ o(1),
\]
and
\[
\Var(W) = o(1), \quad \Var(R) = o(1).
\]
\end{itemize}
All $o(1)$ terms above are uniform in 
$\tau_{\min}, \tau_{\max}, \rho_1, \rho_\infty, c,$ and $\omega$.
\end{thm}
\begin{proof}
First note that for any $j$ with $\theta_j = 0$, $p_j \sim \mathrm{Unif}([0,1])$ and $p_j$ is
independent of $\vct Z_{-j}$. Thus
\begin{equation}\label{eq:Wj_first_moment_null}
\EE[W_j] = \EE[\phi_j(p_j)]\EE[\pi_j(\vct Z_{-j})] 
= \EE[\pi_j(\vct Z_{-j})] 
= \PP(\hat{\theta}_j \neq 0),
\end{equation}
where the last equality uses the iterated law of expectation and the definition of $\pi_j$. Also,
\begin{equation}\label{eq:Wj_second_moment_null}
\EE[W_j^2] 
= \EE[\phi_j^2(p_j)]\EE[\pi^2_j(\vct Z_{-j})] 
= \frac{1}{1 - \zeta}\,\EE[\pi^2_j(\vct Z_{-j})],
\end{equation}
since $\phi_j(p_j)^2 = \one\{p_j > \zeta\}/(1-\zeta)^2$ and $p_j\sim\mathrm{Unif}[0,1]$.

\medskip
\noindent\textbf{Case 1. Null block.}
For a null block, Lemma~\ref{lem:Pthetaj_nonzero_null_block} and \eqref{eq:Wj_first_moment_null} give
\begin{align*}
\left|\frac{\EE[W]}{2m(1 - \Phi(\lambda))} - 1\right|
&= \left|\frac{\EE[R]}{2m(1 - \Phi(\lambda))} - 1\right| \\
&\le 24 m \exp\left\{-\frac{1 - \rho_{\infty}}{2(1+ \rho_\infty)}\lambda^2\right\}
 + 2m^2 \exp\left\{-\zeta(\rho_\infty)\lambda^2\right\}.
\end{align*}
Using $\lambda^2 = c\log d$ and $m = O(d^\omega)$, each term is of order
$d^{\omega - c(1-\rho_\infty)/2(1+\rho_\infty)}$ and $d^{2\omega - c\zeta(\rho_\infty)}$, which both
tend to $0$ under the condition on $\omega$. This yields
\[
\EE[W] = \EE[R] = 2m(1 - \Phi(\lambda))(1 + o(1)).
\]

For the variance of $W$, by \eqref{eq:Wj_second_moment_null},
\[
\Var(W) \le m\sum_{j=1}^m\Var(W_j)
\le m\sum_{j=1}^m\EE(W_j^2)
\le \frac{m}{1-\zeta}\sum_{j=1}^m\EE[\pi_j^2(\vct Z_{-j})].
\]
Lemma~\ref{lem:pij_moments} gives
\begin{align*}
\frac{\Var(W)}{2m (1 - \Phi(\lambda))}
&\preceq \frac{m^2}{\lambda} \exp\left\{-\left( \frac{(1 - \rho_1)^2}{\tau_{\max}^2}\wedge
\frac{\tau_{\min}^2}{2(2 - \tau_{\min}^2)}\right)\lambda^2\right\}\\
& \preceq \exp\left\{-\left[c\left( \frac{(1 - \rho_1)^2}{\tau_{\max}^2}\wedge \frac{\tau_{\min}^2}{2(2 - \tau_{\min}^2)}\right) - 2\omega\right]\log d\right\}
= o(1),
\end{align*}
again using the constraint on $\omega$. Thus $\Var(W) = o(\EE[W])$.

For $R$, we use
\[
\Var(R) \le \EE[R^2]
= \sum_{j=1}^m\PP(\hat{\theta}_j\neq 0)
 + \sum_{j\neq k}\PP(\hat{\theta}_j \neq 0, \hat{\theta}_k \neq 0)
= \EE[R] + \sum_{j\neq k}\PP(\hat{\theta}_j \neq 0, \hat{\theta}_k \neq 0).
\]
Lemma~\ref{lem:Pthetaj_thetak_nonzero} gives
\[
\frac{\sum_{j\neq k}\PP(\hat{\theta}_j \neq 0, \hat{\theta}_k \neq 0)}{2m(1-\Phi(\lambda))}
\preceq m^3\exp\{-\zeta(\rho_\infty)\lambda^2\}
= d^{\,3\omega - c\zeta(\rho_\infty)} = o(1),
\]
again under the condition on $\omega$. Combining this with the expression for $\EE[R]$ yields
$\Var(R) = 2m(1-\Phi(\lambda))(1+o(1))$.

\medskip
\noindent\textbf{Case 2. Good non-null block.}
Assume, without loss of generality, that $\theta_1 \neq 0$ with
$|\theta_1| > (1 + \xi)\lambda$ and $\theta_j = 0$ for all $j > 1$.

For $j\ge 2$, Lemma~\ref{lem:Pthetaj_nonzero_null_coord} and
\eqref{eq:Wj_first_moment_null} give
\begin{align}
\sum_{j=2}^{m}\EE[W_j] = \sum_{j=2}^{m}\PP(\hat{\theta}_j \neq 0)
&\preceq \frac{m}{\lambda}\exp\left\{-\frac{(1 - \rho_1)^2}{2\,\tau_{\max}^2}\lambda^2\right\}
\nonumber\\
&\preceq \exp\left\{\left(c\frac{(1 - \rho_1)^2}{2 \tau_{\max}^2} - \omega\right)\log d\right\}
= o(1).\label{eq:nonnull_block_nulls}
\end{align}
For $j = 1$,
\[
\EE[W_1] \;\le\; \EE[\phi_j(p_j)] \;\preceq\; \PP(p_1 > \zeta) \;=\; o(1)
\]
by Mill's bound. Furthermore, by Lemma~\ref{lem:Pthetaj_nonzero_strong_signal},
\begin{equation}\label{eq:1-Pthetaj_nonzero}
\big|1 - \PP(\hat{\theta}_1\neq 0)\big| 
\le (1-\Phi(\xi \lambda)) + 2m(1 - \Phi(\lambda)) = o(1).
\end{equation}
Thus
\[
\EE[W] = o(1),
\qquad
|\EE[R] - 1| = o(1).
\]

For the variances, Lemma~\ref{lem:pij_moments} yields
\begin{align*}
\Var(W) 
&\le \EE[W^2]
\le m\sum_{j=1}^m\EE[\pi_j^2(\vct Z_{-j})] \\
&\preceq \frac{m^2}{\lambda^2}
\exp\left\{-\frac{2(1 - \rho_1)^2}{2\,\tau_{\max}^2}\lambda^2\right\}
\preceq \exp\left\{2\left(c\frac{(1 - \rho_1)^2}{2\tau_{\max}^2} - \omega\right)\log d\right\}
= o(1),
\end{align*}
under the stated constraint on $\omega$.

For $R$, write
\[
\Var(R)\le 2\Var(\one\{\hat{\theta}_1\neq 0\})
 + 2\Var\Big(\sum_{j=2}^m \one\{\hat{\theta}_j\neq 0\}\Big).
\]
By the same reasoning as in \eqref{eq:nonnull_block_nulls}, one has
\[
\Var\Big(\sum_{j=2}^m \one\{\hat{\theta}_j\neq 0\}\Big) = o(1),
\]
and by \eqref{eq:1-Pthetaj_nonzero}, $\Var(\one\{\hat{\theta}_1\neq 0\}) = o(1)$ as well. Hence
$\Var(R) = o(1)$.

This completes the proof.
\end{proof}

\subsection{Proof of Theorems \ref{thm:CLT_FDR} and \ref{thm:bootstrap}}

\begin{proof}[Proof of Theorem~\ref{thm:CLT_FDR}]
We reintroduce block indices and write
\[
R = \sum_{b=1}^B R_{(b)}, \qquad 
W = \sum_{b=1}^B W_{(b)},
\]
where $R_{(b)}$ and $W_{(b)}$ are, respectively, the blockwise model size and blockwise contribution to $W$ for block $b$. In this section, all $o(1)$, $o_\PP(1)$, $O(1)$, and $O_\PP(1)$ terms are uniform given the constants in Assumptions A1-A5.

\medskip\noindent
\textbf{Step 1: Mean and variance of $R$.}

By Theorem~\ref{thm:bias_variance_single_block}, for each block $b$ we have
\[
\EE[R_{(b)}] = 
\begin{cases}
\displaystyle \frac{v_d\, m_b}{d}\,(1 + o(1)), & b\in \cB_0,\\[0.7em]
1 + o(1), & b\in \cB_{1g},\\[0.25em]
O(1), & b\in \cB_{1b},
\end{cases}
\qquad
\Var(R_{(b)}) = 
\begin{cases}
\displaystyle \frac{v_d\, m_b}{d}\,(1 + o(1)), & b\in \cB_0,\\[0.7em]
o(1), & b\in \cB_{1g},\\[0.25em]
O(1), & b\in \cB_{1b}.
\end{cases}
\]
Here the $o(1)$ terms are uniform over blocks.

Since $R = \sum_{b=1}^B R_{(b)}$, we obtain
\begin{align*}
\EE[R]
&= \sum_{b\in \cB_0} \frac{v_d m_b}{d}\,(1 + o(1))
   + \sum_{b\in \cB_{1g}} (1 + o(1))
   + O(B_{1b}).
\end{align*}
Note that
\[
\sum_{b\in \cB_0} m_b = d - \sum_{b\in \cB_{1g}\cup \cB_{1b}} m_b.
\]
By bounded block size and Assumption~A4,
\[
\Bigg|\sum_{b\in \cB_0} m_b - d\Bigg|
\;\le\; m_{\max}(B_{1g} + B_{1b})
\;=\; o(d),
\]
so
\[
\sum_{b\in \cB_0} m_b = d\,(1+o(1)).
\]
Furthermore, $B_{1b}=o(B_{1g})$ by Assumption~A4, and $B_{1g} = \Theta(v_d)$ by Assumption~A5. Hence
\begin{equation}\label{eq:E_R}
\EE[R] 
= v_d\,(1+o(1)) + B_{1g}\,(1+o(1))
= (v_d + B_{1g})(1 + o(1)).
\end{equation}

By block orthogonality (Assumption~A1), $R_{(b)}$ are independent across $b$, so
\[
\Var(R) = \sum_{b=1}^B \Var(R_{(b)}).
\]
From the blockwise bounds above and the same counting argument as for the mean,
\begin{equation}\label{eq:Var_R}
\Var(R) = v_d\,(1+o(1)) + O(B_{1g} + B_{1b})
= (v_d + B_{1g})(1 + o(1)).
\end{equation}
Assumption~A5 states that $v_d$ is of the same order as $B_{1g}$, so
\[
\Var(R) = \Theta(B_{1g}), \qquad v_d = \Theta(B_{1g}).
\]

\medskip\noindent
\textbf{Step 2: CLT for $R$.}

By the Berry–Esseen theorem for sums of independent variables,
\[
d_K\!\left(\frac{R - \EE[R]}{\sqrt{\Var(R)}},\ \cN(0,1)\right)
\;\le\; 
\frac{\sum_{b=1}^B \EE|R_{(b)} - \EE[R_{(b)}]|^3}{\Var(R)^{3/2}}.
\]
Since $0 \le R_{(b)} \le m_b \le m_{\max}$,
\[
|R_{(b)} - \EE[R_{(b)}]|^3 
\;\le\; m_{\max}\,(R_{(b)} - \EE[R_{(b)}])^2,
\]
so
\[
\sum_{b=1}^B \EE|R_{(b)} - \EE[R_{(b)}]|^3
\;\le\; m_{\max}\sum_{b=1}^B \Var(R_{(b)})
= m_{\max}\Var(R).
\]
Hence
\begin{equation} \label{eq:R_CLT}
d_K\!\left(\frac{R - \EE[R]}{\sqrt{\Var(R)}},\ \cN(0,1)\right)
\;\le\; \frac{m_{\max}\Var(R)}{\Var(R)^{3/2}}
= \frac{m_{\max}}{\sqrt{\Var(R)}}
= O\!\left(\frac{1}{\sqrt{B_{1g}}}\right) = o(1).
\end{equation}
where the last step uses Assumption A4 which implies $B_{1g}\rightarrow \infty$.

\medskip\noindent
\textbf{Step 3: Mean and variance of $W$.}

By Theorem~\ref{thm:bias_variance_single_block}, for each block $b$,
\[
\EE[W_{(b)}] = 
\begin{cases}
\displaystyle \frac{v_d m_b}{d}\,(1 + o(1)), & b\in \cB_0,\\[0.7em]
o(1), & b\in \cB_{1g},\\[0.25em]
O(1), & b\in \cB_{1b},
\end{cases}
\qquad
\Var(W_{(b)}) = 
\begin{cases}
\displaystyle \frac{v_d m_b}{d}\cdot o(1), & b\in \cB_0,\\[0.7em]
o(1), & b\in \cB_{1g},\\[0.25em]
O(1), & b\in \cB_{1b}.
\end{cases}
\]
Summing over blocks and using the same counting argument as for $R$,
\begin{equation}\label{eq:bias_variance_W}
\EE[W] = v_d(1 + o(1)), \qquad
\Var(W) = o(v_d) = o(B_{1g}),
\end{equation}
where the last equality follows from $v_d = \Theta(B_{1g})$.

\medskip\noindent
\textbf{Step 4: Approximating $\hFDR$ by $W/R$.}

Recall $\vct Z_{-(b)} = \set{\vct Z_{(b')}}_{b' \neq b}$ is the subfficient statistic except for block $b$. 
Write
\[
\hFDR = \sum_{b=1}^B\sum_{j \in [m_b]} \hFDR_{(b),j},
\]
where for $j \in [m_b]$,
\begin{align*}
\hFDR_{(b), j}
&= \phi_{(b), j}(p_{(b), j})\cdot 
   \EE_{\theta_{(b),j}=0}\!\left[ \frac{\one\{\hat\theta_{(b),j} \neq 0\}}{R} \;\middle|\; \vct Z_{(b),-j},\, \vct Z_{-(b)}\right]\\
&= \phi_{(b), j}(p_{(b), j})\cdot
   \PP_{\theta_{(b),j}=0}\!\left(\hat\theta_{(b),j} \neq 0 \mid \vct Z_{(b),-j}\right)\cdot
   \EE_{\theta_{(b),j}=0}\!\left[\frac{1}{R} \;\middle|\; \hat\theta_{(b),j} \neq 0,\ \vct Z_{(b),-j},\, \vct Z_{-(b)}\right]\\
&= \frac{\phi_{(b), j}(p_{(b), j})\,\pi_{(b),j}(\vct Z_{(b),-j})}{R + U_{(b),j}}\\
&= \frac{W_{(b),j}}{R + U_{(b),j}},
\end{align*}
where $|U_{(b),j}|\le m_b\le m_{\max}$; the perturbation $U_{(b),j}$ accounts for the fact that
perturbating $Z_{(b),j}$ changes $R$ by at most $m_b$.

Define
\[
\widetilde{\mathrm{FDR}} = \frac{W}{R}.
\]
Then
\begin{align*}
\hFDR - \widetilde{\mathrm{FDR}}
&= \sum_{b=1}^B \sum_{j \in [m_b]}
   \left(
     \frac{W_{(b),j}}{R + U_{(b),j}} - \frac{W_{(b),j}}{R}
   \right)\\
&= -\sum_{b=1}^B \sum_{j \in [m_b]}
   \frac{W_{(b),j} U_{(b),j}}{R(R+U_{(b),j})},
\end{align*}
so, on the event $\{R > 2 m_{\max}\}$,
\[
\bigl|\hFDR - \widetilde{\mathrm{FDR}}\bigr|
\le \sum_{b,j} \frac{2 m_{\max} W_{(b),j}}{R^2}
= \frac{2 m_{\max} W}{R^2}.
\]
Since $\EE[R]$, $\Var(R)$, and $\EE[W]$ are $\Theta(v_d)$, by Markov inequality,
\begin{equation}\label{eq:hFDR_proxy}
\bigl|\hFDR - \widetilde{\mathrm{FDR}}\bigr|
= O_{\PP}\!\left(\frac{1}{v_d}\right).
\end{equation}

\medskip\noindent
\textbf{Step 5: Centering and linearization.}

Let
\[
\tilde{\alpha}_d = \frac{\EE[W]}{\EE[R]}.
\]
By \eqref{eq:E_R} and \eqref{eq:bias_variance_W},
\[
\tilde{\alpha}_d
= \frac{v_d(1+o(1))}{v_d + B_{1g}}(1+o(1))
= \alpha_d (1 + o(1)).
\]

We now expand
\[
\widetilde{\text{FDR}} - \tilde{\alpha}_d
= \frac{W}{R} - \frac{\EE[W]}{\EE[R]}
= \frac{W - \EE[W]}{R}
 - \frac{(R - \EE[R])\,\EE[W]}{R\,\EE[R]}.
\]

For the first term, by Chebyshev’s inequality, \eqref{eq:bias_variance_W}, and the fact that $R$ is
of order $v_d$,
\[
\frac{W - \EE[W]}{R}
= O_{\PP}\!\left(\frac{\sqrt{\Var(W)}}{R}\right)
= o_{\PP}\!\left(\frac{1}{\sqrt{v_d}}\right).
\]
For the second term, by \eqref{eq:E_R},  \eqref{eq:Var_R}, and the Markov's inequality,
\[
\frac{\EE[R]}{R} = 1 - \frac{R - \EE[R]}{R} = 1 + o_{\PP}(1).
\]
Thus
\[
\frac{1}{R} = \frac{1}{\EE[R]} (1 + o_{\PP}(1))
\]
and
\[
\frac{(R - \EE[R])\,\EE[W]}{R\,\EE[R]}
= \tilde{\alpha}_d\;\frac{R - \EE[R]}{R}
= \tilde{\alpha}_d\;\frac{R - \EE[R]}{\EE[R]}\,(1 + o_{\PP}(1)).
\]

Putting this together,
\[
\sqrt{v_d}\,(\widetilde{\text{FDR}} - \tilde{\alpha}_d)
= -\tilde{\alpha}_d\,\sqrt{v_d}\,\frac{R - \EE[R]}{\EE[R]} + o_{\PP}(1).
\]

Multiplying and dividing by $\sqrt{\Var(R)}$ and using \eqref{eq:E_R}–\eqref{eq:Var_R}, we obtain
\begin{align*}
\sqrt{v_d}\,(\widetilde{\text{FDR}} - \tilde{\alpha}_d)
&= -\tilde{\alpha}_d\,\frac{\sqrt{v_d}\,\sqrt{\Var(R)}}{\EE[R]}\cdot
    \frac{R - \EE[R]}{\sqrt{\Var(R)}} + o_{\PP}(1)\\
&= -\tilde{\alpha}_d\,\sqrt{\alpha_d}\,
    \frac{R - \EE[R]}{\sqrt{\Var(R)}} + o_{\PP}(1),
\end{align*}
where the second equation uses
\[
\frac{\sqrt{v_d}\,\sqrt{\Var(R)}}{\EE[R]}
= \frac{\sqrt{v_d}\,\sqrt{v_d + B_{1g}}}{v_d + B_{1g}}(1+o(1))
= \sqrt{\frac{v_d}{v_d + B_{1g}}}(1+o(1))
= \sqrt{\alpha_d}\,(1+o(1)).
\]
By Assumption A5,
\begin{equation}\label{eq:alpha_d}
\alpha_d \in \left[\chi, 1 -\chi\right] = \Theta(1), \quad \tilde{\alpha}_d = \Theta(1).
\end{equation}
Therefore,
\[
\sqrt{\frac{v_d}{\alpha_d^3}}\;(\widetilde{\text{FDR}} - \tilde{\alpha}_d)
= -\frac{\tilde{\alpha}_d}{\alpha_d}\;
   \frac{R - \EE[R]}{\sqrt{\Var(R)}} + o_{\PP}(1).
\]
Recalling that all $o_\PP(1)$ terms are uniform and $\Phi$ is Lipschitz,  
\eqref{eq:R_CLT} implies that 
\[d_{K}\left(\sqrt{\frac{v_d}{\alpha_d^3}}\;(\widetilde{\text{FDR}} - \tilde{\alpha}_d)
,\cN(0,1)\right) = o(1).\]
% the last factor converges in distribution to $\cN(0,1)$, and by
% $\tilde{\alpha}_d/\alpha_d \to 1$, Slutsky’s theorem implies
% \[
% \sqrt{\frac{v_d}{\alpha_d^3}}\;(\widetilde{\text{FDR}} - \tilde{\alpha}_d)
% \;\Rightarrow\; \cN(0,1).
% \]

Finally, by \eqref{eq:hFDR_proxy},
\[
\sqrt{\frac{v_d}{\alpha_d^3}}\;(\hFDR - \tilde{\alpha}_d)
= \sqrt{\frac{v_d}{\alpha_d^3}}\;(\widetilde{\text{FDR}} - \tilde{\alpha}_d)
 + o_\PP\left(\frac{1}{\sqrt{v_d}}\right),
\]
so the same limit holds:
\[
d_{K}\left(\sqrt{\frac{v_d}{\alpha_d^3}}\;(\hFDR - \tilde{\alpha}_d)
,\cN(0,1)\right) = o(1).
\]
\end{proof}

\begin{proof}[Proof of Theorem~\ref{thm:bootstrap}]
We work throughout under the assumptions of Theorem~\ref{thm:CLT_FDR}, with the only
change being that the Lasso penalty defining the OLS-after-Lasso model is
$\bar{\lambda} = \bar{c}\lambda$ with $\bar{c} > 1$. We will show that, with high probability,
the estimated OLS-after-Lasso model satisfies the same structural assumptions (A1–A5)
with $\hat{\vct\theta}^{(*)}$ (the OLS-after-Lasso coefficients) playing the role of the true
coefficients. Conditional on such a ``good'' data event, we can then apply
Theorem~\ref{thm:CLT_FDR} to the bootstrap procedure.

\medskip\noindent
\textbf{Step 1: Null blocks contribute negligibly under the larger penalty.}

By Theorem~\ref{thm:bias_variance_single_block}, applied with penalty level $\bar{\lambda}$,
for any null block $b\in \cB_0$ we have
\[
\EE[R_{(b)}] = 2 m_b\bigl(1 - \Phi(\bar{\lambda})\bigr)\,(1 + o(1)).
\]
For sufficiently large $d$ we have $\bar{\lambda} = \bar{c}\lambda > \lambda > 1$.
By the upper Mill's ratio bound (Proposition~\ref{prop:mill_ratio}),
\[
1 - \Phi(\bar{\lambda})
\;\le\;
\frac{1}{\bar{\lambda}}\exp\!\left(-\frac{\bar{\lambda}^2}{2}\right)
= \frac{1}{\bar{c}\lambda}\exp\!\left(-\frac{\bar{c}^2\lambda^2}{2}\right).
\]
Write
\[
1 - \Phi(\lambda) = \frac{v_d}{2d},
\]
as in Assumption~A5. Using again the Mill's ratio approximation, we have
\[
\lambda^2 = 2\log(d/v_d) + O(\log(\lambda)).
\]
Therefore,
\[
\exp\!\left(-\frac{\bar{c}^2\lambda^2}{2}\right)
= \exp\!\bigl(-\bar{c}^2\log(d/v_d)\bigr) \cdot O(\lambda)
= (d/v_d)^{-\bar{c}^2} \cdot O(\lambda).
\]
Therefore,
\[
1 - \Phi(\bar{\lambda})
= O\!\left(\Bigl(\frac{v_d}{d}\Bigr)^{\bar{c}^2}\right).
\]

Summing over all null blocks and using $\sum_{b\in\cB_0} m_b \asymp d$, we obtain
\[
\EE\Biggl[\;\sum_{b\in \cB_0} R_{(b)}\Biggr]
= O\!\left(d \cdot \Bigl(\frac{v_d}{d}\Bigr)^{\bar{c}^2}\right)
= O\!\left(\frac{v_d^{\bar{c}^2}}{d^{\bar{c}^2-1}}\right)
= o(v_d),
\]
since $v_d/d\to 0$ and $\bar{c}^2>1$.

\medskip\noindent
\textbf{Step 2: Good blocks are recovered (support-wise) with high probability.}

Without loss of generality, for each good non-null block $b\in\cB_{1g}$ we relabel
coordinates so that
\[
\theta_{(b),1}\neq 0, \qquad \theta_{(b),-1} = \vct 0.
\]
Note that
\[
|\theta_j| > (1 + \xi)\lambda = \frac{1+\xi}{\bar{c}}\bar{\lambda} > \bar{\lambda}.
\]

Apply the single-block analysis (Theorem~\ref{thm:bias_variance_single_block} and its lemmas)
with penalty $\bar{\lambda}$; the same structure carries through since the within-block
regularity assumption A3 is unchanged and $\bar{\lambda}\to\infty$ as $d\to\infty$.
In particular, by equations~\eqref{eq:nonnull_block_nulls} and
\eqref{eq:1-Pthetaj_nonzero} (applied with $\lambda$ replaced by $\bar{\lambda}$),
we have, for each such good block,
\[
\PP\bigl(\hat{\theta}_{(b),-1}^{\bar{\lambda}}\neq 0\bigr) = o(1),
\qquad
\PP\bigl(\hat{\theta}_{(b),1}^{\bar{\lambda}} = 0\bigr) = o(1),
\]
where $\hat{\vct\theta}^{\bar{\lambda}}$ is the Lasso estimator at level $\bar{\lambda}$.
Hence
\[
\EE\Biggl[\sum_{b\in \cB_{1g}}
  \one\Bigl(\hat{\theta}_{(b),-1}^{\bar{\lambda}}\neq 0
       \ \text{or}\ 
       \hat{\theta}_{(b), 1}^{\bar{\lambda}} = 0\Bigr)\Biggr]
= o(B_{1g}).
\]

Combining this with the bound from Step~1 and using Markov's inequality, there exists
a deterministic sequence $\iota_d \to 0$ and an event $\cE_d$ with
\[
\PP(\cE_d) = 1 - o(1),
\]
such that on $\cE_d$,
\[
\sum_{b\in \cB_0} R_{(b)}^{\bar{\lambda}} \le \iota_d v_d,
\qquad
\sum_{b\in \cB_{1g}}
\one\Bigl(\hat{\theta}_{(b),-1}^{\bar{\lambda}}\neq 0
       \ \text{or}\ 
       \hat{\theta}_{(b), 1}^{\bar{\lambda}} = 0\Bigr)
\le \iota_d B_{1g},
\]
where $R_{(b)}^{\bar{\lambda}}$ is the blockwise model size under Lasso
with penalty $\bar{\lambda}$.

Define the set of ``error blocks'' as
\[
\cB_{\mathrm{err}}
=
\bigl\{b\in \cB_0:\ R_{(b)}^{\bar{\lambda}} > 0\bigr\}
\;\cup\;
\Bigl\{b\in \cB_{1g}:\ 
  \hat{\theta}_{(b),-1}^{\bar{\lambda}}\neq 0 \ \text{or}\ 
  \hat{\theta}_{(b), 1}^{\bar{\lambda}} = 0
\Bigr\}.
\]
Then on $\cE_d$,
\[
|\cB_{\mathrm{err}}|
\;\le\; \iota_d (v_d + B_{1g})
\;=\; o(B_{1g}),
\]
using Assumption~A5 to relate $v_d$ and $B_{1g}$.

We now define a perturbed block partition:
\[
\tilde{\cB}_0 = \cB_0 \setminus \cB_{\mathrm{err}}, \qquad
\tilde{\cB}_{1g} = \cB_{1g} \setminus \cB_{\mathrm{err}}, \qquad
\tilde{\cB}_{1b} = \cB_{1b} \cup \cB_{\mathrm{err}}.
\]
On $\cE_d$, we then have
\[
|\tilde{\cB}_0| = B_0 + o(B_{1g}) = B_0(1 + o(1)),
\]
\[
|\tilde{\cB}_{1g}| = B_{1g} + o(B_{1g}) = B_{1g}(1 + o(1)),
\]
\[
|\tilde{\cB}_{1b}| = B_{1b} + o(B_{1g}) = o(B_{1g}),
\]
so the sparsity/purity structure of Assumption~A4 is preserved up to $o(B_{1g})$ blocks. On $\cE_d$, the OLS-after-Lasso estimator $\hat{\vct\theta}^{(*)}$ has the \emph{same support}
as $\hat{\vct\theta}^{\bar{\lambda}}$ across the null and good non-null blocks.

\medskip\noindent
\textbf{Step 3: Lower bound on the estimated signal in good blocks.}

For any good block $b\in \tilde{\cB}_{1g}$, on $\cE_d$,
the Lasso at level $\bar{\lambda}$ selects only coordinate 1 in block $b$.
The OLS-after-Lasso coefficient on that coordinate is simply the least squares coefficient
in a univariate regression, and hence
\[
\hat{\theta}_{(b),1}^{(*)} = Z_{(b),1},
\]
where $Z_{(b),1} = \vct X_{(b),1}^\top\vct Y$.

Fix any $\iota > 0$. Using a union bound and the Gaussian tail bound, we have
\begin{align*}
\PP\Biggl(
\max_{b\in \cB_{1g}}
  \bigl|\hat{\theta}_{(b),1}^{(*)} - \theta_{(b),1}\bigr|
  > \sqrt{2(1 + \iota)\log B_{1g}}
\Biggr)
&\le
B_{1g} \exp\bigl(-(1 + \iota)\log B_{1g}\bigr)\\
&= B_{1g}^{-\iota}
\;=\; o(1),
\end{align*}
since $B_{1g}\to\infty$. Thus there exists an event $\tilde{\cE}_d$ with
\[
\PP(\tilde{\cE}_d) = 1 - o(1),
\]
such that on $\tilde{\cE}_d$,
\[
\max_{b\in \cB_{1g}}
  \bigl|\hat{\theta}_{(b),1}^{(*)} - \theta_{(b),1}\bigr|
\;\le\;
\sqrt{2(1 + \iota)\log B_{1g}}.
\]
Therefore, on $\tilde{\cE}_d$,
\[
\min_{b\in \cB_{1g}}
|\hat{\theta}_{(b),1}^{(*)}|
\;\ge\;
(1 + \xi)\lambda - \sqrt{2(1 + \iota)\log B_{1g}}.
\]

We now relate $\lambda$ and $\log B_{1g}$.
By the Mill's ratio expansion, we have
\[
\lambda^2 = 2 \log\Bigl(\frac{d}{v_d}\Bigr)\,(1 + o(1)).
\]
By Assumption~A5, $v_d = \Theta(B_{1g})$. For sufficiently large $d$, it follows that
\[
\lambda^2 \in \left[(2 - 2 \iota)\log\Bigl(\frac{d}{B_{1g}}\Bigr),\ 
                   (2 + 2 \iota)\log\Bigl(\frac{d}{B_{1g}}\Bigr)\right].
\]
Moreover, recall
\[
B_{1g} = O\!\bigl(d^{\nu}\Gamma_d\bigr), \qquad \Gamma_d = o(d^a)\ \text{for all }a>0,
\]
we write
\[
\log B_{1g} = \nu \log d + \log \Gamma_d + o(\log d), \qquad
\log ({d}/{B_{1g}})
= (1 - \nu)\log d - \log \Gamma_d + o(\log d).
\]
Therefore
\begin{align*}
& \quad \frac{\min_{b\in \cB_{1g}}|\hat{\theta}_{(b),1}^{(*)}|}{\bar{\lambda}}\\
&> \frac{(1 + \xi)\sqrt{2 - 2\iota}\cdot \sqrt{\log (d / B_{1g})} - \sqrt{2(1 + \iota)\log B_{1g}}}{\bar{c}\sqrt{2 + 2\iota}\sqrt{\log (d / B_{1g})}}\\
& = \frac{(1 + \xi)\sqrt{1 - \iota}\cdot \sqrt{(1 -  \nu)\log d - \log \Gamma_d + o(\log d)} - \sqrt{(1 + \iota)(\nu \log d + \log \Gamma_d + o(\log d))}}{\bar{c}\sqrt{1 + \iota}\sqrt{(1 - \nu)\log d - \log \Gamma_d + o(\log d)}}.
\end{align*}
Since $\log \Gamma_d = o(\log d)$, for a sufficiently large $d$,
\begin{align*}
\frac{\min_{b\in \cB_{1g}}|\hat{\theta}_{(b),1}^{(*)}|}{\bar{\lambda}}
&>\;
\frac{(1 + \xi)\sqrt{1 - \nu}\,\sqrt{1 - \iota}
      - \sqrt{\nu}\,(1 + \iota)}
     {\bar{c}\,\sqrt{1 + \iota}\,\sqrt{1 - \nu}}.
\end{align*}
We now choose $\iota>0$ small enough so that
\[
1 + \xi
> \frac{1 + \iota}{\sqrt{1 - \iota}}\left(\sqrt{\frac{\nu}{1 - \nu}} + \bar{c}\right),
\]
which is possible precisely because
\[
\xi > \sqrt{\frac{\nu}{1 - \nu}} + \bar{c} - 1
\]
by assumption. For such an $\iota$, the right-hand side of the previous display
exceeds $\sqrt{1 + \iota} > 1$, and therefore, on $\tilde{\cE}_d$,
\[
\frac{\min_{b\in \cB_{1g}}|\hat{\theta}_{(b),1}^{(*)}|}{\bar{\lambda}} > 1.
\]

Combining with the support-error control from Step~2, we obtain that on
$\cE_d\cap \tilde{\cE}_d$ (an event with probability $1 - o(1)$), every good block in
$\tilde{\cB}_{1g}$ has a single non-zero coordinate whose magnitude exceeds $\bar{\lambda}$,
and all other coordinates in those blocks are zero. Moreover, the total number of ``bad''
blocks $\tilde{\cB}_{1b}$ remains $o(B_{1g})$.

\medskip\noindent
\textbf{Step 4: Applying the CLT in the bootstrap world.}

On the event $\cE_d\cap \tilde{\cE}_d$, the OLS-after-Lasso model with coefficients
$\hat{\vct\theta}^{(*)}$ and block partition
$(\tilde{\cB}_0,\tilde{\cB}_{1g},\tilde{\cB}_{1b})$ satisfies
Assumptions~A1–A5 (with unchanged constants
$\tau_{\min},\tau_{\max},\rho_1,\rho_\infty$, and with the same growth rates for
$B_{1g}$ and the number of bad blocks). Therefore, if we now treat
$\hat{\vct\theta}^{(*)}$ as the ``true'' parameter vector in a parametric bootstrap
model, Theorem~\ref{thm:CLT_FDR} applies to the bootstrap estimator $\hFDR^{(*)}$ with
some deterministic centering
\[
\alpha_d^{(*)}(\vct Z,\mat X) \in (0,1),
\]
which depends on the data only through $\hat{\vct\theta}^{(*)}$ and the design.

Concretely, Theorem~\ref{thm:CLT_FDR}, together with the uniformity of $o(1)$ terms mentioned at the beginning of the section, implies that there exists $u_d = o(1)$ such that 
\[d_K\left(\mathcal{L}\left(\sqrt{\frac{v_d}{\alpha_d^3}}\;\bigl(\hFDR^{(*)} - \alpha_d^{(*)}\bigr)\mid \vct X, \vct Y\right),  \cN(0,1)\right)\le u_d, \quad \text{if }(\vct X, \vct Y)\in \cE_d\cap \tilde{\cE}_d,\]
where $\mathcal{L}(\cdot\mid \cdot)$ denote the conditional law. 
% (applied conditionally on 
% $\cE_d\cap \tilde{\cE}_d$) yields
% \[
% \sqrt{\frac{v_d}{\alpha_d^3}}\;\bigl(\hFDR^{(*)} - \alpha_d^{(*)}\bigr)
% \;\Rightarrow\; \cN(0,1)
% \qquad\text{on }\cE_d\cap \tilde{\cE}_d.
% \]
The proof is then completed by noting that $\PP(\cE_d\cap \tilde{\cE}_d)=1-o(1)$.
\end{proof}

\begin{proof} [Proof of Lemma~\ref{lem:equicorrelated}]
A standard computation shows that
\[
\Sigma_{(b);-j,-j}^{-1}
= \frac{1}{1 - \rho_b}
\left(
\mat I_{m_b-1}
- \frac{\rho_b}{(m_b-2)\rho_b + 1}\,\vct 1_{m_b-1}\vct 1_{m_b-1}^\tran
\right).
\]
Therefore
\[
1 - \Sigma_{(b);j,-j}\Sigma_{(b);-j,-j}^{-1}\Sigma_{(b);-j,j}
= 1 - \frac{(m_b - 1)\rho_b^2}{(m_b - 2)\rho_b + 1}
= (1 - \rho_b)\,\frac{(m_b - 1)\rho_b + 1}{(m_b - 2)\rho_b + 1},
\]
and
\[
\Sigma_{(b);-j,-j}^{-1}\Sigma_{(b);-j,j}
= \frac{\rho_b}{(m_b - 2)\rho_b + 1}\,\vct 1_{m_b - 1}.
\]
\end{proof}

\section{Extended numerical results} \label{app:extended_results}

Figure \ref{fig:hiv_all} shows the extended results of $\hFDR$ along with estimated one-standard-error bars in Section~\ref{sec:real_HIV} and Figure \ref{fig:protein_all} shows the extended results for Section~\ref{sec:real_protein}.

\begin{figure}[tbp]
    \centering
    \begin{subfigure}[b]{0.24\linewidth}
      \centering
      \includegraphics[width=\linewidth]{figs/HIV_drug_expr/1_1-APV.pdf}
      \caption{APV, $201$, $767$}
    \end{subfigure}    
    \begin{subfigure}[b]{0.24\linewidth}
      \centering
      \includegraphics[width=\linewidth]{figs/HIV_drug_expr/1_2-ATV.pdf}
      \caption{ATV, $147$, $328$}
    \end{subfigure} 
    \begin{subfigure}[b]{0.24\linewidth}
      \centering
      \includegraphics[width=\linewidth]{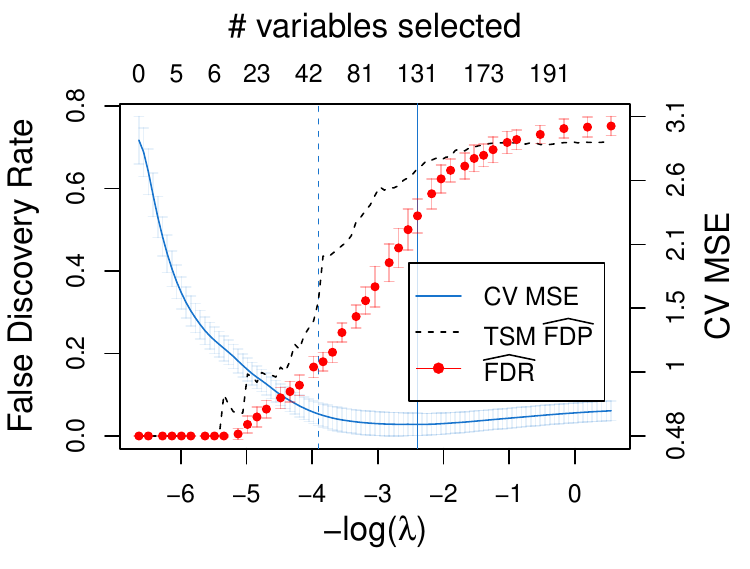}
      \caption{IDV, $206$, $825$}
    \end{subfigure} 
    \begin{subfigure}[b]{0.24\linewidth}
      \centering
      \includegraphics[width=\linewidth]{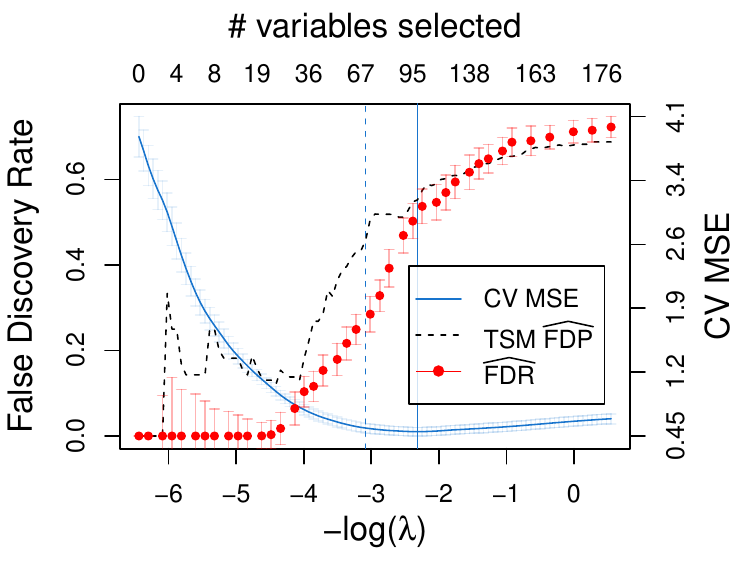}
      \caption{LPV, $148$, $515$}
    \end{subfigure} 
    
    \begin{subfigure}[b]{0.24\linewidth}
      \centering
      \includegraphics[width=\linewidth]{figs/HIV_drug_expr/1_5-NFV.pdf}
      \caption{NFV, $207$, $842$}
    \end{subfigure}    
    \begin{subfigure}[b]{0.24\linewidth}
      \centering
      \includegraphics[width=\linewidth]{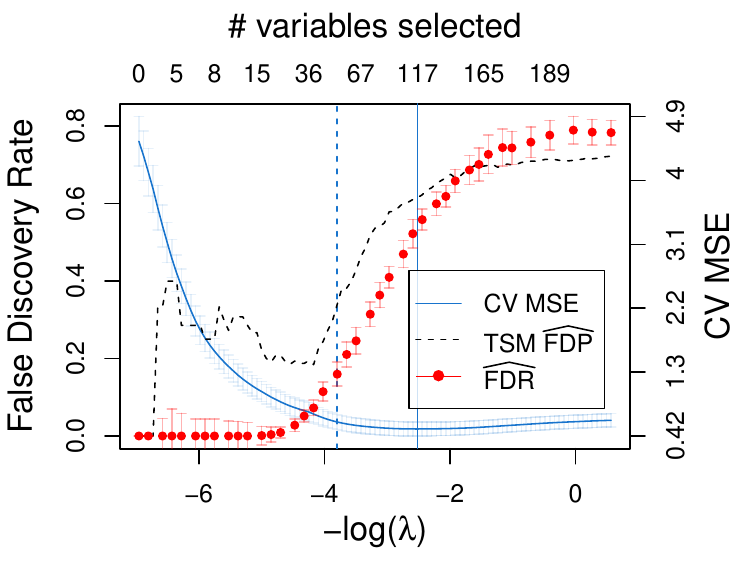}
      \caption{RTV, $205$, $793$}
    \end{subfigure} 
    \begin{subfigure}[b]{0.24\linewidth}
      \centering
      \includegraphics[width=\linewidth]{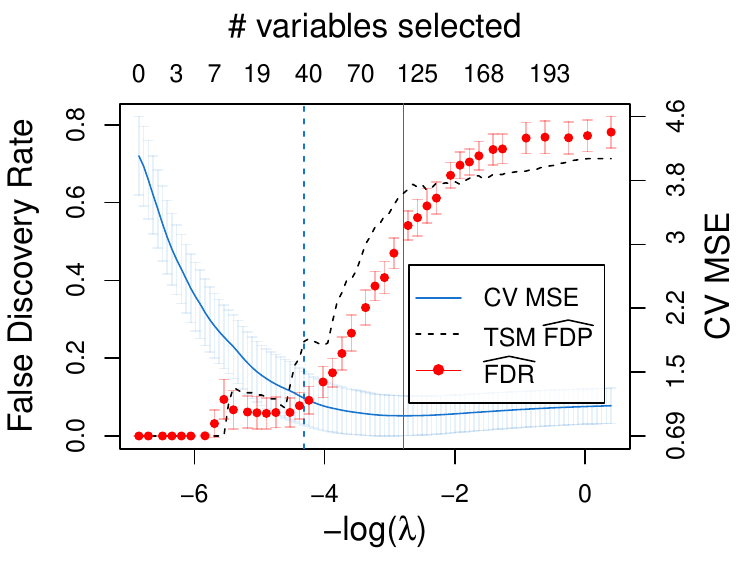}
      \caption{SQV, $206$, $824$}
    \end{subfigure} 
    \begin{subfigure}[b]{0.24\linewidth}
      \centering
      \includegraphics[width=\linewidth]{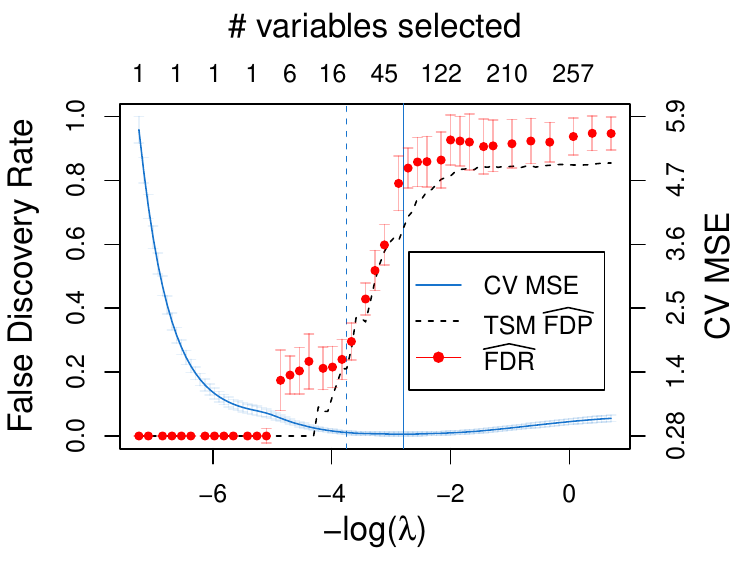}
      \caption{X3TC, $283$, $629$}
    \end{subfigure} 
    
    \begin{subfigure}[b]{0.24\linewidth}
      \centering
      \includegraphics[width=\linewidth]{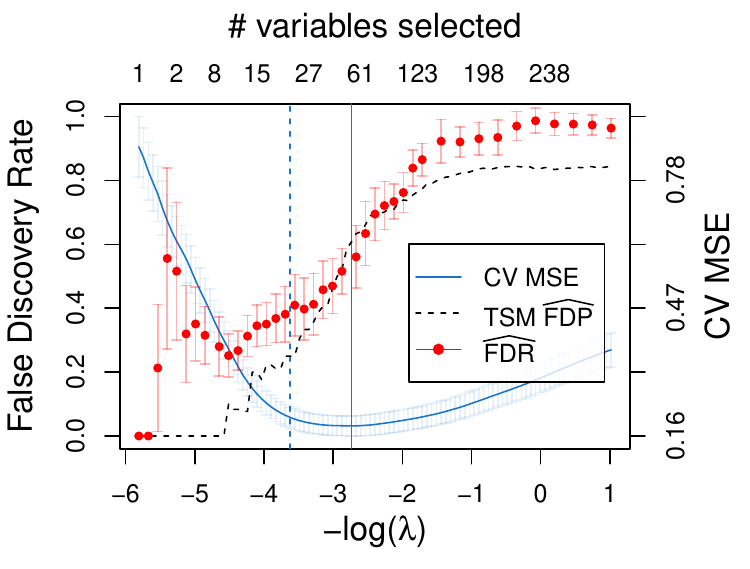}
      \caption{ABC, $283$, $623$}
    \end{subfigure}    
    \begin{subfigure}[b]{0.24\linewidth}
      \centering
      \includegraphics[width=\linewidth]{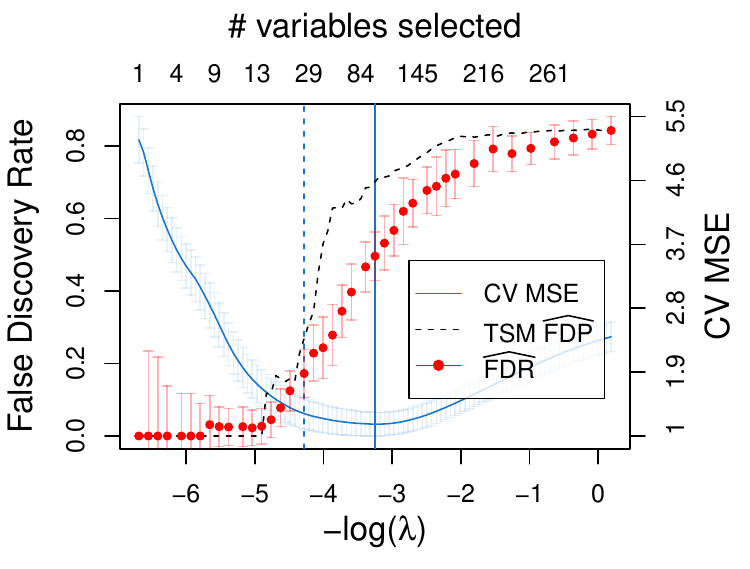}
      \caption{AZT, $283$, $626$}
    \end{subfigure} 
    \begin{subfigure}[b]{0.24\linewidth}
      \centering
      \includegraphics[width=\linewidth]{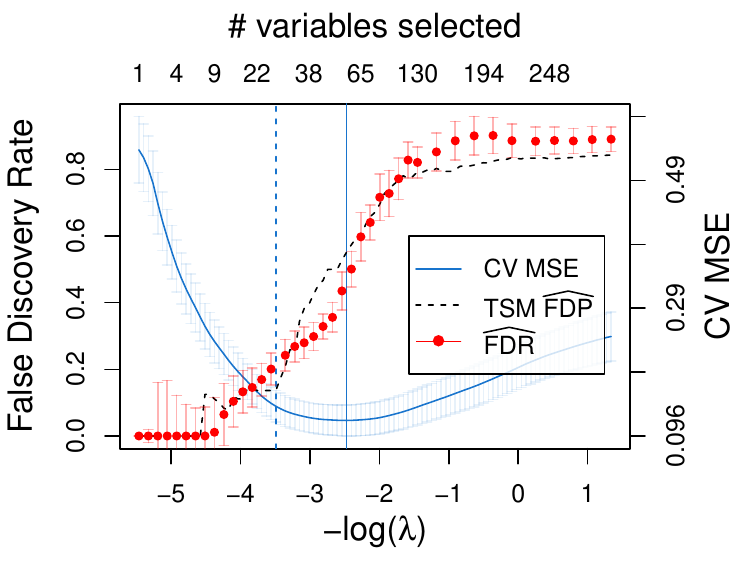}
      \caption{D4T, $281$, $625$}
    \end{subfigure} 
    \begin{subfigure}[b]{0.24\linewidth}
      \centering
      \includegraphics[width=\linewidth]{figs/HIV_drug_expr/2_5-DDI.pdf}
      \caption{DDI, $283$, $628$}
    \end{subfigure} 
    
    \begin{subfigure}[b]{0.24\linewidth}
      \centering
      \includegraphics[width=\linewidth]{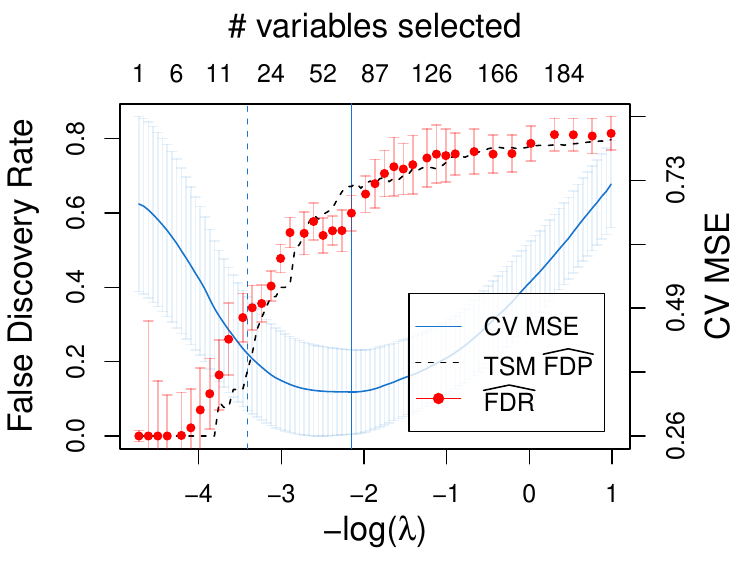}
      \caption{TDF, $215$, $351$}
    \end{subfigure}    
    \begin{subfigure}[b]{0.24\linewidth}
      \centering
      \includegraphics[width=\linewidth]{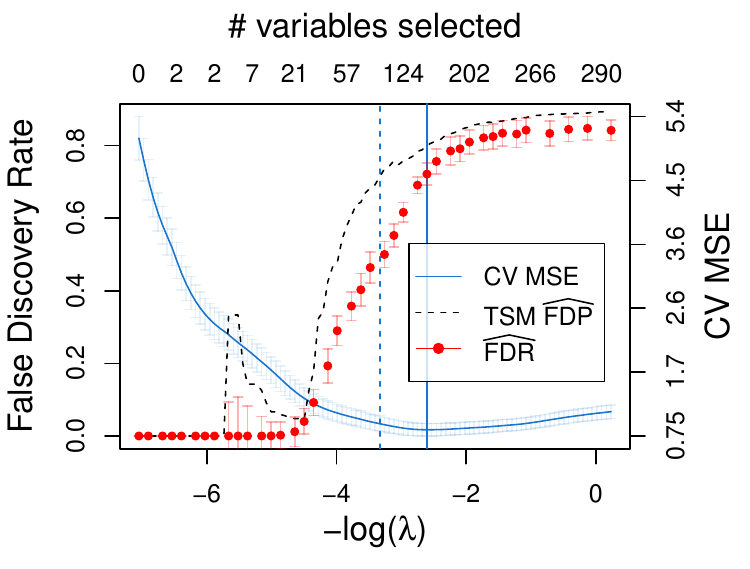}
      \caption{DLV, $305$, $730$}
    \end{subfigure} 
    \begin{subfigure}[b]{0.24\linewidth}
      \centering
      \includegraphics[width=\linewidth]{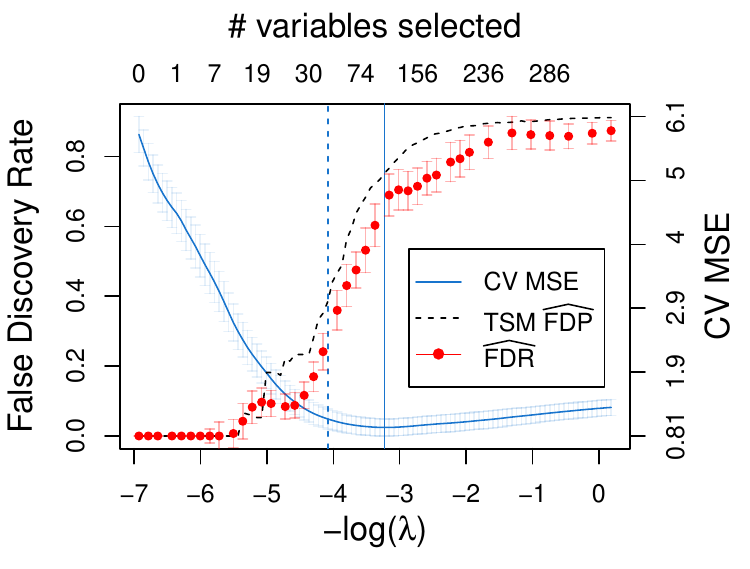}
      \caption{EFV, $312$, $732$}
    \end{subfigure} 
    \begin{subfigure}[b]{0.24\linewidth}
      \centering
      \includegraphics[width=\linewidth]{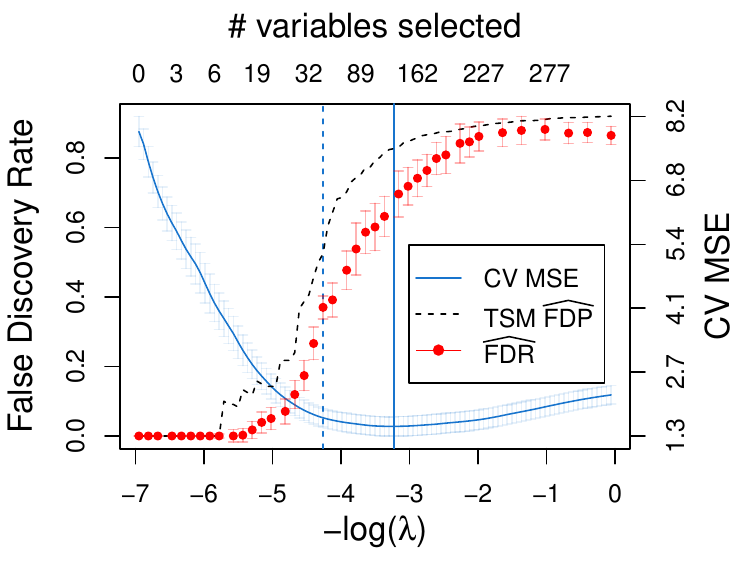}
      \caption{NVP, $313$, $744$}
    \end{subfigure} 
    
    \caption{Extended results of $\hFDR$ along with estimated one-standard-error bars in Section~\ref{sec:real_HIV}. The string and number in the caption of the subfigures represent the name of the drug, the number of variables $d$, and the number of observations $n$, respectively.}
    \label{fig:hiv_all}
\end{figure}

\begin{figure}[tbp]
    \centering
    \begin{subfigure}[b]{0.24\linewidth}
      \centering
      \includegraphics[width=\linewidth]{figs/protein/expr_1.pdf}
    \end{subfigure}
    \begin{subfigure}[b]{0.24\linewidth}
      \centering
      \includegraphics[width=\linewidth]{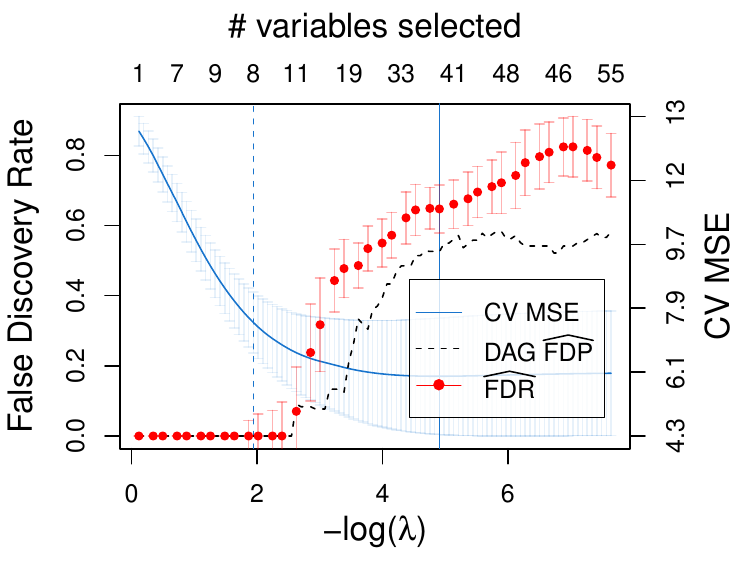}
    \end{subfigure}
    \begin{subfigure}[b]{0.24\linewidth}
      \centering
      \includegraphics[width=\linewidth]{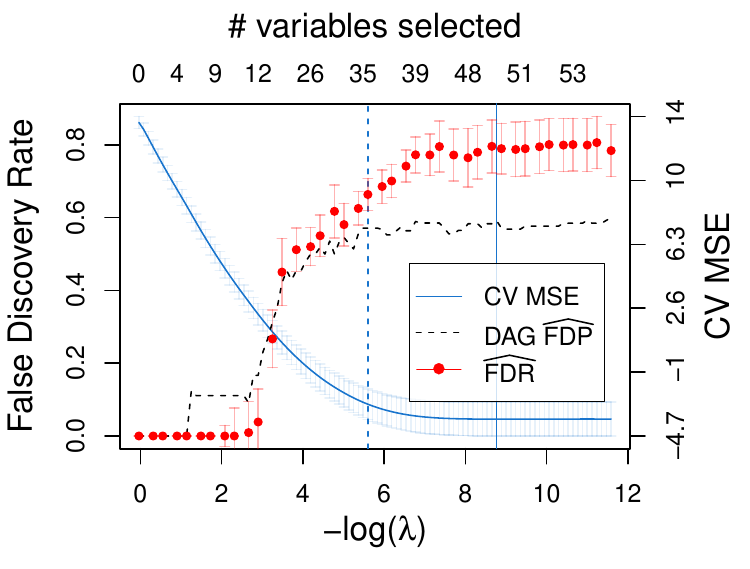}
    \end{subfigure}
    \begin{subfigure}[b]{0.24\linewidth}
      \centering
      \includegraphics[width=\linewidth]{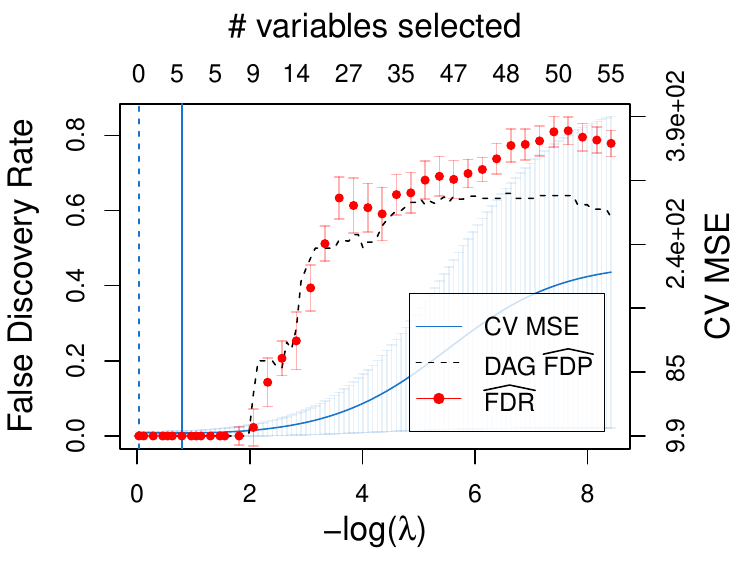}
    \end{subfigure}
    
    \begin{subfigure}[b]{0.24\linewidth}
      \centering
      \includegraphics[width=\linewidth]{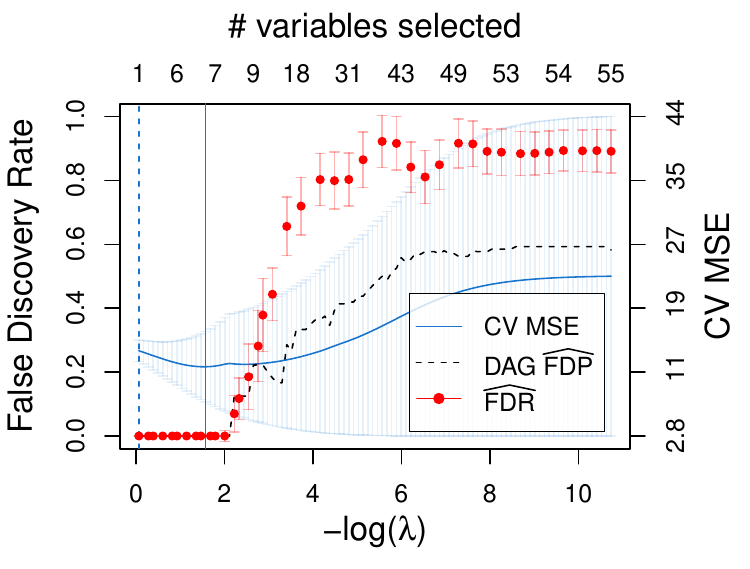}
    \end{subfigure}
    \begin{subfigure}[b]{0.24\linewidth}
      \centering
      \includegraphics[width=\linewidth]{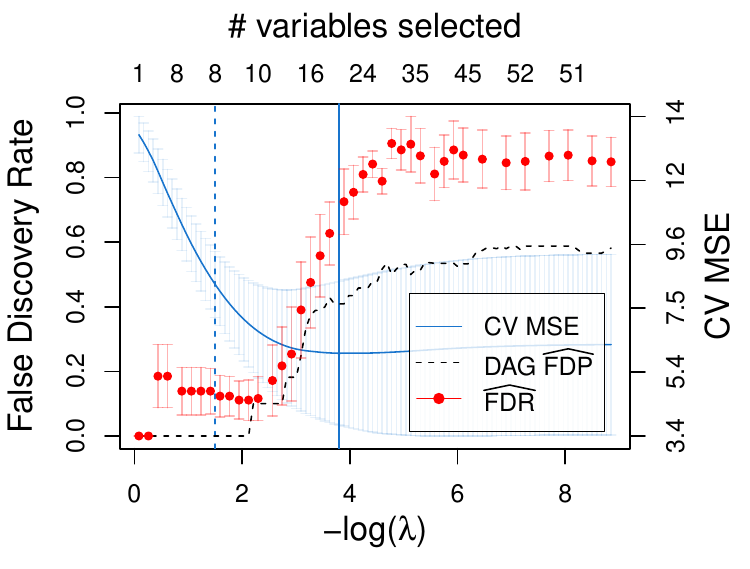}
    \end{subfigure}
    \begin{subfigure}[b]{0.24\linewidth}
      \centering
      \includegraphics[width=\linewidth]{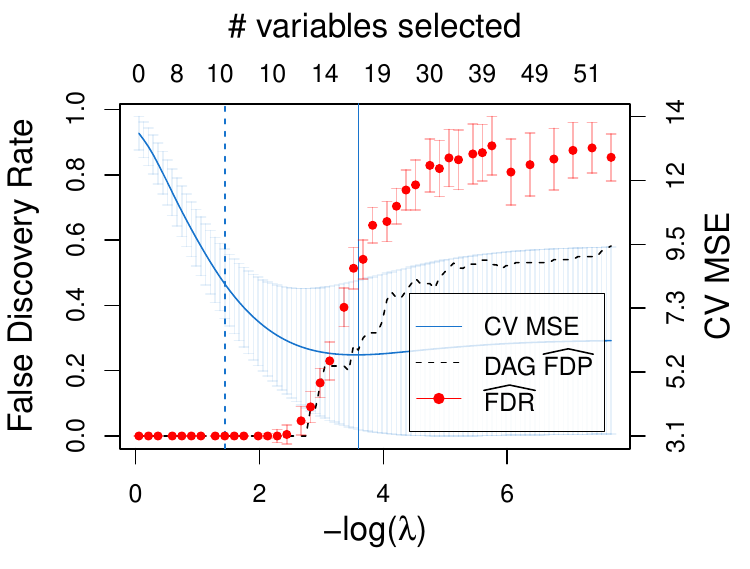}
    \end{subfigure}
    \begin{subfigure}[b]{0.24\linewidth}
      \centering
      \includegraphics[width=\linewidth]{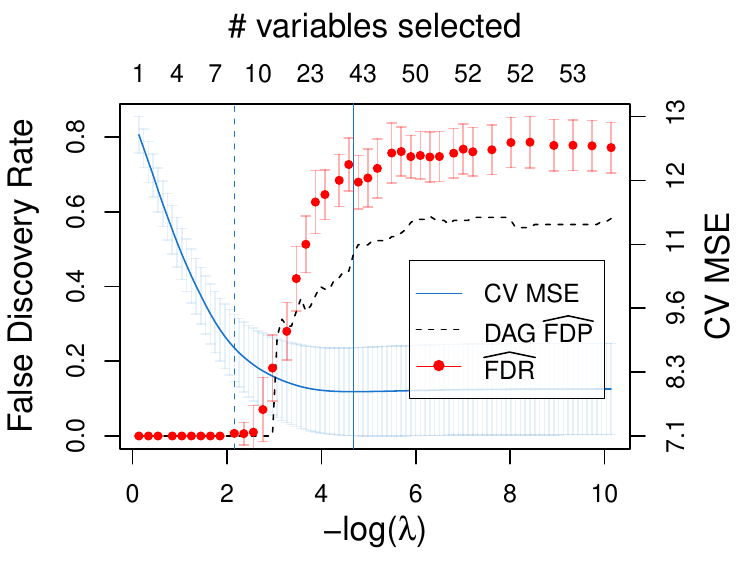}
    \end{subfigure}
    
    \begin{subfigure}[b]{0.24\linewidth}
      \centering
      \includegraphics[width=\linewidth]{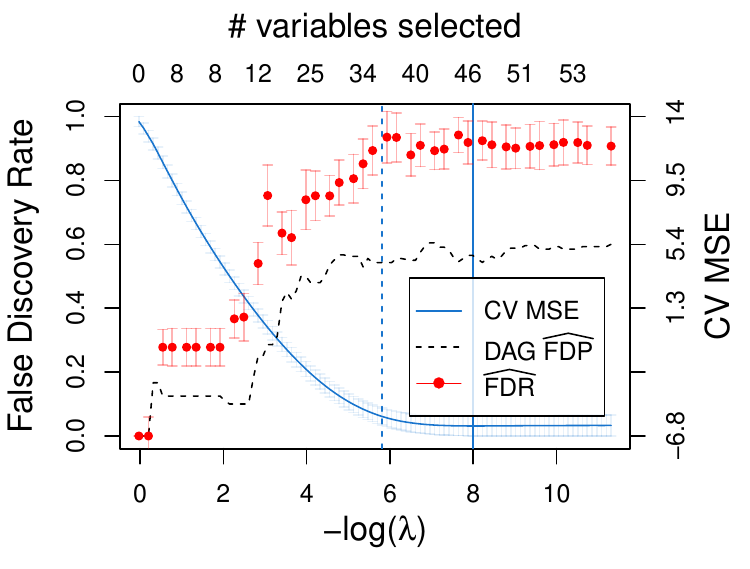}
    \end{subfigure}
    \begin{subfigure}[b]{0.24\linewidth}
      \centering
      \includegraphics[width=\linewidth]{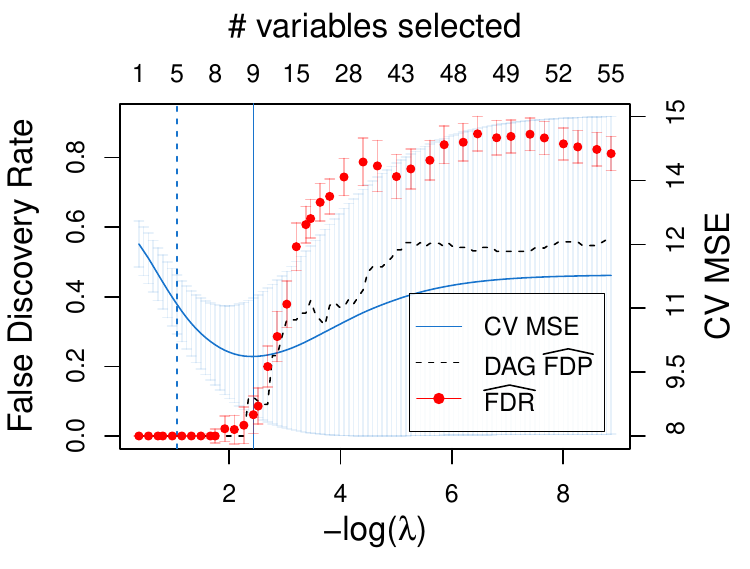}
    \end{subfigure}
    \begin{subfigure}[b]{0.24\linewidth}
      \centering
      \includegraphics[width=\linewidth]{figs/protein/expr_11.pdf}
    \end{subfigure}
    \begin{subfigure}[b]{0.24\linewidth}
      \centering
      \includegraphics[width=\linewidth]{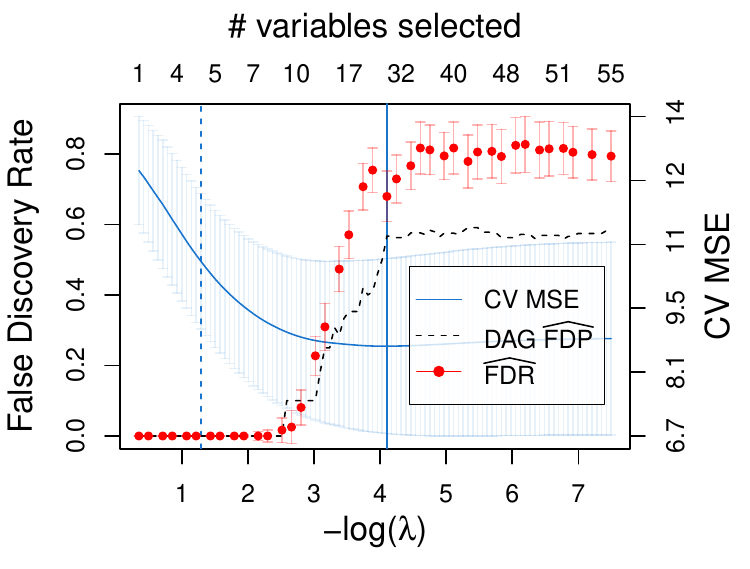}
    \end{subfigure}
    \begin{subfigure}[b]{0.24\linewidth}
      \centering
      \includegraphics[width=\linewidth]{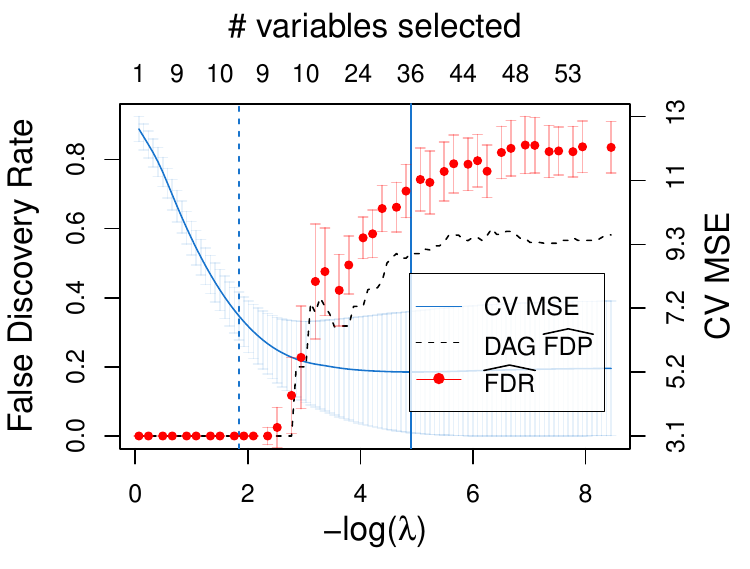}
    \end{subfigure}
    \begin{subfigure}[b]{0.24\linewidth}
      \centering
      \includegraphics[width=\linewidth]{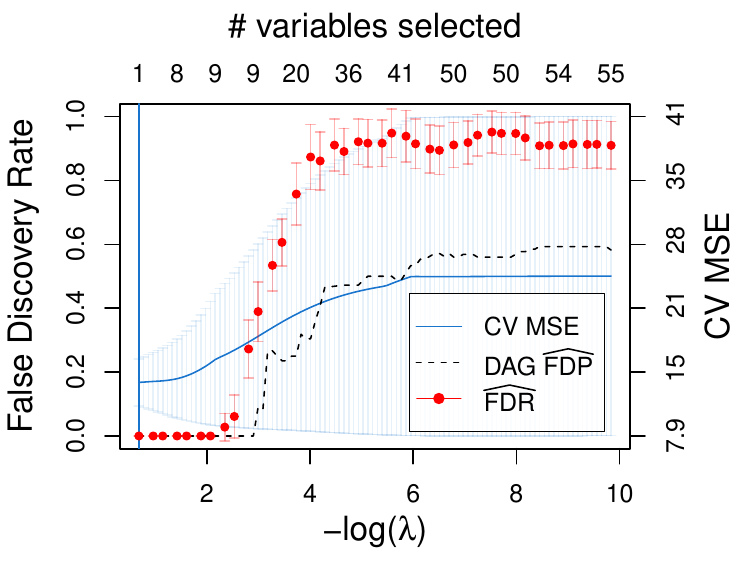}
    \end{subfigure}
    
    \caption{Extended results of $\hFDR$ along with estimated one-standard-error bars in Section~\ref{sec:real_protein}. From top-left to bottom-right are stimulation settings 1 to 14.}
    \label{fig:protein_all}
\end{figure}

\newpage
\section{Running time of $\hFDR$ in simulation problems} \label{app:runtime}

Figure \ref{fig:runtime} shows a box plot of the running time of evaluating $\hFDR$ (without estimating its s.e. by bootstrap) at all $10$ different $\lambda$ in the the simulation problems in Section~\ref{sec:simulations} on a single-cored computer. The runtime is typically several minutes.

\begin{figure}[H]
    \centering
    \includegraphics[width = 0.5\linewidth]{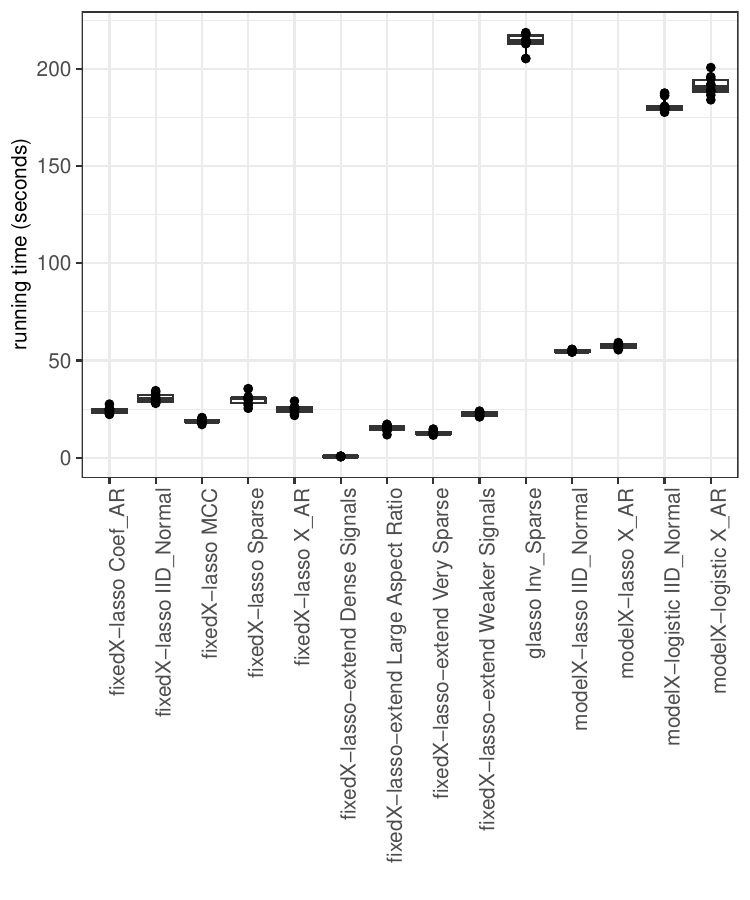}
    \caption{Box plot of the running time of evaluating $\hFDR$ (without estimating its s.e. by bootstrap) at all $10$ different $\lambda$ in the simulation problems in Section~\ref{sec:simulations} on a single-cored computer.}
    \label{fig:runtime}
\end{figure}

\end{document}